\RequirePackage{fix-cm}

\documentclass[twoside,openright,titlepage,numbers=noenddot,headinclude,
               footinclude=true,cleardoublepage=empty,abstractoff, 
               BCOR=5mm,paper=a4,fontsize=11pt,
               british,titlepage=false,pdftex, hyperfootnotes = false, pdfpagelabels
               ]{scrartcl}

\usepackage[pdfspacing, floatperchapter, linedheaders, subfig, nochapters]{classicthesis}

\usepackage[utf8]{inputenc}
\usepackage[T1]{fontenc}
\usepackage[british]{babel}
\usepackage[english = british]{csquotes} 

\usepackage{textcomp} 
\usepackage{scrhack} 
\usepackage{xspace} 
\usepackage{mparhack} 
\usepackage{fixltx2e} 
\usepackage[latest]{latexrelease}
\usepackage[printonlyused, smaller]{acronym} 
    
\usepackage{ragged2e} 

\usepackage{ifoddpage}

\usepackage[fleqn]{amsmath}
\usepackage[warning, all]{onlyamsmath}\AtBeginDocument{\catcode`\$=3} 
\usepackage{amssymb}
\usepackage{mathtools}
\usepackage{mathdots} 
\usepackage{gensymb} 
\usepackage{dsfont}
\usepackage{stmaryrd}
\usepackage{tensor} 

\allowdisplaybreaks

\usepackage{amsthm}
\usepackage{thmtools}

\usepackage{tabularx} 
\usepackage{caption}
\usepackage{subfig}

\PassOptionsToPackage{pdftex}{graphicx}
\PassOptionsToPackage{svgnames}{xcolor}
\usepackage{standalone}

\usepackage[pdftex]{graphicx}
    \graphicspath{{./graphics/}}
\usepackage{trimclip, adjustbox} 

\usepackage[svgnames]{xcolor} 
\usepackage{tikz}
\usetikzlibrary{positioning, patterns, arrows, calc}
\usepackage{tikzscale}

\usepackage{calc}
\usepackage[strict]{changepage} 

\usepackage{scalerel} 

\usepackage{hyperref} 
\hypersetup{%
    colorlinks = true, linktocpage = true, pdfstartpage = 3, pdfstartview = FitV,%
    breaklinks = true, pdfpagemode = UseNone, pageanchor = true, pdfpagemode = UseOutlines,%
    plainpages = false, bookmarksnumbered, bookmarksopen = true, bookmarksopenlevel = 1,%
    hypertexnames = true, pdfhighlight = /O,
    urlcolor = webbrown, linkcolor = RoyalBlue, citecolor = webgreen, 
    pdftitle = {},%
    pdfauthor = {},%
    pdfsubject = {},%
    pdfkeywords = {},%
    pdfcreator = {pdfLaTeX},%
    pdfproducer = {LaTeX with hyperref and classicthesis}%
}
\usepackage[noabbrev]{cleveref} 

\usepackage[
  backend = biber,
  style = alphabetic, 
  sortlocale = en_GB,
  hyperref = true,
  backref = true,
  abbreviate = false,
  url = true,
  doi = true,
  eprint = true,
  ibidtracker = false, 
  citetracker = false, 
  autocite = footnote
]{biblatex}
\DeclareFieldFormat*{citetitle}{#1} 

\usepackage{snotez}
\setsidenotes{
  note-mark-format = #1.,
  text-format = \slshape\footnotesize,
  perpage = true,
  footnote = true
}
\makeatletter 
\def\snotez@text#1#2{%
      \snotez@marginpar{%
        \checkoddpage
        \ifoddpage
          \RaggedRight
          \snotez@format%
          \snotez@write@mark{\snotez@note@mark@format{\@the@snotez@mark}}%
          \snotez@note@mark@sep#2%
        \else
          \RaggedLeft
          \snotez@format%
          \snotez@write@mark{\snotez@note@mark@format{\@the@snotez@mark}}%
          \snotez@note@mark@sep#2%
        \fi%
      }
}
\makeatother

\usepackage{cfr-lm} 

\declaretheoremstyle[
  headfont = \scshape, headpunct = ., postheadspace = 0.5em,
  bodyfont = \itshape,
  spaceabove = \baselineskip, spacebelow = \baselineskip,
  qed = \ensuremath{_\square} 
]{theorem}

\declaretheorem[numberwithin=section, style=theorem, refname={main theorem, main theorems}, Refname={Main Theorem, Main Theorems}, name=Main Theorem]{main-theorem}
\declaretheorem[sibling=main-theorem, style=theorem, refname={theorem,theorems}, Refname={Theorem,Theorems}]{theorem}
\declaretheorem[sibling=main-theorem, style=theorem, refname={lemma,lemmata}, Refname={Lemma,Lemmata}]{lemma}
\declaretheorem[sibling=main-theorem, style=theorem, refname={corollary,corollaries}, Refname={Corollary,Corollaries}]{corollary}

\declaretheoremstyle[
  headfont = \scshape, headpunct = ., postheadspace = 0.5em,
  bodyfont = \upshape,
  spaceabove = \baselineskip, spacebelow = \baselineskip,
  qed = \ensuremath{_\square} 
]{definition}

\declaretheorem[sibling=main-theorem, style=definition, refname={definition,definitions}, Refname={Definition,Definitions}]{definition}

\declaretheorem[sibling=main-theorem, style=definition, refname={remark,remarks}, Refname={Remark,Remarks}]{remark}
\declaretheorem[sibling=main-theorem, style=definition, refname={open problem,open problems}, Refname={Open Problem,Open Problems}, name=Open Problem]{open-problem}

\declaretheoremstyle[
  headfont = \scshape, headpunct = ., postheadspace = 0.5em,
  bodyfont = \normalfont,
  spaceabove = \baselineskip, spacebelow = \baselineskip,
  qed = \ensuremath{_\blacksquare}
]{proof}

\declaretheorem[unnumbered, style=proof, refname={proof,proofs}, Refname={Proof,Proofs}]{proof}
\declaretheorem[unnumbered, style=proof, refname={proof sektch,proof sketches}, Refname={Proof Sketch,Proof Sketches}, name=Proof Sketch]{proof-sketch}
\declaretheorem[unnumbered, style=proof, refname={usage note,usage notes}, Refname={Usage Note,Usage Notes}, name=Usage Note]{usage-note}

\makeatletter
\newcommand{\proofPart}[1]{%
  \par%
  \addvspace{\medskipamount}%
  \noindent\emph{#1.}
  \@afterheading%
}
\makeatother

\newcommand*{\define}[1]{\emph{#1}}
\newcommand*{\defineX}[2]{\emph{#1}}

\newcommand{\mathnote}[1]{
  \checkoddpage%
  \ifoddpage%
    \tag*{\rlap{\hspace{\marginparsep}\smash{\parbox[t]{\marginparwidth}{\itshape\footnotesize\leavevmode\color{Black}\raggedright\hspace{0pt}{#1}}}}}%
  \else%
    \tag*{\llap{\smash{\parbox[t]{\marginparwidth}{\itshape\footnotesize\leavevmode\color{Black}\raggedleft\hspace{0pt}{#1}}}\hspace{\marginparsep}\hspace{\textwidth}}}%
  \fi%
}

\providecommand\ifAndOnlyIf{\DOTSB\;\Longleftrightarrow\;} 

\DeclarePairedDelimiter\parens{\lparen}{\rparen} 

\DeclarePairedDelimiter\cardinalityOf{\lvert}{\rvert} 
\DeclarePairedDelimiter\lengthOf{\lvert}{\rvert} 
\DeclarePairedDelimiter\lengthOfPath{\lvert}{\rvert} 

\DeclarePairedDelimiter\normOf{\lVert}{\rVert} 
\DeclarePairedDelimiter\radiusOf{\lVert}{\rVert} 
\DeclarePairedDelimiterX\innerProductOf[2]{\langle}{\rangle}{#1,#2} 

\DeclarePairedDelimiter\setOf{\{}{\}} 
\newcommand*\suchThat{\mid} 

\DeclarePairedDelimiter\family{\{}{\}} 
\DeclarePairedDelimiter\sequence{\lparen}{\rparen} 

\DeclarePairedDelimiter\equivalenceClassOf{\lbrack}{\rbrack} 

\DeclarePairedDelimiter\closedInterval{\lbrack}{\rbrack}
\DeclarePairedDelimiter\openInterval{\rbrack}{\lbrack}
\DeclarePairedDelimiter\leftOpenAndRightClosedInterval{\rbrack}{\rbrack}
\DeclarePairedDelimiter\leftClosedAndRightOpenInterval{\lbrack}{\lbrack}

\mathchardef\breakingcomma\mathcode`\, 
{\catcode`,=\active
  \gdef,{\breakingcomma\discretionary{}{}{}}
}
\DeclareRobustCommand\ntuple[1]{\lparen\mathcode`\,=\string"8000 #1\rparen} 

\newcommand*{\N}{\mathbb{N}} 
\newcommand*{\Z}{{\mathbb{Z}}} 
\newcommand*{\R}{\mathbb{R}} 

\DeclareMathOperator{\identityMap}{id}

\DeclareMathOperator{\Exists}{\exists} 
\DeclareMathOperator{\notExists}{\nexists} 
\DeclareMathOperator{\ForEach}{\forall} 
\newcommand*\Holds{:} 
\newcommand*\SuchThat{:} 

\newcommand*\actsOnPoint{\triangleright} 
\newcommand*\actsOnMap{\mathbin{\protect\scalerel*{\blacktriangleright}{\triangleright}}} 
\newcommand*\actsByItsCoordinateOn{\mathbin{\protect\scalerel*{\ooalign{$\trianglelefteqslant$\cr\hidewidth\raise.225ex\hbox{$\blacktriangleleft$}\cr}}{\triangleleft}}} 

\makeatletter
\providecommand*{\Dashv}{%
  \mathrel{%
    \mathpalette\@Dashv\vDash
  }%
}
\newcommand*{\@Dashv}[2]{%
  \reflectbox{$\m@th#1#2$}%
}
\makeatother

\newcommand{\VDash}{%
  \mathrel{
    \text{\clipbox{0pt 0pt {.8\width} 0pt}{$\Vdash$}}
    \mkern.9mu
    \text{\adjustbox{width=.87\width,height=\height}{$\vDash$}}
  }
}

\makeatletter
\providecommand*{\DashV}{%
  \mathrel{%
    \mathpalette\@DashV\VDash
  }%
}
\newcommand*{\@DashV}[2]{%
  \reflectbox{$\m@th#1#2$}%
}
\makeatother
\makeatletter
\providecommand*{\dashV}{%
  \mathrel{%
    \mathpalette\@dashV\Vdash
  }%
}
\newcommand*{\@dashV}[2]{%
  \reflectbox{$\m@th#1#2$}%
}
\makeatother

\newcommand*\modulo{\slash}

\newcommand*\from{\colon} 
\newcommand*\after{\circ} 
\DeclareMathOperator{\domainOf}{dom} 
\DeclareMathOperator{\imageOf}{im} 
\newcommand*{\blank}{\mathord{\_}} 
\newcommand*{\restrictedTo}{\mathord{\upharpoonright}} 

\DeclareMathOperator*{\supportOf}{supp} 

\DeclareMathOperator{\distanceOf}{d} 
\DeclareMathOperator{\degreeOf}{deg} 

\newcommand*\occursIn{\sqsubseteq} 
\newcommand*\semiOccursIn{\mathrel{\ooalign{$\occursIn$\cr\hidewidth\raise.225ex\hbox{$\circ\mkern.5mu$}\cr}}} 

\DeclareMathOperator{\powerSetOf}{\mathcal{P}} 
\newcommand{\closureOf}[1]{\mkern 1.5mu\overline{\mkern-1.5mu#1\mkern-1.5mu}\mkern 1.5mu} 

\makeatletter
\def\moverlay{\mathpalette\mov@rlay}
\def\mov@rlay#1#2{\leavevmode\vtop{%
   \baselineskip\z@skip \lineskiplimit-\maxdimen
   \ialign{\hfil$\m@th#1##$\hfil\cr#2\crcr}}}
\newcommand{\charfusion}[3][\mathord]{
    #1{\ifx#1\mathop\vphantom{#2}\fi
        \mathpalette\mov@rlay{#2\cr#3}
      }
    \ifx#1\mathop\expandafter\displaylimits\fi}
\makeatother

\newcommand{\bigDisjointUnionOf}{\charfusion[\mathop]{\bigcup}{\cdot}} 

\newcommand{\myfloatalign}{\centering} 

\newenvironment{wide}{%
  \begin{adjustwidth*}{0pt}{-\marginparsep-\marginparwidth}
    \myfloatalign
}{%
  \end{adjustwidth*}%
}

\newcommand*{\Graph}{\mathcal{G}}

\newcommand*{\Tree}{\mathcal{T}}
\newcommand*{\continuumGraph}{\mathfrak{G}}
\DeclareMathOperator{\Signals}{Sgnl}
\DeclareMathOperator{\Kinds}{Knd}
\DeclareMathOperator{\kindOf}{knd}
\DeclareMathOperator{\Data}{Dt}
\DeclareMathOperator{\NoData}{None}
\DeclareMathOperator{\Vertices}{V}
\newcommand*{\continuumVertices}{\mathfrak{V}}
\DeclareMathOperator{\Edges}{E}
\newcommand*{\continuumEdges}{\mathfrak{E}}
\DeclareMathOperator{\Paths}{P}
\newcommand*{\continuumPaths}{\mathfrak{P}}
\newcommand*{\directionPreserving}{\shortrightarrow}
\newcommand*{\directed}{\rightarrow}
\DeclareMathOperator{\Directions}{Dir}
\DeclareMathOperator{\Arrows}{Arr}
\DeclareMathOperator{\every}{vry}
\DeclareMathOperator{\Times}{T}
\newcommand*{\extendedTimes}{\closureOf{\Times}}
\DeclareMathOperator{\States}{Q}
\DeclareMathOperator{\Configurations}{Cnf}
\DeclareMathOperator{\ReachOf}{R}
\DeclareMathOperator{\directionOf}{dir}
\newcommand*{\start}{\sourceOf}
\newcommand*{\xend}{\targetOf}
\DeclareMathOperator{\speedOf}{spd}
\DeclareMathOperator{\velocityOf}{vel}
\DeclareMathOperator{\datumOf}{dt}
\DeclareMathOperator{\signOf}{sgn}
\DeclareMathOperator{\subtreesOf}{sbtrs}
\DeclareMathOperator{\maxSubtreesOf}{max\,sbtrs}
\newcommand*{\secondMaxSubtreesOf}{\closureOf{\subtreesOf}}
\DeclareMathOperator{\maxPaths}{max\,pths}
\DeclareMathOperator{\maxVertices}{max\,vrtcs}
\DeclareMathOperator{\maxTree}{\hat{\Tree}} 

\newcommand*{\localTransitionFunction}{\delta}
\newcommand*{\globalTransitionFunction}{\mathbin{\boxdot}} 

\makeatletter
\newcommand{\dotDelta}{{\vphantom{\Delta}\mathpalette\d@tD@lta\relax}}
\newcommand{\d@tD@lta}[2]{%
  \ooalign{\hidewidth$\m@th#1\mkern-1mu\cdot$\hidewidth\cr$\m@th#1\Delta$\cr}%
}
\makeatother

\newcommand*{\downTo}{\downarrow} 

\newcommand*{\gtfJump}{\mathbin{\dotDelta}} 
\newcommand*{\gtfIgnoreCollisions}{\mathbin{\boxplus}} 
\newcommand*{\gtfNonNegativeSingularities}{\mathbin{\boxminus}} 
\newcommand*{\gtfNegativeSingularityOfOrderMinusOne}{\mathbin{\boxast}} 
\newcommand*{\concat}{\bullet}
\newcommand*{\reverse}{-}
\newcommand*{\multipli}{\cdot}
\newcommand{\continuumRepresentationOf}[1]{\mkern 1.5mu\overline{\mkern-1.5mu#1\mkern-1.5mu}\mkern 1.5mu} 

\makeatletter
\DeclareRobustCommand{\cev}[1]{%
  \mathpalette\do@cev{#1}%
}
\newcommand{\do@cev}[2]{%
  \fix@cev{#1}{+}%
  \reflectbox{$\m@th#1\vec{\reflectbox{$\fix@cev{#1}{-}\m@th#1#2\fix@cev{#1}{+}$}}$}%
  \fix@cev{#1}{-}%
}
\newcommand{\fix@cev}[2]{%
  \ifx#1\displaystyle
    \mkern#23mu
  \else
    \ifx#1\textstyle
      \mkern#23mu
    \else
      \ifx#1\scriptstyle
        \mkern#22mu
      \else
        \mkern#22mu
      \fi
    \fi
  \fi
}
\makeatother

\newcommand{\character}[1]{\mathtt{#1}}
\newcommand*{\initiateKind}{\character{I}}
\newcommand*{\leafKind}{\character{L}}
\newcommand*{\countKind}{\character{C}}
\newcommand*{\midpointKind}{\character{M}}
\newcommand*{\findMidpointKind}{\character{U}}
\newcommand*{\reflectedFindMidpointKind}{\cev{\findMidpointKind}}
\newcommand*{\slowedDownFindMidpointKind}{\character{V}}
\newcommand*{\freezeKind}{\character{F}}
\newcommand*{\thawKind}{\character{T}}
\newcommand*{\divideKind}{\character{D}}
\newcommand*{\frozenDivideKind}{\freezeKind\divideKind}
\newcommand*{\reflectedDivideKind}{\cev{\divideKind}}
\newcommand*{\boundaryKind}{\character{B}}
\newcommand*{\fireKind}{\character{X}}
\newcommand*{\frozenFireKind}{\freezeKind\fireKind}

\newcommand*{\initiateSignal}[1]{\initiateKind_{#1}}
\newcommand*{\leafSignal}[1]{\leafKind_{#1}}
\newcommand*{\midpointSignal}[2]{\midpointKind_{\setOf{#1, #2}}}
\newcommand*{\countSignal}[1]{\countKind_{#1}}
\newcommand*{\findMidpointSignal}[2]{\findMidpointKind_{#1, #2}} 
\newcommand*{\reflectedFindMidpointSignal}[4]{\reflectedFindMidpointKind_{#1, #2, #3}^{#4}} 
\newcommand*{\slowedDownFindMidpointSignal}[4]{\slowedDownFindMidpointKind_{#1, #2, #3}^{#4}} 
\newcommand*{\freezeSignal}[1]{\freezeKind_{#1}} 
\newcommand*{\thawSignal}[3]{\thawKind_{#1, #2}^{#3}} 
\newcommand*{\divideSignal}[2]{\divideKind_{#1, #2}} 
\newcommand*{\frozenDivideSignal}[2]{\frozenDivideKind_{#1, #2}} 
\newcommand*{\reflectedDivideSignal}[1]{\reflectedDivideKind_{#1}}
\newcommand*{\boundarySignal}{\boundaryKind}
\newcommand*{\fireSignal}{\fireKind}
\newcommand*{\frozenFireSignal}{\frozenFireKind}

\newcommand*{\midpoint}{\mathfrak{m}}
\newcommand*{\general}{\mathfrak{g}}

\newcommand*{\sourceOf}{\sigma}
\newcommand*{\targetOf}{\tau}
\newcommand*{\bedOf}{\beta}
\newcommand*{\eendsOf}{\varepsilon}
\DeclareMathOperator{\invert}{inv}

\newcommand*{\weightOf}{\omega}
\newcommand*{\length}{\weightOf} 

\newcommand*{\direct}[1]{\vec{#1}}

\newcommand*{\emptyWord}{\lambda}

\newcommand*{\booleans}{\mathbb{B}}
\newcommand*{\no}{\character{no}}
\newcommand*{\yes}{\character{yes}}

\DeclareMathOperator{\tree}{tree}
\DeclareMathOperator{\virtualTree}{virtualTree}
\DeclareMathOperator{\boundaryCases}{bndryCases}

\DeclareMathOperator{\areEmpty}{areEmpty}
\DeclareMathOperator{\isInLeaf}{inLeaf}
\DeclareMathOperator{\isPenultimate}{penultimate}

\DeclareMathOperator*{\argmax}{arg\,max}

\newcommand*{\discreteInterval}[2]{[#1 : #2]}

\newcommand*{\deadState}{\mathtt{\#}}
\newcommand*{\generalState}{\mathtt{g}}
\newcommand*{\soldierState}{\mathtt{s}}
\newcommand*{\fireState}{\mathtt{f}}

\renewcommand*{\index}[2][]{}

\definecolor{myred}{HTML}{e41a1c}
\definecolor{myblue}{HTML}{377eb8}
\definecolor{mygreen}{HTML}{4daf4a}
\definecolor{myviolet}{HTML}{984ea3}

\tikzset{%
  edge/.style = {draw = black},
  vertex/.style = {draw = black, thick},
  midpoint/.style = {draw = black, thick},
  boundary/.style = {draw = black, thick},
  find/.style = {draw = black, dotted, very thick}, 
  reflected/.style = {draw = black, densely dotted, thick},
  syncLeft/.style = {pattern = horizontal lines, pattern color = black, draw = none}, 
  syncMiddle/.style = {pattern = fivepointed stars, pattern color = black, draw = none},
  syncRight/.style = {pattern = dots, pattern color = black, draw = none}, 
  slowed/.style = {draw = black}, 
  freeze/.style = {draw = black, dashed}, 
  traversingThaw/.style = {draw = black, loosely dotted, thick},
  thawingThaw/.style = {draw = black, dashed},
  divide0/.style = {draw = black, solid},
  divideN/.style = {draw = black, densely dotted},
  reflectedDivide/.style = {draw = black},
  leftColour/.style = {draw = mygreen, pattern color = mygreen},
  leftFill/.style = {color = mygreen},
  rightColour/.style = {draw = myblue, pattern color = myblue},
  rightFill/.style = {color = myblue},
  overlayLeftColour/.style = {draw = myred, pattern color = myred},
  overlayLeftFill/.style = {color = myred},
  overlayRightColour/.style = {draw = myviolet, pattern color = myviolet},
  overlayRightFill/.style = {color = myviolet}
}

\def\figureSingularities{%
  \subfloat[A singularity of order $1$.]{%
    \resizebox{(\linewidth-1em)/2}{!}{%
      \begin{tikzpicture}[> = To, x = 1cm, y = {-1cm}]
        \pgfmathsetmacro{\l}{3} 
        \pgfmathsetmacro{\h}{\l + \l} 

        \coordinate (gg) at (0, 0); 
        \coordinate (dr) at (\l, \l); 
        \coordinate (rdl) at (0, 2 * \l); 

        \draw[divide0] (gg) -- (dr);
        \draw[reflectedDivide] (dr) -- (rdl);

        \foreach \type in {1, 2, ..., 15} {
          \pgfmathsetmacro{\v}{((2/3)^\type) * \l};
          \draw[divideN] (gg) -- (\v, {2 * \l - \v});
        }

        \draw[vertex] (0, 0) -- (0, 6)
                      (3, 0) -- (3, 6);
        \draw[edge] (0, 0) -- (3, 0)
                    (0, 6) -- (3, 6);
      \end{tikzpicture}
    }
    \label{figure:singularity-of-order-one}
  }
  \hfill
  \subfloat[A singularity of order $-1$.]{%
    \resizebox{(\linewidth-1em)/2}{!}{%
      \begin{tikzpicture}[> = To, x = 1cm, y = {-1cm}]
        \pgfmathsetmacro{\l}{3} 
        \pgfmathsetmacro{\h}{\l + \l} 

        \coordinate (gg) at (0, 0); 
        \coordinate (dr) at (\l, \l); 
        \coordinate (rdl) at (0, 2 * \l); 

        \draw[divide0] (0, 0) -- (-1, 1);
        \draw[reflectedDivide] (-1, 1) -- (2, 4);

        \begin{scope}
          \clip (0, 0) -- (0, 2) -- (2, 4) -- (2, 0) -- cycle;
          \foreach \type in {1, 2, ..., 15} {
            \pgfmathsetmacro{\v}{((2/3)^\type) * \l};
            \draw[divideN] (gg) -- (\v, {2 * \l - \v});
          }
        \end{scope}

        \draw[vertex] (-1, 0) -- (-1, 4)
                      (0, 0) -- (0, 4)
                      (2, 0) -- (2, 4);
        \draw[edge] (-1, 0) -- (2, 0)
                    (-1, 4) -- (2, 4);
      \end{tikzpicture}
    }
    \label{figure:singularity-of-order-minus-one}
  }
}

\def\figureMidpointOneEdge{%
  \begin{tikzpicture}[> = To, x = 1cm, y = {-1cm}]
    \pgfmathsetmacro{\l}{3} 
    \pgfmathsetmacro{\h}{2 * \l} 

    \draw[find] (0, 0) -- (\l, \l);
    \draw[reflected] (\l, \l) -- (0, \h);
    \draw[slowed] (0, 0) -- ({(1/3) * \h}, \h);
    \draw[midpoint] (0.5 * \l, 1.5 * \l) -- (0.5 * \l, \h);

    \draw[vertex] (0, 0) -- (0, \h); 
    \draw[vertex] (\l, 0) -- (\l, \h); 

    \draw[edge] (0, 0) -- (\l, 0); 
    \draw[edge] (0, \h) -- (\l, \h); 
  \end{tikzpicture}
}

\def\figureMidpointTwoEdges{%
  \begin{tikzpicture}[> = To, x = 1cm, y = {-1cm}] 
    \pgfmathsetmacro{\l}{3} 
    \pgfmathsetmacro{\r}{2} 
    \pgfmathsetmacro{\w}{\l + \r} 
    \pgfmathsetmacro{\h}{2 * \l + \r} 

    \begin{scope}[shift = {(\l, -1.5)}, scale = 0.5]
      \draw (0, 0) -- (-2.1213, 2.1213);
      \draw (0, 0) -- (1.414, 1.414);
    \end{scope}

    \coordinate (g) at (\l, 0); 

    \draw[find] (g) -- (0, \l);
    \draw[reflected] (0, \l) -- (\l, 2 * \l);
    \draw[slowed] (\l, 2 * \l) -- ({\l + (1/3) * \r}, \h);

    \draw[find] (g) -- (\w, \r);
    \draw[reflected] (\w, \r) -- (\l, 2 * \r);
    \draw[slowed] (\l, 2 * \r) -- ({0.5 * \w - (1/3) * 0.5 * \w}, \h);

    \draw[midpoint] (0.5 * \w, \h - 0.5 * \w) -- (0.5 * \w, \h);

    \draw[vertex] (0, 0) -- (0, \h);
    \draw[vertex] (g) -- (\l, \h);
    \draw[vertex] (\w, 0) -- (\w, \h);

    \draw[edge] (0, 0) -- (g); 
    \draw[edge] (g) -- (\w, 0); 
    \draw[edge] (0, \h) -- (\l, \h); 
    \draw[edge] (\l, \h) -- (\w, \h); 
  \end{tikzpicture}
}

\def\mazoyerOneLevel(#1,#2){
  \pgfmathsetmacro{\l}{#1} 
  \pgfmathsetmacro{\h}{\l + \l} 

  \coordinate (gg) at (0, 0); 
  \coordinate (dr) at (\l, \l); 
  \coordinate (rdl) at (0, 2 * \l); 

  \draw[divide0] (gg) -- (dr);
  \draw[reflectedDivide] (dr) -- (rdl);

  \foreach \type in {1, 2, ..., 15} {
    \pgfmathsetmacro{\v}{((2/3)^\type) * \l};
    \coordinate (b #2 \type) at (\v, {2 * \l - \v}); 
    \coordinate (be #2 \type) at (\v, \h); 
    \draw[divideN] (gg) -- (b #2 \type);
    \draw[boundary] (b #2 \type) -- (be #2 \type);
  }
}

\def\mazoyerTwoLevels(#1){%
  \pgfmathsetmacro{\le}{#1} 
  \mazoyerOneLevel(\le,0); 
  \foreach \type in {1, 2, ..., 15} {
    \begin{scope}[shift = (b 0 \type)]
      \pgfmathsetmacro{\typeMinusOne}{\type - 1};
      \mazoyerOneLevel({(1/3) * (2/3)^\typeMinusOne * \le},\type);
    \end{scope}
  }
}

\def\figureMazoyer{%
  \begin{tikzpicture}[> = To, x = 1cm, y = {-1cm}]
    \mazoyerTwoLevels(3);
    \draw[vertex] (0, 0) -- (0, 6)
                  (3, 0) -- (3, 6);
    \draw[edge] (0, 0) -- (3, 0)
                (0, 6) -- (3, 6);
  \end{tikzpicture}
}

\def\figureFSSPMazoyerWithTwoEdges{%
  \begin{tikzpicture}[> = To, x = 1cm, y = {-1cm}]
    \pgfmathsetmacro{\l}{3} 
    \pgfmathsetmacro{\r}{2} 
    \pgfmathsetmacro{\w}{\l + \r} 
    \pgfmathsetmacro{\h}{\l + \l + \r} 

    \begin{scope}[shift = {(\l, -1.5)}, scale = 0.5]
      \draw (0, 0) -- (-2.1213, 2.1213);
      \draw (0, 0) -- (1.414, 1.414);
    \end{scope}

    \coordinate (g) at (\l, 0); 
    \coordinate (lb) at (0, \l); 
    \coordinate (rb) at (\w, \r); 
    \coordinate (lg) at (\l, 2 * \l); 
    \coordinate (rg) at (\l, 2 * \r); 
    \coordinate (m) at (0.5 * \w, \l + 0.5 * \w); 
    \coordinate (mlf) at (0.5 * \l, \l + 0.5 * \l); 
    \coordinate (mlt) at (0.5 * \l, \h - 0.5 * \l); 
    \coordinate (mrf) at (\l + 0.5 * \r, \r + 0.5 * \r); 
    \coordinate (mrt) at (\l + 0.5 * \r, \h - 0.5 * \r); 
    \coordinate (lfb) at ($(mlf) + (- 0.5 * \l, 0.5 * \l)$); 
    \coordinate (rfb) at ($(mrf) + (0.5 * \r, 0.5 * \r)$); 
    \coordinate (lfg) at ($(mlf) + (0.5 * \l, 0.5 * \l)$); 
    \coordinate (rfg) at ($(mrf) + (- 0.5 * \r, 0.5 * \r)$); 

    \begin{scope} 
      \clip (0, 0) -- (0, 2 * \l) -- (0.5 * \l, 1.5 * \l) -- (\l, 2 * \l) -- (\l, 0) -- cycle;
      \begin{scope}[xscale = -1, shift = {(-\l, 0)}]
        \mazoyerTwoLevels(\l);
      \end{scope}
    \end{scope}
    \draw[freeze] (0.5 * \l, 1.5 * \l) -- (0, 2 * \l)
                  (0.5 * \l, 1.5 * \l) -- (\l, 2 * \l); 
    \begin{scope}[shift = {(0, \h - 2 * \l)}] 
      \begin{scope}
        \clip (0, 2 * \l) -- (0.5 * \l, 1.5 * \l) -- (\l, 2 * \l) -- cycle;
        \begin{scope}[xscale = -1, shift = {(-\l, 0)}]
          \mazoyerTwoLevels(\l);
        \end{scope}
      \end{scope}
      \draw[thawingThaw] (0.5 * \l, 1.5 * \l) -- (0, 2 * \l)
                  (0.5 * \l, 1.5 * \l) -- (\l, 2 * \l); 
    \end{scope}

    \begin{scope}[shift = {(\l, 0)}]
      \begin{scope} 
        \clip (0, 0) -- (0, 2 * \r) -- (0.5 * \r, 1.5 * \r) -- (\r, 2 * \r) -- (\r, 0) -- cycle;
        \mazoyerTwoLevels(\r);
      \end{scope}
      \draw[freeze] (0.5 * \r, 1.5 * \r) -- (0, 2 * \r)
                    (0.5 * \r, 1.5 * \r) -- (\r, 2 * \r); 
      \begin{scope}[shift = {(0, \h - 2 * \r)}] 
        \begin{scope} 
          \clip (0, 2 * \r) -- (0.5 * \r, 1.5 * \r) -- (\r, 2 * \r) -- cycle;
          \mazoyerTwoLevels(\r);
        \end{scope}
        \draw[thawingThaw] (0.5 * \r, 1.5 * \r) -- (0, 2 * \r)
                    (0.5 * \r, 1.5 * \r) -- (\r, 2 * \r); 
      \end{scope}
    \end{scope}

    \draw[slowed] (g) -- (mlf);
    \draw[slowed] (g) -- (mrf);
    \draw[slowed] (\l, \r + \r) -- (0.5 * \w, {\r + \r + 1.5 * (\l - \r)});

    \draw[white, semithick] (mlf) -- (lfg)
                            (mrf) -- (rfg);
    \draw[freeze] (mlf) -- (lfg)
                  (mlf) -- (lfb)
                  (mrf) -- (rfb)
                  (mrf) -- (rfg);

    \draw[traversingThaw] (m) -- (mlt)
                          (m) -- (mrt);

    \draw[white, semithick] (0, \h) -- (0.5 * \l, \h - 0.5 * \l) -- (\l, \h);
    \draw[white, semithick] (\w, \h) -- (\l + 0.5 * \r, \h - 0.5 * \r) -- (\w - \r, \h);
    \draw[thawingThaw] (0.5 * \l, \h - 0.5 * \l) -- (0, \h)
                       (0.5 * \l, \h - 0.5 * \l) -- (\l, \h);
    \draw[thawingThaw] (\l + 0.5 * \r, \h - 0.5 * \r) -- (\w, \h)
                       (\l + 0.5 * \r, \h - 0.5 * \r) -- (\w - \r, \h);

    \draw[midpoint] (mlf) -- (mlt)
                    (mrf) -- (mrt);

    \draw[vertex] (0, 0) -- (0, \h)
                  (g) -- (\l, \h)
                  (\w, 0) -- (\w, \h);

    \draw[edge] (0, 0) -- (g) -- (\w, 0)
                (0, \h) -- (\l, \h) -- (\w, \h);
  \end{tikzpicture}
}

\def\figureFSSPWithTwoEdgesAndGeneralInBetween{%
  \begin{tikzpicture}[> = To, x = 1cm, y = {-1cm}] 
    \pgfmathsetmacro{\l}{3} 
    \pgfmathsetmacro{\r}{2} 
    \pgfmathsetmacro{\o}{1.333} 
    \pgfmathsetmacro{\x}{max(\l, \r)} 
    \pgfmathsetmacro{\y}{min(\l, \r)} 
    \pgfmathsetmacro{\w}{\l + \r} 
    \pgfmathsetmacro{\h}{\x + \x + \y} 

    \begin{scope}[shift = {(\l, -1.5)}, scale = 0.5]
      \draw (0, 0) -- (-2.1213, 2.1213);
      \draw (0, 0) -- (1.414, 1.414);
    \end{scope}

    \coordinate (g) at (\l, 0); 
    \coordinate (lb) at (0, \l); 
    \coordinate (rb) at (\w, \r); 
    \coordinate (lg) at (\l, 2 * \l); 
    \coordinate (rg) at (\l, 2 * \r); 
    \coordinate (m) at (0.5 * \w, \x + 0.5 * \w); 
    \coordinate (mlf) at (0.5 * \l, \l + 0.5 * \l); 
    \coordinate (mlt) at (0.5 * \l, \h - 0.5 * \l); 
    \coordinate (mrf) at (\l + 0.5 * \r, \r + 0.5 * \r); 
    \coordinate (mrt) at (\l + 0.5 * \r, \h - 0.5 * \r); 
    \coordinate (lfb) at ($(mlf) + (- 0.5 * \l, 0.5 * \l)$); 
    \coordinate (rfb) at ($(mrf) + (0.5 * \r, 0.5 * \r)$); 
    \coordinate (lfg) at ($(mlf) + (0.5 * \l, 0.5 * \l)$); 
    \coordinate (rfg) at ($(mrf) + (- 0.5 * \r, 0.5 * \r)$); 

    \draw[syncLeft] (g) -- (lb) -- (lfb) -- (mlf) -- (lg) -- cycle;
    \draw[syncRight] (g) -- (rb) -- (rfb) -- (mrf) -- (rg) -- cycle;

    \draw[syncLeft, thawingThaw]
        (0, \h) -- (0.5 * \l, \h - 0.5 * \l)
                -- (\l, \h);
    \draw[syncRight, thawingThaw]
        (\w, \h) -- (\l + 0.5 * \r, \h - 0.5 * \r)
                 -- (\w - \r, \h);

    \draw[find] (g) -- (lb);
    \draw[reflected] (lb) -- (mlf);
    \draw[find] (g) -- (rb);
    \draw[reflected] (rb) -- (\x, \y + \y);
    \draw[slowed] (\x, \y + \y) --  (0.5 * \w, {\y + \y + 1.5 * (\x - \y)});

    \draw[slowed] (g) -- (mlf);
    \draw[slowed] (g) -- (mrf);

    \draw[midpoint] (mlf) -- (mlt)
                    (mrf) -- (mrt);

    \draw[traversingThaw] (m) -- (mlt)
                          (m) -- (mrt);

    \draw[freeze] (mlf) -- (lfg)
                  (mlf) -- (lfb)
                  (mrf) -- (rfb)
                  (mrf) -- (rfg);

    \draw[vertex] (0, 0) -- (0, \h)
                  (g) -- (\l, \h)
                  (\w, 0) -- (\w, \h);

    \draw[edge] (0, 0) -- (g) -- (\w, 0)
                (0, \h) -- (\l, \h) -- (\w, \h);
  \end{tikzpicture}
}

\def\figureFSSPWithTwoEdgesAndGeneralAtTheLeft{%
  \begin{tikzpicture}[> = To, x = 1cm, y = {-1cm}] 
    \pgfmathsetmacro{\l}{3} 
    \pgfmathsetmacro{\r}{2} 
    \pgfmathsetmacro{\o}{1.333} 
    \pgfmathsetmacro{\w}{\l + \r} 
    \pgfmathsetmacro{\h}{2 * \w} 

    \coordinate (g) at (0, 0); 
    \coordinate (v) at (\l, 0); 
    \coordinate (iv) at (\l, \l); 
    \coordinate (lb) at (\l, \l); 
    \coordinate (rb) at (\w, \l + \r); 
    \coordinate (lg) at (0, 2 * \l); 
    \coordinate (rv) at (\l, \l + 2 * \r); 
    \coordinate (m) at (0.5 * \w, \w + 0.5 * \w); 
    \coordinate (mlf) at (0.5 * \l, \l + 0.5 * \l); 
    \coordinate (mlt) at (0.5 * \l, \h - 0.5 * \l); 
    \coordinate (mrf) at (\l + 0.5 * \r, \l + \r + 0.5 * \r); 
    \coordinate (mrt) at (\l + 0.5 * \r, \h - 0.5 * \r); 
    \coordinate (lfv) at ($(mlf) + (0.5 * \l, 0.5 * \l)$); 
    \coordinate (rfb) at ($(mrf) + (0.5 * \r, 0.5 * \r)$); 
    \coordinate (lfg) at ($(mlf) + (- 0.5 * \l, 0.5 * \l)$); 
    \coordinate (rfg) at ($(mrf) + (- 0.5 * \r, 0.5 * \r)$); 

    \draw[syncLeft] (g) -- (lb) -- (lfv) -- (mlf) -- (lg) -- cycle;
    \draw[syncRight] (iv) -- (rb) -- (rfb) -- (mrf) -- (rv) -- cycle;

    \draw[syncLeft, thawingThaw]
        (0, \h) -- (0.5 * \l, \h - 0.5 * \l)
                -- (\l, \h);
    \draw[syncRight, thawingThaw]
        (\w, \h) -- (\l + 0.5 * \r, \h - 0.5 * \r)
                 -- (\w - \r, \h);

    \draw[find] (g) -- (lb);
    \draw[reflected] (lb) -- (mlf); 
    \draw[find] (iv) -- (rb);
    \draw[reflected] (rb) -- (rv) -- (m);

    \draw[traversingThaw] (m) -- (mlt);
    \draw[traversingThaw] (m) -- (mrt);

    \draw[midpoint] (mlf) -- (mlt);
    \draw[midpoint] (mrf) -- (mrt);

    \draw[slowed] (g) -- (mlf)
                      -- (\w/2, \h - \w/2);
    \draw[slowed] (iv) -- (mrf);

    \draw[freeze] (mlf) -- (lfg);
    \draw[freeze] (mlf) -- (lfv);
    \draw[freeze] (mrf) -- (rfb);
    \draw[freeze] (mrf) -- (rfg);

    \draw[vertex] (v) -- (\l, \h);
    \draw[vertex] (0, 0) -- (0, \h);
    \draw[vertex] (\w, 0) -- (\w, \h);

    \draw[edge] (0, 0) -- (v); 
    \draw[edge] (v) -- (\w, 0); 
    \draw[edge] (0, \h) -- (\l, \h); 
    \draw[edge] (\l, \h) -- (\w, \h); 
  \end{tikzpicture}
}

\def\drawMidpointAtOf(#1,#2,#3){%
  \begin{scope}[shift = #1]
    \begin{scope}
      \clip (0, 3pt) rectangle (-3pt, -3pt);
      \fill[#2, draw = black] (0, 0) circle (2pt);
    \end{scope}
    \begin{scope}
      \clip (0, -3pt) rectangle (3pt, 3pt);
      \fill[#3, draw = black] (0, 0) circle (2pt);
    \end{scope}
  \end{scope}
}

\def\drawMidpointAtOfThreeEdges(#1,#2,#3,#4){%
  \begin{scope}[shift = #1]
    \fill[#3, draw = none] (-1pt, -2pt) rectangle (1pt, 2pt);
    \draw[black] (-1pt, -2pt) -- (1pt, -2pt)
                 (-1pt, 2pt) -- (1pt, 2pt);
    \begin{scope}[shift = {(-1pt,0)}]
      \clip (0, 3pt) rectangle (-3pt, -3pt);
      \fill[#2, draw = black] (0, 0) circle (2pt);
    \end{scope}
    \begin{scope}[shift = {(1pt,0)}]
      \clip (0, -3pt) rectangle (3pt, 3pt);
      \fill[#4, draw = black] (0, 0) circle (2pt);
    \end{scope}
  \end{scope}
}

\def\figureFSSPWithThreeEdgesInARowAndGeneralAtTheSecondVertexFromTheLeft{%
  \begin{tikzpicture}[> = To, x = 1cm, y = {-1cm}] 
    \pgfmathsetmacro{\l}{1} 
    \pgfmathsetmacro{\b}{3} 
    \pgfmathsetmacro{\r}{2} 
    \pgfmathsetmacro{\w}{\l + \b + \r} 
    \pgfmathsetmacro{\h}{\w + \b + \r} 

    \coordinate (g) at (\l, 0); 
    \coordinate (m) at (\w/2, \h - \w/2); 

    \coordinate (mf2) at ({(\l + \b)/2}, {\b + (\l + \b)/2}); 
    \coordinate (ml2) at ({\l + (\b + \r)/2}, {1.5 * (\b + \r)}); 
    \coordinate (ml) at (\l/2, 1.5 * \l); 
    \coordinate (mr) at (\l + \b + \r/2, \b + 1.5 * \r); 
    \coordinate (mb) at (\l + \b/2, 1.5 * \b); 

    \coordinate (mf2e) at ({(\l + \b)/2}, {\h - (\l + \b)/2}); 
    \coordinate (ml2e) at ({\l + (\b + \r)/2}, {\h - (\b + \r)/2}); 
    \coordinate (mle) at (\l/2, \h - \l/2); 
    \coordinate (mre) at (\l + \b + \r/2, \h - \r/2); 
    \coordinate (mbe) at (\l + \b/2, \h - \b/2); 

    \draw[syncLeft] (g) -- ($(ml) + (\l/2, \l/2)$) -- (ml) -- ($(ml) + (-\l/2, \l/2)$) -- (0, \l) -- cycle;
    \draw[syncMiddle] (g) -- (\l + \b, \b) -- ($(mb) + (\b/2, \b/2)$) -- (mb) -- ($(mb) + (-\b/2, \b/2)$) -- cycle;
    \draw[syncRight] (\l + \b, \b) -- (\w, \b + \r) -- ($(mr) + (\r/2, \r/2)$) -- (mr) -- ($(mr) + (-\r/2, \r/2)$) -- cycle;

    \draw[find] (g) -- (0, \l);
    \draw[reflected] (0, \l) -- (\l, 2 * \l);
    \draw[slowed] (\l, 2 * \l) -- (m); 
    \draw[midpoint] (mf2) -- (mf2e); 

    \draw[find] (g) -- (\l + \b, \b);
    \draw[reflected] (\l + \b, \b) -- (mf2); 
    \draw[find] (\l + \b, \b) -- (\w, \b + \r);
    \draw[reflected] (\w, \b + \r) -- (m);
    \draw[slowed] (\l + \b, \b) -- (mr);
    \draw[midpoint] (mr) -- (mre);

    \draw[slowed] (g) -- (\l/2, 1.5 * \l);
    \draw[midpoint] (ml) -- (mle);

    \draw[slowed] (g) -- (mb) -- (ml2);
    \draw[midpoint] (mb) -- (mbe);
    \draw[midpoint] (ml2) -- (ml2e);

    \draw[freeze] (ml) -- ($(ml) + (\l/2, \l/2)$)
                  (ml) -- ($(ml) + (-\l/2, \l/2)$)
                  (mb) -- ($(mb) + (\b/2, \b/2)$)
                  (mb) -- ($(mb) + (-\b/2, \b/2)$)
                  (mr) -- ($(mr) + (\r/2, \r/2)$)
                  (mr) -- ($(mr) + (-\r/2, \r/2)$);

    \draw[traversingThaw] (m) -- (mf2e)
                          (m) -- (ml2e)
                          (mf2e) -- (mle)
                          (mf2e) -- (mbe)
                          (ml2e) -- (mbe)
                          (ml2e) -- (mre);

    \draw[syncLeft] (mle) -- ($(mle) + (\l/2, \l/2)$) -- ($(mle) + (-\l/2, \l/2)$);
    \draw[thawingThaw] (mle) -- ($(mle) + (\l/2, \l/2)$)
                       (mle) -- ($(mle) + (-\l/2, \l/2)$);
    \draw[syncMiddle] (mbe) -- ($(mbe) + (\b/2, \b/2)$) -- ($(mbe) + (-\b/2, \b/2)$);
    \draw[thawingThaw] (mbe) -- ($(mbe) + (\b/2, \b/2)$)
                       (mbe) -- ($(mbe) + (-\b/2, \b/2)$);
    \draw[syncRight] (mre) -- ($(mre) + (\r/2, \r/2)$) -- ($(mre) + (-\r/2, \r/2)$);
    \draw[thawingThaw] (mre) -- ($(mre) + (\r/2, \r/2)$)
                       (mre) -- ($(mre) + (-\r/2, \r/2)$);

    \drawMidpointAtOfThreeEdges((m),leftFill,overlayLeftFill,rightFill);
    \drawMidpointAtOf((mf2),leftFill,overlayLeftFill);
    \drawMidpointAtOf((mf2e),leftFill,overlayLeftFill);
    \drawMidpointAtOf((mr),rightFill,rightFill);
    \drawMidpointAtOf((mre),rightFill,rightFill);
    \drawMidpointAtOf((ml),leftFill,leftFill);
    \drawMidpointAtOf((mle),leftFill,leftFill);
    \drawMidpointAtOf((mb),overlayLeftFill,overlayLeftFill);
    \drawMidpointAtOf((mbe),overlayLeftFill,overlayLeftFill);
    \drawMidpointAtOf((ml2),overlayLeftFill,rightFill);
    \drawMidpointAtOf((ml2e),overlayLeftFill,rightFill);

    \draw[vertex] (0, 0) -- (0, \h)
                  (\l, 0) -- (\l, \h)
                  (\l + \b, 0) -- (\l + \b, \h)
                  (\w, 0) -- (\w, \h);
    \draw[edge, leftColour] (0, 0) -- (\l, 0)
                            (0, \h) -- (\l, \h);
    \draw[edge, overlayLeftColour] (\l, 0) -- (\l + \b, 0)
                              (\l, \h) -- (\l + \b, \h);
    \draw[edge, rightColour] (\l + \b, 0) -- (\w, 0)
                             (\l + \b, \h) -- (\w, \h);
  \end{tikzpicture}
}

\def\figureTheTreeThatIsSynchronised{%
  \begin{tikzpicture}[> = To, x = 1cm, y = {-1cm}]
    \draw[leftColour] (0, 0) -- (-0.707, 0.707);
    \draw[overlayLeftColour] (0, 0) -- (2.1213, 2.1213);
    \draw[rightColour] (2.1213, 2.1213) -- (2.1213, 4.1213);
  \end{tikzpicture}
}

\def\lengthsForFSSPWithThreeEdgesIncidentToOneVertex{%
  \pgfmathsetmacro{\l}{3} 
  \pgfmathsetmacro{\r}{2} 
  \pgfmathsetmacro{\o}{1.333} 
  \pgfmathsetmacro{\w}{\l + \r} 
}

\def\verticesAndEdgesForFSSPWithThreeEdgesIncidentToOneVertex{%
  \draw[vertex] (v) -- (\l, \h);
  \draw[vertex, leftColour] (0, 0) -- (0, \h);
  \draw[vertex, rightColour] (\w, 0) -- (\w, \h);
  \draw[vertex, overlayLeftColour] (\l - \o, 0) -- (\l - \o, \h);
  \draw[vertex, overlayRightColour] (\l + \o, 0) -- (\l + \o, \h);

  \draw[edge, leftColour] (0, 0) -- (v) 
                          (0, \h) -- (\l, \h); 
  \draw[edge, rightColour] (v) -- (\w, 0) 
                           (\l, \h) -- (\w, \h); 
  \draw[edge, overlayLeftColour] ($(v) - (0, 0.4pt)$) -- ($(v) + (-\o, 0) - (0, 0.4pt)$)
                                 ($(v) + (0, \h) - (0, 0.4pt)$) -- ($(v) + (-\o, \h) - (0, 0.4pt)$);
  \draw[edge, overlayRightColour] ($(v) - (0, 0.4pt)$) -- ($(v) + (\o, 0) - (0, 0.4pt)$)
                                  ($(v) + (0, \h) - (0, 0.4pt)$) -- ($(v) + (\o, \h) - (0, 0.4pt)$);
}

\def\figureFSSPWithThreeEdgesIncidentToTheGeneralVertex{%
  \begin{tikzpicture}[> = To, x = 1cm, y = {-1cm}] 
    \lengthsForFSSPWithThreeEdgesIncidentToOneVertex;
    \pgfmathsetmacro{\h}{\l + \l + \r} 

    \begin{scope}[shift = {(0.8, -0.8)}, scale = 0.2]
      \draw[leftColour] (0, 0) -- (-2.1213, 2.1213);
      \draw[overlayLeftColour] (0, 0) -- (0, 1.333);
      \draw[overlayRightColour] ($(0, 0) + (0.4pt, 0)$) -- ($(0, 1.333) + (0.4pt, 0)$);
      \draw[rightColour] (0, 0) -- (1.414, 1.414);
    \end{scope}
    \begin{scope}[shift = {(4.3, -0.8)}, scale = 0.2]
      \draw[leftColour] (0, 0) -- (-2.1213, 2.1213);
      \draw[overlayLeftColour] (0, 0) -- (0, -1.333);
      \draw[overlayRightColour] ($(0, 0) + (0.4pt, 0)$) -- ($(0, -1.333) + (0.4pt, 0)$);
      \draw[rightColour] (0, 0) -- (1.414, 1.414);
    \end{scope}

    \coordinate (v) at (\l, 0); 
    \coordinate (lb) at (0, \l); 
    \coordinate (rb) at (\w, \r); 
    \coordinate (obl) at (\l - \o, \o); 
    \coordinate (obr) at (\l + \o, \o); 
    \coordinate (lg) at (\l, 2 * \l); 
    \coordinate (rg) at (\l, 2 * \r); 
    \coordinate (og) at (\l, 2 * \o); 
    \coordinate (m) at (0.5 * \w, \l + 0.5 * \w); 
    \coordinate (m1) at ({0.5 * (\l + \o)}, {\l + 0.5 * (\l + \o)}); 
    \coordinate (m2) at ({\l - \o + 0.5 * (\r + \o)}, {\r + 0.5 * (\r + \o)}); 
    \coordinate (m1t) at ({0.5 * (\l + \o)}, {\h - 0.5 * (\l + \o)}); 
    \coordinate (m2t) at ({\l - \o + 0.5 * (\r + \o)}, {\h - 0.5 * (\r + \o)}); 
    \coordinate (mlf) at (0.5 * \l, \l + 0.5 * \l); 
    \coordinate (mlt) at (0.5 * \l, \h - 0.5 * \l); 
    \coordinate (mrf) at (\l + 0.5 * \r, \r + 0.5 * \r); 
    \coordinate (mrt) at (\l + 0.5 * \r, \h - 0.5 * \r); 
    \coordinate (mofl) at (\l - 0.5 * \o, \o + 0.5 * \o); 
    \coordinate (motl) at (\l - 0.5 * \o, \h - 0.5 * \o); 
    \coordinate (mofr) at (\l + 0.5 * \o, \o + 0.5 * \o); 
    \coordinate (motr) at (\l + 0.5 * \o, \h - 0.5 * \o); 
    \coordinate (lfb) at ($(mlf) + (- 0.5 * \l, 0.5 * \l)$); 
    \coordinate (rfb) at ($(mrf) + (0.5 * \r, 0.5 * \r)$); 
    \coordinate (ofbl) at ($(mofl) + (- 0.5 * \o, 0.5 * \o)$); 
    \coordinate (ofbr) at ($(mofr) + (0.5 * \o, 0.5 * \o)$); 
    \coordinate (lfg) at ($(mlf) + (0.5 * \l, 0.5 * \l)$); 
    \coordinate (rfg) at ($(mrf) + (- 0.5 * \r, 0.5 * \r)$); 
    \coordinate (ofg) at ($(mofr) + (- 0.5 * \o, 0.5 * \o)$); 

    \fill[syncLeft, leftColour, draw = none] (v) -- (lb) -- (lfb) -- (mlf) -- (lg) -- cycle;
    \fill[syncRight, rightColour, draw = none] (v) -- (rb) -- (rfb) -- (mrf) -- (rg) -- cycle;
    \fill[syncMiddle, overlayLeftColour, draw = none] (v) -- (obl) -- (ofbl) -- (mofl) -- (og) -- cycle;
    \fill[syncMiddle, overlayRightColour, draw = none] (v) -- (obr) -- (ofbr) -- (mofr) -- (og) -- cycle;

    \draw[find, leftColour] (v) -- (lb);
      \draw[reflected, leftColour] (lb) -- (m); 
    \draw[find, rightColour] (v) -- (rb);
      \draw[reflected, rightColour] (rb) -- (rg);
    \draw[find, overlayLeftColour] (v) -- (obl);
      \draw[reflected, overlayLeftColour] (obl) -- (og);
      \draw[slowed, rightColour] (og) -- (m2);
    \draw[find, overlayRightColour] (v) -- (obr);
      \draw[reflected, overlayRightColour] (obr) -- (og);
      \draw[slowed, leftColour] (og) -- (m1);

    \draw[thawingThaw, syncLeft, leftColour]
        (0, \h) -- (0.5 * \l, \h - 0.5 * \l)
                -- (\l, \h);
    \draw[thawingThaw, syncRight, rightColour]
        (\w, \h) -- (\l + 0.5 * \r, \h - 0.5 * \r)
                 -- (\w - \r, \h);
    \draw[thawingThaw, syncMiddle, overlayLeftColour]
        (\l, \h) -- (\l - 0.5 * \o, \h - 0.5 * \o)
                 -- (\l - \o, \h);
    \draw[thawingThaw, syncMiddle, overlayRightColour]
        (\l, \h) -- (\l + 0.5 * \o, \h - 0.5 * \o)
                 -- (\l + \o, \h);

    \draw[traversingThaw, leftColour] (m) -- (m1t);
      \draw[traversingThaw, leftColour] (m1t) -- (mlt);
      \draw[traversingThaw, leftColour] (m1t) -- ($(m1t) + ({\l - 0.5 * (\l + \o)}, {\l - 0.5 * (\l + \o)})$); 
          \draw[traversingThaw, overlayRightColour] ($(m1t) + ({\l - 0.5 * (\l + \o)}, {\l - 0.5 * (\l + \o)})$) -- (motr); 

    \draw[traversingThaw, leftColour] (m) -- ($(m) + ({\l - 0.5 * (\l + \r)}, {\l - 0.5 * (\l + \r)})$); 
      \draw[traversingThaw, rightColour] ($(m) + ({\l - 0.5 * (\l + \r)}, {\l - 0.5 * (\l + \r)})$) -- (m2t); 
          \draw[traversingThaw, rightColour] (m2t) -- (mrt); 
          \draw[traversingThaw, rightColour] (m2t) -- ($(m2t) + ({-(\l - \o + 0.5 * (\r + \o) - \l)}, {\l - \o + 0.5 * (\r + \o) - \l})$); 
              \draw[traversingThaw, overlayLeftColour] ($(m2t) + ({-(\l - \o + 0.5 * (\r + \o) - \l)}, {\l - \o + 0.5 * (\r + \o) - \l})$) -- (motl); 

    \draw[slowed, leftColour] (v) -- (mlf); 
    \draw[slowed, rightColour] (v) -- (mrf); 
    \draw[slowed, overlayLeftColour] ($(v) + (0.4pt, 0)$) -- ($(mofl) + (0.4pt, 0)$); 
    \draw[slowed, overlayRightColour] ($(v) - (0.4pt, 0)$) -- ($(mofr) - (0.4pt, 0)$); 
    \draw[slowed, leftColour] (\l, \r + \r) -- (0.5 * \w, {\r + \r + 1.5 * (\l - \r)});

    \draw[freeze, leftColour] (mlf) -- (lfg)
                              (mlf) -- (lfb);
    \draw[freeze, rightColour] (mrf) -- (rfb)
                               (mrf) -- (rfg);
    \draw[freeze, overlayLeftColour] (mofl) -- (ofbl)
                                     (mofl) -- (ofg);
    \draw[freeze, overlayRightColour] (mofr) -- (ofbr)
                                      (mofr) -- (ofg);

    \drawMidpointAtOf((m),leftFill,rightFill);
    \draw[midpoint, leftColour] (mlf) -- (mlt); 
        \drawMidpointAtOf((mlf),leftFill,leftFill);
        \drawMidpointAtOf((mlt),leftFill,leftFill);
    \draw[midpoint, rightColour] (mrf) -- (mrt); 
        \drawMidpointAtOf((mrf),rightFill,rightFill);
        \drawMidpointAtOf((mrt),rightFill,rightFill);
    \draw[midpoint, overlayLeftColour] (mofl) -- (motl); 
        \drawMidpointAtOf((mofl),overlayLeftFill,overlayLeftFill);
        \drawMidpointAtOf((motl),overlayLeftFill,overlayLeftFill);
    \draw[midpoint, overlayRightColour] (mofr) -- (motr); 
        \drawMidpointAtOf((mofr),overlayRightFill,overlayRightFill);
        \drawMidpointAtOf((motr),overlayRightFill,overlayRightFill);
    \draw[midpoint, leftColour] (m1) -- (m1t); 
        \drawMidpointAtOf((m1),leftFill,overlayRightFill);
        \drawMidpointAtOf((m1t),leftFill,overlayRightFill);
    \draw[midpoint, rightColour] (m2) -- (m2t); 
        \drawMidpointAtOf((m2),overlayLeftFill,rightFill);
        \drawMidpointAtOf((m2t),overlayLeftFill,rightFill);

    \verticesAndEdgesForFSSPWithThreeEdgesIncidentToOneVertex;
  \end{tikzpicture}
}

\def\figureFSSPWithThreeEdgesIncidentToTheSameVertexAndTheGeneralIsNotOnTheLongestPath{%
  \begin{tikzpicture}[> = To, x = 1cm, y = {-1cm}] 
    \lengthsForFSSPWithThreeEdgesIncidentToOneVertex
    \pgfmathsetmacro{\h}{\o + \l + \l + \r} 

    \coordinate (v) at (\l, 0); 

    \coordinate (gl) at (\l - \o, 0); 
    \coordinate (gr) at (\l + \o, 0); 
    \coordinate (lgv) at (\l, \o); 
    \coordinate (fov) at (\l, 2 * \o); 
    \coordinate (folg) at (\l - \o, 2 * \o); 
    \coordinate (forg) at (\l + \o, 2 * \o); 
    \coordinate (mol) at (\l - \o/2, 1.5 * \o); 
    \coordinate (mor) at (\l + \o/2, 1.5 * \o); 
    \path let \p1 = (mol) in coordinate (stol) at (\x1, \h - \o/2); 
    \path let \p1 = (mor) in coordinate (stor) at (\x1, \h - \o/2); 
    \coordinate (sv) at ($(mol) + (\o/2, 3 * \o/2)$); 

    \coordinate (lvl) at ($(lgv) + (-\l, \l)$); 
    \coordinate (rlv) at ($(lvl) + (\l, \l)$); 
    \coordinate (ml) at ($(lvl) + (\l/2, \l/2)$); 
    \coordinate (fll) at ($(ml) + (-\l/2, \l/2)$); 
    \coordinate (flv) at ($(ml) + (\l/2, \l/2)$); 
    \path let \p1 = (ml) in coordinate (stl) at (\x1, \h - \l/2); 

    \coordinate (lvr) at ($(lgv) + (\r, \r)$); 
    \coordinate (rfv) at ($(lvr) + (-\r, \r)$); 
    \coordinate (mr) at ($(lvr) + (-\r/2, \r/2)$); 
    \coordinate (frr) at ($(mr) + (\r/2, \r/2)$); 
    \coordinate (frv) at ($(mr) + (-\r/2, \r/2)$); 
    \path let \p1 = (mr) in coordinate (str) at (\x1, \h - \r/2); 

    \coordinate (mlo) at ({(\l + \o)/2}, {1.5 * (\l + \o)}); 
    \coordinate (mro) at ({\l - \o + (\o + \r)/2}, {1.5 * (\o + \r)}); 
    \coordinate (m) at ({(\l + \r)/2}, {\o + \l + (\l + \r)/2}); 
    \coordinate (stlo) at ({(\l + \o)/2}, {\h - (\l + \o)/2}); 
    \coordinate (stro) at ({\l - \o + (\o + \r)/2}, {\h - (\o + \r)/2}); 

    \coordinate (tv) at (\l, \h); 
    \coordinate (tr) at (\w, \h); 
    \coordinate (tl) at (0, \h); 
    \coordinate (tol) at (\l - \o, \h); 
    \coordinate (tor) at (\l + \o, \h); 
    \path let \p1 = (m) in coordinate (tmhv) at (\l, \y1 + \x1 - \l cm); 
    \path let \p1 = (stlo) in coordinate (tmlohv) at (\l, \y1 + \x1 - \l cm); 
    \path let \p1 = (stro) in coordinate (tmrohv) at (\l, \y1 - \x1 + \l cm); 

    \fill[syncLeft, leftColour, draw = none] (lgv) -- (lvl) -- (fll) -- (ml) -- (flv) -- cycle;
    \fill[syncRight, rightColour, draw = none] (lgv) -- (lvr) -- (frr) -- (mr) -- (frv) -- cycle;
    \fill[syncMiddle, overlayLeftColour, draw = none] (gl) -- (lgv) -- (fov) -- (mol) -- (folg) -- cycle;
    \fill[syncMiddle, overlayRightColour, draw = none] (gr) -- (lgv) -- (fov) -- (mor) -- (forg) -- cycle;

    \fill[syncLeft, leftColour, draw = none] (tl) -- (stl) -- (tv) -- cycle;
    \fill[syncRight, rightColour, draw = none] (tr) -- (str) -- (tv) -- cycle;
    \fill[syncMiddle, overlayLeftColour, draw = none] (tol) -- (stol) -- (tv) -- cycle;
    \fill[syncMiddle, overlayRightColour, draw = none] (tor) -- (stor) -- (tv) -- cycle;

    \draw[find, overlayLeftColour] (gl) -- (lgv); 
    \draw[find, rightColour] (lgv) -- (lvr); 
    \draw[reflected, overlayLeftColour] (lgv) -- (mol); 
    \draw[slowed, rightColour] (lgv) -- (mr); 
    \draw[reflected, rightColour] (lvr) -- ($(lvr) + (-\r, \r)$); 
        \draw[slowed, leftColour] ($(lvr) + (-\r, \r)$) -- (m); 

    \draw[find, overlayRightColour] (gr) -- (lgv); 
    \draw[find, leftColour] (lgv) -- (lvl); 
    \draw[reflected, overlayRightColour] (lgv) -- (mor); 
    \draw[slowed, leftColour] (lgv) -- (ml); 
    \draw[reflected, leftColour] (lvl) -- (m); 

    \draw[slowed, overlayLeftColour] (gl) -- (mol) -- (sv);
        \draw[slowed, rightColour] (sv) -- (mro);
    \draw[slowed, overlayRightColour] (gr) -- (mor) -- (sv);
        \draw[slowed, leftColour] (sv) -- (mlo);

    \draw[midpoint, leftColour] (ml) -- (stl)
                                (mlo) -- (stlo);
    \draw[midpoint, rightColour] (mr) -- (str)
                                 (mro) -- (stro);
    \draw[midpoint, overlayLeftColour] (mol) -- (stol);
    \draw[midpoint, overlayRightColour] (mor) -- (stor);

    \draw[freeze, leftColour] (ml) -- (fll)
                              (ml) -- (flv);
    \draw[freeze, rightColour] (mr) -- (frr)
                               (mr) -- (frv);
    \draw[freeze, overlayLeftColour] (mol) -- (fov)
                                     (mol) -- (folg);
    \draw[freeze, overlayRightColour] (mor) -- (fov)
                                      (mor) -- (forg);

    \draw[traversingThaw, leftColour] (m) -- (stlo) -- (stl);
    \draw[traversingThaw, leftColour] (m) -- (tmhv);
        \draw[traversingThaw, rightColour] (tmhv) -- (stro) -- (str);
    \draw[traversingThaw, leftColour] (stlo) -- (tmlohv);
        \draw[traversingThaw, overlayRightColour] (tmlohv) -- (stor);
    \draw[traversingThaw, rightColour] (stro) -- (tmrohv);
        \draw[traversingThaw, overlayLeftColour] (tmrohv) -- (stol);

    \draw[thawingThaw, leftColour] (stl) -- (tl)
                                   (stl) -- (tv);
    \draw[thawingThaw, rightColour] (str) -- (tr)
                                    (str) -- (tv);
    \draw[thawingThaw, overlayLeftColour] (stol) -- (tol)
                                          (stol) -- (tv);
    \draw[thawingThaw, overlayRightColour] (stor) -- (tor)
                                           (stor) -- (tv);

    \drawMidpointAtOf((ml),leftFill,leftFill);
        \drawMidpointAtOf((stl),leftFill,leftFill);
    \drawMidpointAtOf((mr),rightFill,rightFill);
        \drawMidpointAtOf((str),rightFill,rightFill);
    \drawMidpointAtOf((mol),overlayLeftFill,overlayLeftFill);
        \drawMidpointAtOf((stol),overlayLeftFill,overlayLeftFill);
    \drawMidpointAtOf((mor),overlayRightFill,overlayRightFill);
        \drawMidpointAtOf((stor),overlayRightFill,overlayRightFill);
    \drawMidpointAtOf((mlo),leftFill,overlayRightFill);
        \drawMidpointAtOf((stlo),leftFill,overlayRightFill);
    \drawMidpointAtOf((mro),overlayLeftFill,rightFill);
        \drawMidpointAtOf((stro),overlayLeftFill,rightFill);
    \drawMidpointAtOf((m),leftFill,rightFill);

    \verticesAndEdgesForFSSPWithThreeEdgesIncidentToOneVertex
  \end{tikzpicture}
}

\def\figureTheTimeAtWhichTheMidpointsOfLongestPathsAreFound{%
  \begin{tikzpicture}
    \fill (0, 1) circle (2pt) node[right] {$\general$};
    \fill (0, -1) circle (2pt) node[below] {$v''$};
    \fill (-3, -1) circle (2pt) node[below] {$\sourceOf(\hat{p})$};
    \fill (4, -1) circle (2pt) node[below] {$\targetOf(\hat{p})$};
    \draw (0, 1) -- (0, -1)
          (-3, -1) -- (0, -1) -- (4, -1);

    \fill (-2, 0.5) circle (2pt) node[left] {$v$};
    \fill (0, 0.5) circle (2pt) node[right] {$v'$};
    \draw[loosely dashed] (-2, 0.5) -- (0, 0.5);

    \fill (-2, -2) circle (2pt) node[below] {$v$};
    \fill (-1, -1) circle (2pt) node[above] {$v'$};
    \draw[thick, loosely dotted] (-2, -2) -- (-1, -1);

    \fill (3, -3) circle (2pt) node[below] {$v$};
    \fill (2, -1) circle (2pt) node[above] {$v'$};
    \draw[dashdotted] (3, -3) -- (2, -1);
  \end{tikzpicture}
}

\def\figureRemarkMidpointsOfMaximumWeightPathsAreRecognised{%
  \begin{tikzpicture}[tree/.style = {dashed, gray}]
    \coordinate (g) at (0, 0);
    \coordinate (u1) at (1, 0);
    \coordinate (u+) at (2, 0);
    \coordinate (uk) at (3, 0);
    \coordinate (hatv) at (4, 0);
    \coordinate (w1) at (4.5, -0.25);
    \coordinate (wl) at (5.25, -0.6125);
    \coordinate (hatv1) at (4 + 5, 5/2); 
    \coordinate (hatv2) at (4 + 8, -8/2); 
    \coordinate (m) at (4 + 1.47987, -1.47987/2); 

    \draw[tree] (g) -- ++(-0.8, -0.4) -- ++(0, 0.8) -- cycle;
        \draw[tree] (g) -- ++(-1, -0.7) -- ++(0, 0.9) -- cycle;
    \draw[tree] (u1) -- ++(-0.2, -1.5) -- ++(0.4, 0) -- cycle;
        \draw[tree] (u1) -- ++(-0.4, -1.7) -- ++(0.4, 0) -- cycle;
    \draw[tree] (uk) -- ++(-0.2, -2) -- ++(0.4, 0) -- cycle;
        \draw[tree] (uk) -- ++(-0.4, -2.2) -- ++(0.4, 0) -- cycle;
    \draw[tree] (hatv) -- ++(-0.2, -1.8) -- ++(0.4, 0) -- cycle;
        \draw[tree] (hatv) -- ++(-0.4, -2) -- ++(0.4, 0) -- cycle;
    \draw[tree] (w1) -- ++(-0.2, -1) -- ++(0.4, 0) -- cycle;
        \draw[tree] (w1) -- ++(-0.1, -1.1) -- ++(0.4, 0) -- cycle;
    \draw[tree] (wl) -- ++(-0.2, -2) -- ++(0.4, 0) -- cycle;
        \draw[tree] (wl) -- ++(-0.1, -2.1) -- ++(0.4, 0) -- cycle;
    \draw[tree] (hatv) -- ($(hatv1) + (-1, 1)$) -- ($(hatv1) - (-1, 1)$) -- cycle;
    \draw[tree] (hatv) -- ($(hatv2) + (-1, -1)$) -- ($(hatv2) - (-1, -1)$) -- cycle;

    \draw[dotted, thick] (g) -- (hatv);
    \draw (hatv) -- (hatv1)
          (hatv) -- (hatv2);

    \fill (g) circle (2pt) node[above] {$\general$};
    \fill (u1) circle (2pt) node[above] {$u_1$};
    \fill (u+) circle (0pt) node[above] {$\dotsb$};
    \fill (uk) circle (2pt) node[above] {$u_k$};
    \fill (hatv) circle (2pt) node[above] {$\hat{v}$};
    \fill (w1) circle (2pt) node[above right] {$w_1$};
    \fill (wl) circle (2pt) node[above right] {$w_\ell$}; 
    \fill (hatv1) circle (2pt) node[above right] {$\hat{v}_1$};
    \fill (hatv2) circle (2pt) node[below right] {$\hat{v}_2$};
    \fill[draw = black, fill = white] (m) circle (2pt) node[below] {$\hat{\midpoint}$};
  \end{tikzpicture}
}

\def\figureWhenDoesCountSignalMemoriseThePenultimateDirectionBaseCase{%
  \begin{tikzpicture}[tree/.style = {dashed, gray}]
    \coordinate (g) at (0, 0);
    \draw[tree] (g) -- ++(-0.8, -0.4) -- ++(0, 0.8) -- cycle;
        \draw[tree] (g) -- ++(-1, -0.7) -- ++(0, 0.9) -- cycle;
    \fill (g) circle (2pt) node[above] {$\general$};

    \coordinate (w) at (1.5, 0);
    \draw[tree] (w) -- ++(-0.2, -1.5) -- ++(0.4, 0) -- cycle;
        \draw[tree] (w) -- ++(-0.4, -1.7) -- ++(0.4, 0) -- cycle;
    \fill (w) circle (2pt) node[above] {$w$};

    \coordinate (v) at (5, 0);
    \fill (v) circle (2pt) node[below left] {$v$};
    \coordinate (v') at ($(v) + (60: 0.5)$);
    \fill (v') circle (2pt) node[left] {$v'$};
    \coordinate (v'') at ($(v) + (0: 1)$);
    \fill (v'') circle (2pt) node[below] {$v''$};
    \coordinate (v''') at ($(v) + (-60: 1.5)$);
    \fill (v''') circle (2pt) node[left] {$v'''$};

    \draw (g) -- (w)
          (w) -- (v)
          (v) -- (v')
          (v) -- (v'')
          (v) -- (v''');
  \end{tikzpicture}
}

\def\figureWhenDoesCountSignalMemoriseThePenultimateDirectionInductiveStep{%
  \begin{tikzpicture}[tree/.style = {dashed, gray}]
    \coordinate (g) at (0, 0);
    \draw[tree] (g) -- ++(-0.8, -0.4) -- ++(0, 0.8) -- cycle;
        \draw[tree] (g) -- ++(-1, -0.7) -- ++(0, 0.9) -- cycle;
    \fill (g) circle (2pt) node[above] {$\general$};

    \coordinate (w) at (1.5, 0);
    \draw[tree] (w) -- ++(-0.2, -1.5) -- ++(0.4, 0) -- cycle;
        \draw[tree] (w) -- ++(-0.4, -1.7) -- ++(0.4, 0) -- cycle;
    \fill (w) circle (2pt) node[above] {$w$};

    \coordinate (v) at (5, 0);
    \fill (v) circle (2pt) node[below left] {$v$};

    \coordinate (v') at ($(v) + (60: 0.5)$);
    \draw[tree] (v') -- ++(0.6, 1) -- ++(0.4, 0) -- cycle;
        \draw[tree] (v') -- ++(0.4, 1.1) -- ++(0.5, 0) -- cycle;
    \fill (v') circle (2pt) node[left] {$v'$};

    \coordinate (v'') at ($(v) + (0: 1)$);
    \draw[tree] (v'') -- ++(1, -0.3) -- ++(0, 0.8) -- cycle;
        \draw[tree] (v'') -- ++(1.2, -0.5) -- ++(0, 0.9) -- cycle;
    \fill (v'') circle (2pt) node[below] {$v''$};

    \coordinate (v''') at ($(v) + (-60: 1.5)$);
    \draw[tree] (v''') -- ++(0.2, -1) -- ++(0.4, 0) -- cycle;
        \draw[tree] (v''') -- ++(0.4, -1.1) -- ++(0.5, 0) -- cycle;
    \fill (v''') circle (2pt) node[left] {$v'''$};

    \draw (g) -- (w)
          (w) -- (v)
          (v) -- (v')
          (v) -- (v'')
          (v) -- (v''');
  \end{tikzpicture}
}

\def\figureWGeneralV{%
  \begin{tikzpicture}[tree/.style = {dashed, gray}]
    \coordinate (g) at (0, 0);
    \draw[tree] (g) -- ++(-0.8, -0.4) -- ++(0, 0.8) -- cycle;
        \draw[tree] (g) -- ++(-1, -0.7) -- ++(0, 0.9) -- cycle;
    \fill (g) circle (2pt) node[above] {$\general$};

    \coordinate (v) at (3, 0);
    \draw[tree] (v) -- ++(7, -1) -- ++(0, 1.7) -- cycle;
    \draw[tree] (v) -- ++(-0.2, -1.5) -- ++(0.4, 0) -- cycle;
        \draw[tree] (v) -- ++(-0.4, -1.7) -- ++(0.4, 0) -- cycle;
    \fill (v) circle (2pt) node[above] {$v$};

    \coordinate (w) at (2.5, 0);
    \draw[tree] (w) -- ++(-0.3, -0.7) -- ++(0.5, 0) -- cycle;
    \fill (w) circle (2pt) node[above] {$w$};

    \coordinate (w') at (2, 0);
    \draw[tree] (w') -- ++(-0.6, -1) -- ++(0.4, 0) -- cycle;
        \draw[tree] (w') -- ++(-0.4, -1.1) -- ++(0.5, 0) -- cycle;
    \fill (w') circle (2pt) node[above] {$w'$};

    \coordinate (w'') at (0.5, 0);
    \draw[tree] (w'') -- ++(-1, -2.3) -- ++(1.5, 0) -- cycle;
        \draw[tree] (w'') -- ++(-0.8, -1.1) -- ++(0.8, 0) -- cycle;
    \fill (w'') circle (2pt) node[above] {$w''$};

    \draw (g) -- (v);
  \end{tikzpicture}
}

\def\figureMarkedSignalsInnerVerticesInducitveStep{%
  \begin{tikzpicture}[tree/.style = {dashed, gray}]
    \coordinate (g) at (0, 0);
    \draw[tree] (g) -- ++(-0.8, -0.4) -- ++(0, 0.8) -- cycle;
        \draw[tree] (g) -- ++(-1, -0.7) -- ++(0, 0.9) -- cycle;
    \fill (g) circle (2pt) node[above] {$\general$};

    \coordinate (v) at (3, 0);
    \draw[tree] (v) -- ++(7, -1) -- ++(0, 1.7) -- cycle;
    \draw[tree] (v) -- ++(-0.2, -1.5) -- ++(0.4, 0) -- cycle;
        \draw[tree] (v) -- ++(-0.4, -1.7) -- ++(0.4, 0) -- cycle;
    \fill (v) circle (2pt) node[above] {$v$};

    \coordinate (v') at (2, 0);
    \draw[tree] (v') -- ++(-0.6, -1) -- ++(0.4, 0) -- cycle;
        \draw[tree] (v') -- ++(-0.4, -1.1) -- ++(0.5, 0) -- cycle;
    \fill (v') circle (2pt) node[above] {$v'$};

    \draw (g) -- (v')
          (v') -- (v);
  \end{tikzpicture}
}

\def\figureMidpointsOfMaximumWeightPathsAreRecognised{%
  \subfloat[$g = \hat{v} = \hat{\midpoint}$]{%
    \begin{tikzpicture}[tree/.style = {dashed, gray}]
      \coordinate (g) at (0, 0);
      \coordinate (v1) at (-0.5, 0);
      \coordinate (v2) at (1, 0);
      \coordinate (hatv1) at (-4, 0);
      \coordinate (hatv2) at (4, 0);

      \draw[tree] (v1) -- ($(hatv2) + (0, -1)$) -- ($(hatv2) + (0, 1)$) --node[above] {$\maxTree_{v_1}$} cycle;
      \draw[tree] (v2) -- ($(hatv1) + (0, -1.5)$) -- ($(hatv1) + (0, 1.5)$) --node[above] {$\maxTree_{v_2}$} cycle;

      \draw[dotted, thick]
          (hatv1) -- (v1)
          (hatv2) -- (v2);
      \draw (v1) -- (g) -- (v2);

      \fill (g) circle (2pt) node[above] {$\general$}; 
          \path (g) circle (2pt) node[below] {$\hat{\midpoint}$};
      \fill (v1) circle (2pt) node[above] {$v_1$};
      \fill (v2) circle (2pt) node[above] {$v_2$};
      \fill (hatv1) circle (2pt) node[left] {$\hat{v}_1$};
      \fill (hatv2) circle (2pt) node[right] {$\hat{v}_2$};
    \end{tikzpicture}
    \label{figure:midpoints-of-maximum-weight-paths-are-recognised:general-lies-on-path-and-midpoint-is-nearest-to-general}
  }

  \subfloat[$g = \hat{v} \neq \hat{\midpoint}$]{%
    \begin{tikzpicture}[tree/.style = {dashed, gray}]
      \coordinate (g) at (-2, 0);
      \coordinate (m) at (1.5, 0);
      \coordinate (v1) at (1, 0);
      \coordinate (v2) at (2.5, 0);
      \coordinate (hatv1) at (-4, 0);
      \coordinate (hatv2) at (7, 0);

      \draw[tree] (v1) -- ($(hatv2) + (0, -1)$) -- ($(hatv2) + (0, 1)$) --node[above] {$\maxTree_{v_1}$} cycle;
      \draw[tree] (v2) -- ($(hatv1) + (0, -1.5)$) -- ($(hatv1) + (0, 1.5)$) --node[above] {$\maxTree_{v_2}$} cycle;

      \draw[dotted, thick]
          (g) -- (v1)
          (hatv1) -- (g)
          (hatv2) -- (v2);
      \draw (v1) -- (v2);

      \fill (g) circle (2pt) node[above] {$\general$}; 
      \fill[draw = black, fill = white] (m) circle (2pt) node[below] {$\hat{\midpoint}$};
      \fill (v1) circle (2pt) node[above] {$v_1$};
      \fill (v2) circle (2pt) node[above] {$v_2$};
      \fill (hatv1) circle (2pt) node[left] {$\hat{v}_1$};
      \fill (hatv2) circle (2pt) node[right] {$\hat{v}_2$};
    \end{tikzpicture}
    \label{figure:midpoints-of-maximum-weight-paths-are-recognised:general-lies-on-path-and-midpoint-is-not-nearest-to-general}
  }

  \subfloat[$g \neq \hat{v} = \hat{\midpoint}$]{%
    \begin{tikzpicture}[tree/.style = {dashed, gray}]
      \coordinate (g) at (-1, 0);
      \coordinate (v1) at (0.5, 0.25);
      \coordinate (v2) at (1, -0.5);
      \coordinate (hatv) at (0, 0);
      \coordinate (hatv1) at (4, 2);
      \coordinate (hatv2) at (4, -2);

      \draw[tree] (g) -- ++(-0.8, -0.4) -- ++(0, 0.8) -- cycle;
          \draw[tree] (g) -- ++(-1, -0.7) -- ++(0, 0.9) -- cycle;

      \draw[dotted, thick]
          (g) -- (hatv)
          (hatv1) -- (v1)
          (hatv2) -- (v2);
      \draw (hatv) -- (v1)
            (hatv) -- (v2);

      \fill (g) circle (2pt) node[above] {$\general$}; 
      \fill (v1) circle (2pt) node[above] {$v_1$};
      \fill (v2) circle (2pt) node[above] {$v_2$};
      \fill (hatv) circle (2pt) node[above] {$\hat{v}$};
          \path (hatv) circle (2pt) node[below] {$\hat{\midpoint}$};
      \fill (hatv1) circle (2pt) node[right] {$\hat{v}_1$};
      \fill (hatv2) circle (2pt) node[right] {$\hat{v}_2$};
    \end{tikzpicture}
    \label{figure:midpoints-of-maximum-weight-paths-are-recognised:midpoint-is-nearest-to-general}
  }

  \subfloat[$g \neq \hat{v} \neq \hat{\midpoint}$]{%
    \begin{tikzpicture}[tree/.style = {dashed, gray}]
      \coordinate (g) at (0, 0);
      \coordinate (hatv) at (4, 0);
      \coordinate (v1) at (5.25, -0.6125);
      \coordinate (v2) at (6.5, -1.25);
      \coordinate (hatv1) at (4 + 5, 5/2); 
      \coordinate (hatv2) at (4 + 8, -8/2); 
      \coordinate (m) at (4 + 1.47987, -1.47987/2); 

      \draw[tree] (g) -- ++(-0.8, -0.4) -- ++(0, 0.8) -- cycle;
          \draw[tree] (g) -- ++(-1, -0.7) -- ++(0, 0.9) -- cycle;

      \draw[dotted, thick]
          (g) -- (hatv)
          (hatv) -- (v1)
          (hatv) -- (hatv1)
          (v2) -- (hatv2);
      \draw
          (v1) -- (m)
          (m) -- (v2);

      \fill (g) circle (2pt) node[above] {$\general$};
      \fill (hatv) circle (2pt) node[above] {$\hat{v}$};
      \fill (v1) circle (2pt) node[above] {$v_1$};
      \fill (v2) circle (2pt) node[above] {$v_2$};
      \fill (hatv1) circle (2pt) node[above right] {$\hat{v}_1$};
      \fill (hatv2) circle (2pt) node[below right] {$\hat{v}_2$};
      \fill[draw = black, fill = white] (m) circle (2pt) node[below] {$\hat{\midpoint}$};
    \end{tikzpicture}
    \label{figure:midpoints-of-maximum-weight-paths-are-recognised:midpoint-is-not-nearest-to-general}
  }
}

\begin{document}
  \frenchspacing
  \raggedbottom
  \selectlanguage{british}

  \pagestyle{plain}
  \title{\rmfamily\normalfont\spacedallcaps{Signal Machine And Cellular Automaton Time-Optimal Quasi-Solutions Of The Firing Squad/Mob Synchronisation Problem On Connected Graphs}} 
  \author{\spacedlowsmallcaps{Simon Wacker} \\ \small\href{mailto:simon.wacker@kit.edu}{simon.wacker@kit.edu}} 
  \date{} 


  \maketitle

  \begin{abstract}
    We construct a time-optimal quasi-solution of the firing mob synchronisation problem over finite, connected, and undirected multigraphs whose maximum degrees are uniformly bounded by a constant. It is only a quasi-solution because its number of states depends on the graph or, from another perspective, does not depend on the graph but is countably infinite. To construct this quasi-solution we introduce signal machines over continuum representations of such multigraphs and construct a signal machine whose discretisation is a cellular automaton that quasi-solves the problem. This automaton uses a time-optimal solution of the firing squad synchronisation problem in dimension one with one general at one end to synchronise edges, and freezes and thaws the synchronisation of edges in such a way that all edges synchronise at the same time.
  \end{abstract}


  \paragraph{Introduction.} The firing squad synchronisation problem in dimension one with one general at one end is to synchronise each finite one-dimensional array of cells starting from one end of the array and the cell at this end is called \emph{general}. It was proposed by John R. Myhill in 1957, solved by John McCarthy and Marvin Lee Minsky, and published by Edward Forrest Moore in 1962 (see \cite{moore:1964}). The first time-optimal several-thousand-states solution was found by Eiichi Goto in 1962 (see \cite{goto:1962}), reduced to $16$ states by Abraham Waksman in 1966 (see \cite{waksman:1966}), and reduced to $8$ states by Robert Balzer in 1967 (see \cite{balzer:1967}). Hein D. Gerken found another time-optimal $7$-states solution in 1987 (see \cite{gerken:1987}) and Jacques Mazoyer found a time-optimal $6$-states solution also in 1987 (see \cite{mazoyer:1987}). It is unknown whether there is a time-optimal $5$-states solution but it is known that there is no time-optimal $4$-states solution, a result due to Robert Balzer and Peter Sanders (see \cite{balzer:1967,sanders:1994}). 

  The firing mob synchronisation problem is to synchronise each finite, connected, and undirected graph whose maximum degree is bounded by a fixed constant starting from any vertex and this vertex is called \emph{general}. It was solved by P. Rosenstiehl, J.\;R. Fiksel, and A. Holliger in 1972 (see \cite{rosenstiehl:1972}) and also by Francesco Romani in 1976 (see \cite{romani:1976}), where the latter solution achieves better running times than the former. The problem for specific classes of graphs were for example studied by Kojiro Kobayashi in 1977 and 1978 (see \cite{kobayashi:1977,kobayashi2:1978,kobayashi:1978}) and by Zsuzsanna Róka in 2000 (see \cite{roka:2000}). Karel Culik II and Simant Dube presented a solution of the general case in 1991 (see \cite{culik-dube:1991}). It needs $3.5 r$-many steps, where $r$ is the maximal distance of the general to a vertex and is called \emph{radius of the graph with respect to the general}. By using more and more states, the solution can be adjusted such that the number of steps it needs approaches $3 r$.

  It was shown that $r + d$ is a lower bound for the number of steps that solutions of the firing mob synchronisation problem need by John J. Grefenstette in 1983 (see \cite{grefenstette:1983}), where $d$ is the maximal distance between two vertices of the graph and is called \emph{diameter of the graph}. Because there are graphs and choices of generals such that the diameter is $2 r$, the solutions by Karel Culik II and Simant Dube approach the optimal number of steps, namely $3 r$, if $r$ is taken as problem size. However, if $r + d$ is taken as problem size, then their solutions do not approach the optimal number of steps.

  In the present chapter we construct a time-optimal quasi-solution that needs exactly $r + d$ steps but whose number of states depends on the graph or, from another perspective, does not depend on the graph but is countably infinite (this is why we call it a quasi-solution). It can also be turned into a time-optimal quasi-solution of the firing squad synchronisation problem for any region in any dimension with one general at any position by regarding each region to be synchronised as a graph, where cells in the region are vertices and edges are neighbourhood relationships. 

  However, restricted to specific classes of problems, the quasi-solution may not be time-optimal. For example, restricted to rectangular regions with one general at one corner, the quasi-solution needs $2 (k + \ell - 2)$-many steps whereas $(k + \ell + \max\setOf{k, \ell} - 3)$-many steps is optimal (see for example \cite{umeo:2009}), where $k$ and $\ell$ are the side lengths of the rectangle. Nevertheless, because the quasi-solution is (trivially) embeddable in the sense of \cite{grefenstette:1983}, according to theorem~1 in \cite{grefenstette:1983}, it can be combined with finitely many embeddable time-optimal solutions for specific classes of problems to get one quasi-solution that is also time-optimal for those classes. Examples of solutions for specific classes, like rectangular regions with one general at the upper left corner, are given in sections~5 and~6 in \cite{grefenstette:1983}.

  To design, explain, and draw solutions of firing squad/mob synchronisation problems, it is convenient to think about, talk about, and draw continuous space-time diagrams of different kinds of signals that move across the cell space, vanish or give rise to new signals upon reaching boundaries or junctions of the space or upon colliding with each other. This is mostly done in an informal way, but the idea of signals has also been formalised for one-dimensional cellular automata by Jérôme Olivier Durand-Lose in 2005 (see \cite{durand-lose:2005}).

  This formalisation however does not handle accumulations of events like collisions and does not allow infinitely many signals of different speeds, which naturally occur and are necessary in descriptions of many solutions of the firing squad synchronisation problem by signals. For example, collisions accumulate at the time synchronisation finishes and infinitely many signals of different speeds may originate from the general. In the time evolutions of the actual cellular automata, the accumulations of collisions disappear due to the discreteness of space and time, and the infinitely many signals are cleverly produced by finitely many states (see for example \cite{mazoyer:1987}).

  Because we want to describe our quasi-solution in terms of signals in a formal way, we first introduce continuum representations of finite and connected multigraphs (without self-loops), we secondly introduce signal machines over such representations that allow infinitely many signals of different speeds and seamlessly handle accumulations of events and accumulations of accumulations of events and so forth, and we thirdly construct a signal machine for the continuous firing mob synchronisation problem over such representations and shortly note how to discretise it to get a cellular automaton quasi-solution of the firing mob synchronisation problem.

  \paragraph{Contents.} In \cref{section:firing-squad-problem} we state the firing squad and the firing mob synchronisation problems. In \cref{section:undirected-multigraphs} we introduce undirected multigraphs (without self-loops) and direction-preserving paths in such graphs, which are paths that do not make U-turns. In \cref{section:continuum-representation} we introduce continuum representations of undirected multigraphs, which are in a sense drawings of graphs in a high-dimensional Euclidean space. In \cref{section:signal-machines} we introduce signal machines, which can be studied in their own right, but which can also be thought of as high-level views of time evolutions of cellular automata over graphs, like cellular automata over finitely right generated cell spaces, that are restricted to configurations with a fixed finite support. In \cref{section:firing-squad-solution} we construct a signal machine whose discretisation is a cellular automaton that quasi-solves the firing mob synchronisation problem in $(r + d)$-many steps. And in \cref{section:proof-sketch} we sketch a proof for that statement. The impatient may right now have a look at the continuous space-time diagrams of the synchronisations of small trees as performed by the quasi-solution: See \cref{figure:MazoyerWithOneAndTwoEdges,figure:fsspWithTwoEdgesAndGeneralInBetweenAndAtTheLeft,figure:fsspWithThreeEdgesInARowAndGeneralAtTheSecondVertexFromTheLeft,figure:fsspWithThreeEdgesIncidentToTheSameVertexAndGeneralAtTheVertexAndAtTheLeafOfTheShortestEdge} on \cpageref{figure:MazoyerWithOneAndTwoEdges,figure:fsspWithTwoEdgesAndGeneralInBetweenAndAtTheLeft,figure:fsspWithThreeEdgesInARowAndGeneralAtTheSecondVertexFromTheLeft,figure:fsspWithThreeEdgesIncidentToTheSameVertexAndGeneralAtTheVertexAndAtTheLeafOfTheShortestEdge}.

  \paragraph{Preliminary Notions.} The affinely extended real numbers\graffito{affinely extended real numbers $\overline{\R}$} $\R \cup \setOf{-\infty, +\infty}$ are denoted by $\overline{\R}$\index[symbols]{Roverline@$\overline{\R}$}. For each tuple $(r, r') \in \overline{\R} \times \overline{\R}$ such that $r \leq r'$, the closed, open, and the two half-open extended real intervals\graffito{extended real intervals $\closedInterval{r, r'}$, $\openInterval{r, r'}$, and $[r, r'[$ and $]r, r']$} with the endpoints $r$ and $r'$ are denoted by $\closedInterval{r, r'}$, $\openInterval{r, r'}$, and $[r, r'[$ and $]r, r']$ respectively. And, for each tuple $(z, z') \in \Z \times \Z$ such that $z \leq z'$, the closed integer-valued interval\graffito{closed integer-valued interval $\discreteInterval{z}{z'}$} with the endpoints $z$ and $z'$, namely $\setOf{z, z + 1, \dotsc, z'}$ or equivalently $\closedInterval{z, z'} \cap \Z$, is denoted by $\discreteInterval{z}{z'}$. 

  \section{The Firing Squad/Mob Synchronisation Problems}
  \label{section:firing-squad-problem}

  We formally state the problems for generalised cellular automata as introduced in \cite{wacker:automata:2016} over spaces as introduced in \cite{wacker:growth:2017}. However, you do not need to be familiar with these automata and spaces to understand the problems on an intuitive level or to understand the quasi-solutions in detail. Such an automaton is essentially a cellular automaton over a vertex-transitive graph, which means that each vertex of the graph is a copy of a fixed possibly-infinite-state machine, its inputs are the states of the vertex's neighbours, and all vertices change their states synchronously.

  In this section, let $\mathcal{R} = \ntuple{\ntuple{M, G, \actsOnPoint}, \ntuple{m_0, \family{g_{m_0, m}}_{m \in M}}}$ be a finitely right generated cell space, let $N$ be a finite right generating set of $\mathcal{R}$ that contains $G_0$, where $G_0$ is the stabiliser of $m_0$ under $\actsOnPoint$, let $\mathcal{G} = \ntuple{M, E}$ be the coloured $N$-Cayley graph of $\mathcal{R}$, let $\mathcal{C}$ be a semi-cellular automaton over $\mathcal{R}$ with state set $Q$, neighbourhood $N$, and local transition function $\delta$, and let $\Delta$ be the global transition function of $\mathcal{C}$. 

  To state the problems succinctly we introduce the notions of passive subsets of states, dead states, supports of global configurations with respect to a distinguished dead state, and what it means for a global configuration to be of the form of a pattern in the following four definitions.

  \begin{definition}
    Let $P$ be a subset of $Q$. It is called \define{passive}\graffito{passive set of states} if and only if, for each local configuration $\ell \in Q^N$ with $\imageOf(\ell) \subseteq P$, we have $\delta(\ell) = \ell(G_0)$.
  \end{definition}

  \begin{definition}
    Let $q$ be a state of $Q$. It is called \define{dead}\graffito{dead state} if and only if, for each local configuration $\ell \in Q^N$ with $\ell(G_0) = q$, we have $\delta(\ell) = q$.
  \end{definition}

  In the remainder of this section, let $Q$ contain a distinguished dead state named $\deadState$.

  \begin{definition}
    Let $c$ be a global configuration of $Q^M$. The set $\supportOf(c) = M \smallsetminus c^{-1}(\deadState)$ is called \define{support of $c$}\graffito{support $\supportOf(c)$ of $c$}\index[symbols]{suppc@$\supportOf(c)$}.
  \end{definition}

  \begin{definition} 
    Let $A$ be a subset of $M$, let $p$ be a pattern of $Q^A$, and let $c$ be a global configuration of $Q^M$. The global configuration $c$ is said to \define{be of the form $p$}\graffito{global configuration is of the form $p$} if and only if there is an element $g \in G$ such that $c\restrictedTo_{g \actsOnPoint A} = g \actsOnMap p$ and $c\restrictedTo_{M \smallsetminus (g \actsOnPoint A)} \equiv \deadState$. 
  \end{definition}


  We state the firing squad synchronisation problem in

  \begin{definition} 
    Let $\deadState$, $\generalState$, $\soldierState$, and $\fireState$ be four distinct states, and let $Q'$ be the set that consists of those states. A solution of the \graffito{firing squad synchronisation problem in dimension one with one general at the left end}\define{firing squad synchronisation problem in dimension one with one general at the left end} is a cellular automaton $\mathcal{C}$ over $\ntuple{\ntuple{\Z, \Z, +}, \ntuple{0, \family{z}_{z \in \Z}}}$ with neighbourhood $\setOf{-1, 0, 1}$ and finite set of states that includes $Q'$ such that the state $\deadState$ is dead and the set $\setOf{\deadState, \soldierState}$ is passive, and whose global transition function $\Delta$ has the following property:

    For each global configuration $c$ with finite support of the form $\generalState \soldierState \soldierState \dotsb \soldierState$, there is a non-negative integer $k$ such that the global configuration $\Delta^k(c)$ is of the form $\fireState \fireState \dotsb \fireState$ and has the same support as $c$, and such that the state $\fireState$ does not occur in any of the global configurations $\Delta^j(c)$, for $j \in \N_0$ with $j < k$.
  \end{definition}

  \begin{remark}
    Let $\mathcal{C}$ be a solution of the above problem, let $c$ be a global configurations of the form $\generalState \soldierState \soldierState \dotsb \soldierState$, and let $k$ be the non-negative integer from the problem definition. Then, because the state $\deadState$ is dead and the support of $\Delta^k(c)$ is the same as the one of $c$, for each non-negative integer $j$ with $j \leq k$, the support of $\Delta^j(c)$ is the same as the one of $c$. Broadly speaking, in the time evolution of solutions, the support of initial configurations can neither shrink nor grow before synchronisation is finished. Moreover, because the set $\setOf{\deadState, \soldierState}$ is passive, if the support of $c$ consists of at least $3$ cells, then $\Delta(c)$ cannot be of the form $\fireState \fireState \dotsb \fireState$. Broadly speaking, the problem cannot be solved trivially.
  \end{remark}

  \begin{remark}
  \label{rem:one-dimensional-array}
    As mentioned above, for each global configuration $c$ of the form $\generalState \soldierState \soldierState \dotsb \soldierState$, the supports of the global configurations that are observable in the time evolutions that begin in the configuration $c$ of cellular automata that solve the above problem, are included in the support of $c$. Hence, we can regard such cellular automata as automata over one-dimensional arrays with one dummy neighbour in the state $\deadState$ at each end.
  \end{remark}

  \begin{remark}
    The above problem can be generalised in many ways. For example, by allowing the general to be placed anywhere or by allowing more than one general.
  \end{remark}

  We state the firing mob synchronisation problem in

  \begin{definition} 
    Let $\deadState$, $\generalState$, $\soldierState$, and $\fireState$ be four distinct states, and let $Q'$ be the set that consists of those states. A solution of the \graffito{firing mob synchronisation problem in $\mathcal{R}$ with respect to $S$}\define{firing mob synchronisation problem in $\mathcal{R}$ with respect to $S$} is a semi-cellular automaton over $\mathcal{R}$ with neighbourhood $S$ and finite set of states that includes $Q'$ such that the state $\deadState$ is dead and the set $\setOf{\deadState, \soldierState}$ is passive, and whose global transition function $\Delta$ has the following property:

    For each finite subset $A$ of $M$ such that the subgraph $\mathcal{G}[A]$ of $\mathcal{G}$ induced by $A$ is connected, each element $a \in A$, each pattern $p \in Q^A$ such that $p(a) = \generalState$ and $p\restrictedTo_{A \smallsetminus \setOf{a}} \equiv \soldierState$, and each global configuration $c$ of the form $p$, there is a non-negative integer $k$ such that the global configuration $\Delta^k(c)$ is of the form $A \to Q$, $a \mapsto \fireState$, and such that the state $\fireState$ does not occur in any of the global configurations $\Delta^j(c)$, for $j \in \N_0$ with $j < k$.
  \end{definition}

  \begin{remark}
    The firing squad synchronisation problem with one general at an arbitrary position is the firing mob synchronisation problem in $\ntuple{\ntuple{\Z, \Z, +}, \ntuple{0, \family{z}_{z \in \Z}}}$ with respect to $\setOf{-1, 0, 1}$. Note that the notions of semi-cellular and cellular automata are identical over $\ntuple{\Z, \Z, +}$.
  \end{remark}

  \begin{remark}
    Each semi-cellular automaton over $\mathcal{R}$ with neighbourhood $S$ is equivalent to a cellular automaton over the coloured $S$-Cayley graph of $\mathcal{R}$ acted upon by its automorphism group, in the sense that, for each of the former kind of automata, there is one of the latter kind with the same global transition function, and vice versa. Note that the stabilisers of coloured $S$-Cayley graphs of $\mathcal{R}$ are trivial, and hence the notions of semi-cellular and cellular automata are identical over such graphs.
  \end{remark}

  \begin{remark}
    We can regard semi-cellular automata that solve the above problem as semi-cellular automata over subgraphs of $\mathcal{G}$ that are induced by finite subsets of $M$ with one dummy neighbour in the state $\deadState$ at each edge that leads out of the subgraph. Note that, because the graph $\mathcal{G}$ is of bounded degree, the maximum degrees of the subgraphs it induces are uniformly bounded by a constant.
  \end{remark}

  \begin{remark}
    Ideally we would like an abstract description of a semi-cellular automaton that does not depend on any specifics of $\mathcal{R}$ and $S$ and that yields a solution for each choice of $\mathcal{R}$ and $S$ or at least for as huge a class of such choices as possible.
  \end{remark}

  \section{Undirected Multigraphs}
  \label{section:undirected-multigraphs}

  Undirected multigraphs without self-loops are introduced in

  \begin{definition} 
    Let $\Vertices$ and $\Edges$ be two disjoint sets, and let $\eendsOf$ be a map from $\Edges$ to $\setOf{\setOf{v, v'} \subseteq \Vertices \suchThat v \neq v'}$. The triple $\Graph = \ntuple{\Vertices, \Edges, \eendsOf}$ is called \define{undirected multigraph}\graffito{undirected multigraph $\Graph = \ntuple{\Vertices, \Edges, \eendsOf}$}\index[symbols]{Gcalligraphic@$\Graph$}\index[symbols]{V@$\Vertices$}\index[symbols]{E@$\Edges$}; each element $v \in \Vertices$ is called \define{vertex}\graffito{vertex $v$}\index[symbols]{v@$v$}; each element $e \in \Edges$ is called \define{edge}\graffito{edge $e$}\index[symbols]{e@$e$}; and, for each edge $e \in \Edges$, each vertex of $\eendsOf(e)$ is called \define{end of $e$}\graffito{ends $\eendsOf(e)$ of $e$}\index[symbols]{epsilonvare@$\eendsOf(e)$}.
  \end{definition}

  \begin{remark}
    Because each set in the codomain of $\eendsOf$ consists of exactly two distinct vertices, there are no self-loops in the undirected multigraph $\Graph$. With minor modifications the theory and the automata presented in this chapter also work if there are self-loops. They were merely excluded to make the presentation a little simpler. 
  \end{remark}

  In the remainder of this section, let $\Graph = \ntuple{\Vertices, \Edges, \eendsOf}$ be an undirected multigraph.

  What being finite means for multigraphs is said in

  \begin{definition}
    The multigraph $\Graph$ is called \define{finite}\graffito{finite multigraph} if and only if the sets $\Vertices$ and $\Edges$ are both finite.
  \end{definition}

  Isolated vertices are the ones without incident edges as introduced in

  \begin{definition}
    Let $v$ be a vertex of $\Graph$. It is called \define{isolated}\graffito{isolated vertex} if and only if, for each edge $e \in \Edges$, we have $v \notin \eendsOf(e)$.
  \end{definition}

  Directed edges are edges with distinguished source and target vertices as introduced in

  \begin{definition} 
    Let $e$ be an edge of $\Graph$, and let $v_1$ and $v_2$ be two vertices of $\Graph$ such that $\setOf{v_1, v_2} = \eendsOf(e)$. The triple $\direct{e} = (v_1, e, v_2)$ is called \define{directed edge from $v_1$ through $e$ to $v_2$}\graffito{directed edge $\direct{e}$ from $v_1$ through $e$ to $v_2$}\index{edge!directed}\index[symbols]{earrowontop@$\direct{e}$}; the vertex $\sourceOf(\direct{e}) = v_1$ is called \define{source of $\direct{e}$}\graffito{source $\sourceOf(\direct{e})$ of $\direct{e}$}\index[symbols]{sigmaearrowontop@$\sourceOf(\direct{e})$}; the edge $\bedOf(\direct{e}) = e$ is called \define{bed of $\direct{e}$}\graffito{bed $\bedOf(\direct{e})$ of $\direct{e}$}\index[symbols]{betaearrowontop@$\bedOf(\direct{e})$}; and the vertex $\targetOf(\direct{e}) = v_2$ is called \define{target of $\direct{e}$}\graffito{target $\targetOf(\direct{e})$ of $\direct{e}$}\index[symbols]{tauearrowontop@$\targetOf(\direct{e})$}.
  \end{definition}

  At each vertex there is an empty path that starts and ends at the vertex, and non-empty paths are concatenations of directed edges with matching source and target vertices as introduced in

  \begin{definition}
    \begin{aenumerate}
      \item Let $v$ be a vertex of $\Graph$. The singleton $p = (v)$ is called \define{empty path in $v$}\graffito{empty path $(v)$ in $v$}\index{path!empty}\index[symbols]{vleftparenrightparen@$(v)$}, the vertex $\sourceOf(p) = \targetOf(p) = v$ is called \defineX{source}{source of $p$}\index[symbols]{sigmap@$\sourceOf(p)$} and \define{target of $p$}\graffito{source $\sourceOf((v))$ and target $\targetOf((v))$ of $(v)$}\index[symbols]{taup@$\targetOf(p)$}, and the non-negative integer $\lengthOfPath{p} = 0$ is called \define{length of $p$}\graffito{length $\lengthOfPath{(v)}$ of $(v)$}\index[symbols]{absolutep@$\lengthOfPath{p}$}.
      \item Let $n$ be a positive integer and, for each index $i \in \discreteInterval{1}{n}$, let $\direct{e}_i$ be a directed edge of $\Graph$ such that, if $i \neq 1$, then $\sourceOf(\direct{e}_i) = \targetOf(\direct{e}_{i - 1})$. The $(2n + 1)$-tuple $p = (\sourceOf(\direct{e}_1), \bedOf(\direct{e}_1), \targetOf(\direct{e}_1), \dotsc, \bedOf(\direct{e}_n), \targetOf(\direct{e}_n))$ is called \define{path from $\sourceOf(\direct{e}_1)$ to $\targetOf(\direct{e}_n)$}\graffito{path $p$ from $\sourceOf(\direct{e}_1)$ to $\targetOf(\direct{e}_n)$}\index[symbols]{p@$p$}; the vertex $\sourceOf(p) = \sourceOf(\direct{e}_1)$ is called \define{source of $p$}\graffito{source $\sourceOf(p)$ of $p$}\index[symbols]{sigmap@$\sourceOf(p)$}; the vertex $\targetOf(p) = \targetOf(\direct{e}_n)$ is called \define{target of $p$}\graffito{target $\targetOf(p)$ of $p$}\index[symbols]{taup@$\targetOf(p)$}; and the positive integer $\lengthOfPath{p} = n$ is called \define{length of $p$}\graffito{length $\lengthOfPath{p}$ of $p$}\index[symbols]{absolutep@$\lengthOfPath{p}$}.
      \item The set of paths is denoted by $\Paths$\graffito{set $\Paths$ of paths}\index[symbols]{P aths@$\Paths$}. \qedhere
    \end{aenumerate}
  \end{definition}

  \begin{remark}
    Each directed edge is a path of length $1$.
  \end{remark}

  Subpaths are connected parts of paths as introduced in

  \begin{definition}
    Let $p = (v_0, e_1, v_1, \dotsc, e_n, v_n)$ be a path of $\Graph$, and let $k$ and $\ell$ be two indices of $\discreteInterval{0}{n}$ such that $k \leq \ell$. The path $(v_k)$, if $k = \ell$, or $(v_k, e_{k + 1}, v_{k + 1}, \dotsc, e_\ell, v_\ell)$, otherwise, is called \define{subpath of $p$}\graffito{subpath $p_{\discreteInterval{k}{\ell}}$ of $p$} and is denoted by $p_{\discreteInterval{k}{\ell}}$\index[symbols]{pklsubscript@$p_{\discreteInterval{k}{\ell}}$}.
  \end{definition}

  Direction-preserving paths are the ones without U-turns as introduced in

  \begin{definition}
    Let $p = (v_0, e_1, v_1, \dotsc, e_n, v_n)$ be a path of $\Graph$. It is called \define{direction-preserving}\graffito{direction-preserving} if and only if, for each index $i \in \discreteInterval{1}{n - 1}$, we have $e_i \neq e_{i + 1}$. The \graffito{set $\Paths_{\directionPreserving}$ of direction-preserving paths}set of direction-preserving paths is denoted by $\Paths_{\directionPreserving}$\index[symbols]{Parrowrightshortsubscript@$\Paths_{\directionPreserving}$}.
  \end{definition}


  Two paths with matching target and source vertices can be concatenated as introduced in

  \begin{definition}
    Let $p = (v_0, e_1, v_1, \dotsc, e_n, v_n)$ and $p' = (v_0', e_1', v_1',\allowbreak \dotsc, e_{n'}', v_{n'}')$ be two paths of $\Graph$ such that $v_n = v_0'$. The path $p \concat p' = (v_0, e_1, v_1, \dotsc, e_n, v_n, e_1', v_1', \dotsc, e_{n'}', v_{n'}')$ is called \graffito{concatenation of $p$ and $p'$}\define{concatenation $p \concat p'$ of $p$ and $p'$}\index[symbols]{bullet@$\bullet$}.
  \end{definition}

  \begin{remark}
    Each non-empty path is the concatenation of directed edges.
  \end{remark}

  What being connected means for multigraphs is said in

  \begin{definition}
    The multigraph $\Graph$ is called \define{connected}\graffito{connected multigraph} if and only if, for each tuple $(v, v') \in \Vertices \times \Vertices$, there is a path from $v$ to $v'$.
  \end{definition}

  \begin{remark}
    The multigraph $\Graph$ is connected if and only if, for each tuple $(v, v') \in \Vertices \times \Vertices$, there is a direction-preserving path from $v$ to $v'$.
  \end{remark}


  We can assign weights to edges as done in

  \begin{definition} 
    Let $\weightOf$ be a map from $\Edges$ to $\R_{> 0}$. It is called \graffito{edge weighting $\weightOf$ of $\Graph$}\define{edge weighting of $\Graph$}\index[symbols]{omega@$\weightOf$}, and, for each edge $e \in E$, the element $\weightOf(e)$ is called \define{edge weight of $e$}\graffito{edge weight $\weightOf(e)$ of $e$}\index[symbols]{omegae@$\weightOf(e)$}.
  \end{definition}


  Edge weights induce weights of paths as introduced in

  \begin{definition}
    Let $\weightOf$ be an edge weighting of $\Graph$ and let $p = (v_0, e_1, v_1, \dotsc, e_n, v_n)$ be a path of $\Graph$. The sum $\weightOf(p) = \sum_{i = 1}^n \weightOf(e_i)$ is called \define{weight of $p$}\graffito{weight $\weightOf(p)$ of $p$}\index[symbols]{omegap@$\weightOf(p)$}. 
  \end{definition}

  \begin{remark}
    Each empty path has weight $0$.
  \end{remark}

  \begin{remark}
    Each directed edge has the same weight as its bed.
  \end{remark}

  \section{Continuum Representation}
  \label{section:continuum-representation}

  In this section, let $\Graph = \ntuple{\Vertices, \Edges, \eendsOf}$ be an undirected multigraph without isolated vertices and let $\weightOf$ be an edge weighting of $\Graph$. 

  An orientation is a choice of source and target vertices for each edge as introduced in

  \begin{definition}
    Let $\sigma$ and $\tau$ be two maps from $\Edges$ to $\Vertices$ such that, for each edge $e \in \Edges$, we have $\setOf{\sigma(e), \tau(e)} = \eendsOf(e)$. The tuple $(\sigma, \tau)$ is called \define{orientation of $\Graph$}\graffito{orientation $(\sigma, \tau)$ of $\Graph$}\index[symbols]{sigmatautuple@$(\sigma, \tau)$}.
  \end{definition}

  Realising weighted edges as disjoint intervals and gluing these intervals together at shared ends yields a continuum representation of $\Graph$ and is done in

  \begin{definition}
  \label{definition:continuum-representation}
    Let $(\sigma, \tau)$ be an orientation of $\Graph$, let 
    \begin{align*}
      \zeta \from \R \smallsetminus \setOf{0} &\to \setOf{\sigma, \tau}, \mathnote{map $\zeta$ from $\R \smallsetminus \setOf{0}$ to $\setOf{\sigma, \tau}$}\index[symbols]{zeta@$\zeta$}\\
      r &\mapsto \begin{dcases*}
                   \sigma, &if $r < 0$,\\
                   \tau,   &if $r > 0$,
                 \end{dcases*}
    \end{align*}
    let
    \begin{equation*}
      \left\{
        \begin{aligned}
          \phi \from \Edges &\to \R_{< 0}, \index[symbols]{phi@$\phi$}\\
          e &\mapsto - \frac{\weightOf(e)}{2},
        \end{aligned}
      \right\}
      \text{ and }
      \left\{
        \begin{aligned}
          \psi \from \Edges &\to \R_{> 0}, \index[symbols]{psi@$\psi$}\\
          e &\mapsto \frac{\weightOf(e)}{2},
        \end{aligned}
      \right\}
      \mathnote{maps $\phi$ and $\psi$ from $\Edges$ to $\R_{< 0}$ and $\R_{> 0}$}
    \end{equation*}
    and let $\sim$\graffito{equivalence relation $\sim$ on $\R \times \Edges$}\index[symbols]{tilde@$\sim$} be the equivalence relation on $\R \times \Edges$ such that, for each tuple $(r, e) \in \R \times \Edges$ and each tuple $(r', e') \in \R \times \Edges$, 
    \begin{align*}
      (r, e) \sim (r', e') \ifAndOnlyIf
      &r \in \setOf{\phi(e), \psi(e)}\\
      &\land r' \in \setOf{\phi(e'), \psi(e')}\\
      &\land \zeta(r)(e) = \zeta(r')(e').
    \end{align*}
    The set $\continuumRepresentationOf{\Graph} = (\bigcup_{e \in \Edges} \closedInterval{\phi(e), \psi(e)} \times \setOf{e}) \modulo {\sim}$ is called \graffito{continuum representation $\continuumRepresentationOf{\Graph}$ of $\Graph$}\define{continuum representation of $\Graph$}\index[symbols]{Gcalligraphicbar@$\continuumRepresentationOf{\Graph}$}.
  \end{definition} 

  \begin{remark}
    Each weighted edge is realised as a closed interval whose length is the edge's weight. These intervals are made disjoint by taking the Cartesian product with the respective edge. And they are glued together at shared ends by taking the set of all these disjoint intervals modulo the equivalence relation $\sim$. The vertices are implicitly realised as end points or junctions of the glued disjoint intervals.
  \end{remark}

  \begin{remark}
    If the graph $\Graph$ contained isolated vertices, then they would not be represented in $\continuumRepresentationOf{\Graph}$. With minor modifications the theory presented in this chapter also works if there are isolated vertices. They were merely excluded to make the presentation a little simpler.
  \end{remark}

  In the remainder of this section, let $\continuumGraph$ be a continuum representation of $\Graph$ with respect to an orientation $(\sigma, \tau)$, and let $\zeta$, $\phi$, $\psi$, and $\sim$ be the maps and the equivalence relation from \cref{definition:continuum-representation}.

  Vertices are canonically embedded into $\continuumGraph$ as is done in

  \begin{definition}
    The map
    \begin{align*}
      \continuumRepresentationOf{\phantom{v}} \from V &\to \continuumGraph, \mathnote{vertex embedding $\continuumRepresentationOf{\phantom{v}}$}\index[symbols]{vbar@$\continuumRepresentationOf{v}$}\\
      \sourceOf(e) &\mapsto \equivalenceClassOf{(\phi(e), e)}_\sim,\\
      \targetOf(e) &\mapsto \equivalenceClassOf{(\psi(e), e)}_\sim,
    \end{align*}
    embeds vertices of $\Graph$ into $\continuumGraph$. Its image is denoted by $\continuumVertices$\graffito{vertices $\continuumVertices$}\index[symbols]{Vfraktur@$\continuumVertices$} and each element $\mathfrak{v} \in \continuumVertices$ is called \define{vertex}\graffito{vertex $\mathfrak{v}$}\index[symbols]{vfraktur@$\mathfrak{v}$}. 
  \end{definition}

  \begin{remark}
    The embedding is well-defined due to the definition of the equivalence relation $\sim$. 
  \end{remark}

  Edges are canonically embedded into the power set of $\continuumGraph$ as is done in

  \begin{definition}
    The map
    \begin{align*}
      \continuumRepresentationOf{\phantom{e}} \from E &\to \powerSetOf(\continuumGraph), \mathnote{edge embedding $\continuumRepresentationOf{\phantom{e}}$}\index[symbols]{ebar@$\continuumRepresentationOf{e}$}\\
      e &\mapsto ([\phi(e), \psi(e)] \times \setOf{e}) \modulo {\sim},
    \end{align*}
    embeds edges of $\Graph$ into $\continuumGraph$. Its image is denoted by $\continuumEdges$\graffito{edges $\continuumEdges$}\index[symbols]{Efraktur@$\continuumEdges$} and each element $\mathfrak{e} \in \continuumEdges$ is called \define{edge}\graffito{edge $\mathfrak{e}$}\index[symbols]{efraktur@$\mathfrak{e}$}. 
  \end{definition}

  At each point of $\continuumGraph$ there is at least one direction to move: In a vertex of degree $k$, there are $k$ directions; and on an edge but not in one of its endpoints, there are $2$ directions. An inefficient but immediate way to represent these directions is as in

  \begin{definition}
    The set $\setOf{-1, 1} \times \Edges$ is denoted by $\Directions$\graffito{set $\Directions$ of directions}\index[symbols]{Dir@$\Directions$}, each element $d = (o, e) \in \Directions$ is called \defineX{direction on $e$}{direction}\graffito{direction $d$ on $e$}\index[symbols]{direction@$d$}, the element $o$ is called \define{orientation of $d$}\graffito{orientation $o$ of $d$}\index[symbols]{orientation@$o$}, the involution
    \begin{align*}
      \reverse \from \Directions &\to \Directions, \mathnote{orientation reversing involution $\reverse$}\index[symbols]{minus@$\reverse$}\\
      (o, e) &\mapsto (-o, e),
    \end{align*}
    reverses the orientation of directions, and the map
    \begin{align*}
      \directionOf \from \continuumGraph &\to \powerSetOf(\Directions), \mathnote{map $\directionOf$ that assigns directions}\index[symbols]{direction@$\directionOf$}\\
      \equivalenceClassOf{(r, e)}_\sim &\mapsto
        \left\{\begin{aligned} 
          &\setOf{(-1, e), (1, e)}, \text{ if $r \in \openInterval{\phi(e), \psi(e)}$},\\ 
          &\setOf{(- \signOf(r'), e') \suchThat (r', e') \in \equivalenceClassOf{(r, e)}_\sim}, \text{ otherwise}, 
        \end{aligned}\right.
    \end{align*}
    assigns to each point in $\continuumGraph$ the set of possible directions in which someone standing on that point can move.
  \end{definition}

  \begin{remark}
  \label{remark:efficient-representation-of-directions}
    This representation of directions is inefficient in the following sense: If we stand on an edge but not on one of its endpoints, then the orientation is enough directional information; and if we stand on a vertex, then the edge is enough directional information, because the orientation is implicit in the fact that we can only move onto the edge but not off it since we are in one end of the edge. On a vertex we do not even need the edge itself but only an identifier for the edge that is locally unique; for example, we could colour the edges such that no two edges of the same colour are incident to the same vertex and use this colour instead.
  \end{remark} 

  Like vertices, paths of $\Graph$ are also canonically embedded into $\continuumGraph$ and each embedding can be unit-speed parametrised by the interval from $0$ to the path's weight as is inductively done in

  \begin{definition}
    The map 
    \begin{align*}
      \continuumRepresentationOf{\phantom{p}} \from \Paths &\to \continuumGraph^{\setOf{\closedInterval{0, r} \suchThat r \in \R_{\geq 0}}}, \mathnote{path embedding $\continuumRepresentationOf{\phantom{p}}$}\index[symbols]{pbar@$\continuumRepresentationOf{p}$}\\
      (v) &\mapsto \left[
                     \begin{aligned}
                       \closedInterval{0, 0} &\to \continuumGraph,\\
                       r &\mapsto \continuumRepresentationOf{v}
                     \end{aligned}
                   \right]\\
      (\sourceOf(e), e, \targetOf(e)) &\mapsto \left[
                                             \begin{aligned}
                                               \closedInterval{0, \weightOf(e)} &\to \continuumGraph,\\
                                               r &\mapsto \equivalenceClassOf{(\phi(e) + r, e)}_\sim,
                                             \end{aligned}
                                           \right]\\
      (\targetOf(e), e, \sourceOf(e)) &\mapsto \left[
                                             \begin{aligned}
                                               \closedInterval{0, \weightOf(e)} &\to \continuumGraph,\\
                                               r &\mapsto \equivalenceClassOf{(\psi(e) - r, e)}_\sim,
                                             \end{aligned}
                                           \right]\\
      (v_0, e_1, v_1) \concat p' &\mapsto \left[ 
                                            \begin{aligned}
                                              &\closedInterval{0, \weightOf((v_0, e_1, v_1) \concat p')} \to \continuumGraph,\\
                                              &r \mapsto \begin{dcases*}
                                                           \continuumRepresentationOf{(v_0, e_1, v_1)}(r), &if $r \leq \weightOf(e_1)$,\\
                                                           \continuumRepresentationOf{p'}(r - \weightOf(e_1)), &otherwise,
                                                         \end{dcases*} 
                                            \end{aligned}
                                          \right]
    \end{align*}
    maps paths of $\Graph$ to unit-speed parametrisations of them in $\continuumGraph$.
  \end{definition}

  \begin{remark}
    The base cases of the inductive definition do not overlap because there are no self-loops, and the inductive step is well-defined because $\weightOf((v_0, e_1, v_1) \concat p') = \weightOf(e_1) + \weightOf(p')$.
  \end{remark}

  The images $\continuumRepresentationOf{\Paths}$ and $\continuumRepresentationOf{\Paths_{\directionPreserving}}$ consist broadly speaking of paths and direction-preserving paths in $\continuumGraph$ from vertices to vertices that only change direction at vertices. Restricting the parametrisation intervals of paths in $\continuumRepresentationOf{\Paths_{\directionPreserving}}$ to subintervals and doing a reparametrisation such that the new parametrisation starts at $0$ yields all direction-preserving paths in $\continuumGraph$ and is done in

  \begin{definition}
    The set
    \begin{equation*}
      \setOf{\continuumRepresentationOf{p}\restrictedTo_{\closedInterval{r, s}}(\blank + r) \suchThat p \in \Paths_{\directionPreserving} \text{ and } r, s \in \closedInterval{0, \weightOf(p)} \text{ with } r \leq s}
    \end{equation*} 
    is denoted by $\continuumPaths_{\directionPreserving}$\graffito{set $\continuumPaths_{\directionPreserving}$ of direction-preserving paths $\mathfrak{p}$}\index[symbols]{Parrowrightsubscriptfraktur@$\continuumPaths_{\directionPreserving}$}; each element $\mathfrak{p} \in \continuumPaths_{\directionPreserving}$ is called \define{direction-preserving path}\index[symbols]{pfraktur@$\mathfrak{p}$}, the length of the interval $\domainOf(\mathfrak{p})$ is called \define{length of $\mathfrak{p}$}\graffito{length $\length(\mathfrak{p})$ of $\mathfrak{p}$} and is denoted by $\length(\mathfrak{p})$\index[symbols]{omegapfraktur@$\length(\mathfrak{p})$}, the point $\sourceOf(\mathfrak{p}) = \mathfrak{p}(0)$ is called \define{source of $\mathfrak{p}$}\graffito{source $\sourceOf(\mathfrak{p})$ of $\mathfrak{p}$}\index[symbols]{sigmapfraktur@$\sourceOf(\mathfrak{p})$}, the point $\targetOf(\mathfrak{p}) = \mathfrak{p}(\length(\mathfrak{p}))$ is called \define{target of $\mathfrak{p}$}\graffito{target $\targetOf(\mathfrak{p})$ of $\mathfrak{p}$}\index[symbols]{taupfraktur@$\targetOf(\mathfrak{p})$}, and the path $\mathfrak{p}$ is called \define{empty}\graffito{empty path} if and only if $\length(\mathfrak{p}) = 0$.
  \end{definition} 

  \begin{remark}
    Doing the same with the paths in $\continuumRepresentationOf{\Paths}$ does not yield all paths in $\continuumGraph$ but only those that change direction at vertices and not on edges. Because we only need direction-preserving paths in what is to come, we do not define what a general path on $\continuumGraph$ is.
  \end{remark}

  \begin{remark}
    Sources and targets of direction-preserving paths in $\continuumGraph$ are in general not vertices.
  \end{remark}

  The distance between two points is the length of the shortest path between the points as introduced in

  \begin{definition}
    The map
    \begin{align*}
      \distanceOf \from \continuumGraph \times \continuumGraph &\to \R_{\geq 0} \cup \setOf{\infty}, \mathnote{distance $\distanceOf$}\index[symbols]{d@$\distanceOf$}\\
      (\mathfrak{m}, \mathfrak{m}') &\mapsto \inf\setOf{\length(\mathfrak{p}) \suchThat \mathfrak{p} \in \continuumPaths_{\directionPreserving} \text{ with } \sourceOf(\mathfrak{p}) = \mathfrak{m} \text{ and } \targetOf(\mathfrak{p}) = \mathfrak{m}'} 
    \end{align*}
    is called \define{distance}, where the infimum of the empty set is infinity.
  \end{definition}

  \begin{remark}
    If the graph $\Graph$ is finite and connected, then the distance map $\distanceOf$ is a metric. Otherwise, it may not be a metric. For example, if there are two distinct vertices $v$, $v' \in \Vertices$ such that, for each $n \in \N_+$, there is an edge $e \in \Edges$ whose weight is $1/n$, then the distance of $\continuumRepresentationOf{v}$ and $\continuumRepresentationOf{v'}$ is $0$ although $\continuumRepresentationOf{v} \neq \continuumRepresentationOf{v'}$. Or, if the graph $\Graph$ is not connected, then there are two points $\mathfrak{m}$, $\mathfrak{m}' \in \continuumGraph$ whose distance is $\infty$.
  \end{remark}

  Each non-zero vector of a vector space is uniquely determined by its magnitude and its direction, and the zero vector is already uniquely determined by its magnitude, which is $0$, and can be thought of as pointing in every direction, which can be represented by the set of directions. A generalisation of vector spaces is given in

  \begin{definition} 
    Let $\every$ be the set $\Directions$. The set
    \begin{equation*}
      \Arrows = \setOf{(0, \every)} \cup (\R_{> 0} \times \Directions) \mathnote{arrow space $\Arrows$}\index[symbols]{Arr@$\Arrows$}
    \end{equation*}
    is called \define{arrow space}; each element $a \in \Arrows$ is called \define{arrow}\graffito{arrow $a$}\index[symbols]{a@$a$}; the set $\every$ is called \define{semi-direction}\graffito{semi-direction $\every$}\index{direction!semi-}\index[symbols]{vry@$\every$}; for each element $a = (r, d) \in \Arrows$, the real number $\normOf{a} = r$ is called \define{magnitude of $a$}\graffito{magnitude $\normOf{a}$ of $a$}\index[symbols]{norma@$\normOf{a}$}, and the (semi-)direction $\directionOf(a) = d$ is called \defineX{(semi-)direction of $a$}{direction of $a$ semi@(semi-)direction of $a$}\graffito{(semi-)direction $\directionOf(a)$ of $a$}\index{semi-direction of $a$@(semi-)direction of $a$}\index[symbols]{dira@$\directionOf(a)$}.
  \end{definition}

  Arrow spaces will be used to represent both velocities, which are directed speeds, and directed distances. Multiplying a velocity by a time yields a directed distance. This scalar multiplication is introduced in

  \begin{definition}
    The map
    \begin{align*}
      \multipli \from \Arrows \times \R_{\geq 0} &\to \Arrows, \mathnote{scalar multiplication $\multipli$}\index[symbols]{dotcentre@$\multipli$}\index[symbols]{centredot@$\multipli$}\\
      ((r, d), s) &\mapsto \begin{dcases*}
                             (0, \every), &if $s = 0$,\\
                             (r \multipli s, d), &if $s > 0$,
                           \end{dcases*}
    \end{align*}
    is called \define{scalar multiplication}.
  \end{definition}

  When we stand at the beginning of a non-empty direction-preserving path and walk along it until we reach its end, we start our walk on the first edge of the path in a certain direction and we end it on the last edge of the path in a certain direction. These directions are introduced in

  \begin{definition}
    Let $\mathfrak{p}$ be a direction-preserving path of $\continuumPaths_{\directionPreserving}$. If $\mathfrak{p}$ is empty, let $\directionOf_\sourceOf(\mathfrak{p}) = \every$ and let $\directionOf_\targetOf(\mathfrak{p}) = \every$.

    Otherwise, there are two edges $e_\sourceOf$, $e_\targetOf \in E$, which may be the same, and there are two positive real numbers $\xi_\sourceOf$, $\xi_\targetOf \in \leftOpenAndRightClosedInterval{0, \length(\mathfrak{p})}$ such that $\mathfrak{p}(\closedInterval{0, \xi_\sourceOf}) \subseteq \continuumRepresentationOf{e}_\sourceOf$ and $\mathfrak{p}(\closedInterval{\length(\mathfrak{p}) - \xi_\targetOf, \length(\mathfrak{p})}) \subseteq \continuumRepresentationOf{e}_\targetOf$. Moreover, there are four real numbers $r_\sourceOf$, $r_\sourceOf'$, $r_\targetOf$, and $r_\targetOf'$ such that $\equivalenceClassOf{(r_\sourceOf, e_\sourceOf)}_\sim = \mathfrak{p}(0)$, $\equivalenceClassOf{(r_\sourceOf', e_\sourceOf)}_\sim = \mathfrak{p}(\xi_\sourceOf)$, $\equivalenceClassOf{(r_\targetOf, e_\targetOf)}_\sim = \mathfrak{p}(\length(\mathfrak{p}) - \xi_\targetOf)$, and $\equivalenceClassOf{(r_\targetOf', e_\targetOf)}_\sim = \mathfrak{p}(\length(\mathfrak{p}))$. Let $\directionOf_\sourceOf(\mathfrak{p}) = (\signOf(r_\sourceOf' - r_\sourceOf), e_\sourceOf)$ and let $\directionOf_\targetOf(\mathfrak{p}) = (\signOf(r_\targetOf' - r_\targetOf), e_\targetOf)$.

    In both cases, the (semi-)direction $\directionOf_\sourceOf(\mathfrak{p})$ is called \graffito{source direction $\directionOf_\sourceOf(\mathfrak{p})$ of $\mathfrak{p}$}\define{source direction of $\mathfrak{p}$}\index[symbols]{dirsigmapfraktur@$\directionOf_\sourceOf(\mathfrak{p})$} and the (semi-)direction $\directionOf_\targetOf(\mathfrak{p})$ is called \define{target direction of $\mathfrak{p}$}\graffito{target direction $\directionOf_\targetOf(\mathfrak{p})$ of $\mathfrak{p}$}\index[symbols]{dirtaupfraktur@$\directionOf_\targetOf(\mathfrak{p})$}.
  \end{definition} 

  \begin{remark}
    The edge $e_\sourceOf$ is the first edge of the path $\mathfrak{p}$ and the edge $e_\targetOf$ is its last edge. The numbers $\xi_\sourceOf$ and $\xi_\targetOf$ are two positive real numbers such that the first $\xi_\sourceOf$ length units of the path run on its first edge and the last $\xi_\targetOf$ length units of the path run on its last edge. The numbers $r_\sourceOf$ and $r_\sourceOf'$ are the positions of the path on its first edge at its very beginning and after $\xi_\sourceOf$ length units, and the numbers $r_\targetOf$, and $r_\targetOf'$ are the positions of the path on its last edge $\xi_\targetOf$ length units before its end and at its very end. Therefore, the signum of $r_\sourceOf' - r_\sourceOf$ is the start direction on the first edge of the path and the signum of $r_\targetOf' - r_\targetOf$ is the end direction on the last edge of the path.
  \end{remark}

  \begin{remark}
    For each non-empty path $\mathfrak{p} \in \continuumPaths_{\directionPreserving}$, we have $\directionOf_\sourceOf(\mathfrak{p}) \in \directionOf(\sourceOf(\mathfrak{p}))$ and $\directionOf_\targetOf(\mathfrak{p}) \in \reverse \directionOf(\targetOf(\mathfrak{p}))$.
  \end{remark}

  \section{Signal Machines}
  \label{section:signal-machines}

  In this section, let $\Graph = \ntuple{\Vertices, \Edges, \eendsOf}$ be a non-trivial, finite, and connected undirected multigraph, let $\weightOf$ be an edge weighting of $\Graph$, and let $M$ be a continuum representation of $\Graph$. Recall that, according to our definition of undirected multigraphs, there are no self-loops in $\Graph$. To motivate the definitions in this section, we talk as if there were a signal machine in front of us whose time evolution we can observe, although this evolution is not completely defined until the end of this section. 

  If you observe the time evolution of a signal machine on the graph $M$, you see signals of different kinds and various speeds each carrying some data move along edges. When signals collide, they may be reflected, removed, new signals may be created, and so on. Similarly, when signals reach a vertex, they may be removed, copies of them may be sent onto all incident edges, new signals may be created, and so on. You may also see stationary signals and signals that travel side-by-side at the same speed. What happens when signals collide or reach a vertex is decided by two local transition functions, one that handles such \defineX{events}{event}\graffito{event} in vertices and one that handles them on edges.

  The only events on edges are collisions. In each collision on an edge, there are at least two signals involved, the involved signals are either stationary or they move in one of the two possible directions, and at least two of the signals collide head-on or rear-end. Such a collision results in a set of signals that are either stationary or move in one of the two possible directions.

  A vertex may be reached by just one signal or multiple signals may collide in it. In both cases, there is at least one signal involved, the involved signals are either stationary or they moved towards the vertex just before the event, and at least one signal is moving. Such an event results in a set of signals that are either stationary or move away from the vertex along incident edges.

  \begin{definition}
    Let $\Kinds$\graffito{set $\Kinds$}\index[symbols]{Knd@$\Kinds$} be a set, let $\speedOf$\graffito{map $\speedOf$} be a map from $\Kinds$ to $\R_{\geq 0}$, let $\family{\Data_k}_{k \in \Kinds}$\graffito{family $\family{\Data_k}_{k \in \Kinds}$} be a family of sets, let\graffito{set $\Signals$}\graffito{set $\Directions_e$} 
    \begin{equation*}
      \Signals = \setOf{(k, d, u) \suchThat k \in \Kinds \text{, } (\speedOf(k), d) \in \Arrows \text{, and } u \in \Data_k}, \index[symbols]{Sgnl@$\Signals$}
    \end{equation*}
    let $\Directions_e = \setOf{\setOf{d, \reverse d} \suchThat d \in \Directions}$\index[symbols]{Dire@$\Directions_e$}, let
    \begin{align*}
      \domainOf(\localTransitionFunction_e) = \{&S \in \powerSetOf(\Signals) \suchThat
          \Exists D \in \Directions_e \SuchThat
              \cardinalityOf{S} \geq 2 \text{ and }\mathnote{set $\domainOf(\localTransitionFunction_e)$}\index[symbols]{domdeltae@$\domainOf(\localTransitionFunction_e)$}\\ 
              &\ForEach (k, d, u) \in S \Holds d \in \setOf{\every} \cup D \text{ and }\\ 
              &\Exists (k, d, u) \in S \Exists (k', d', u') \in S \SuchThat
                  \begin{aligned}[t]
                    &d \neq d' \text{ or }\\
                    &\speedOf(k) \neq \speedOf(k')\},
                  \end{aligned} 
    \end{align*}
    let $\localTransitionFunction_e$\graffito{map $\localTransitionFunction_e$}\index[symbols]{deltae@$\localTransitionFunction_e$} be a map from $\domainOf(\localTransitionFunction_e)$ to $\powerSetOf(\Signals)$ such that
    \begin{equation*}
      \ForEach S \in \domainOf(\localTransitionFunction_e) \Exists D \in \Directions_e \SuchThat \ForEach (k, d, u) \in S \cup \localTransitionFunction_e(S) \Holds d \in \setOf{\every} \cup D,
    \end{equation*}
    let $\Directions_v = \directionOf(\continuumVertices)$\graffito{set $\Directions_v$}\index[symbols]{Dirv@$\Directions_v$}, let
    \begin{align*}
      \domainOf(\localTransitionFunction_v) = \setOf{&(D, S) \in \Directions_v \times \powerSetOf(\Signals) \suchThat
          \cardinalityOf{S} \geq 1 \text{ and }\mathnote{set $\domainOf(\localTransitionFunction_v)$}\index[symbols]{domdeltav@$\domainOf(\localTransitionFunction_v)$}\\
          &\ForEach (k, d, u) \in S \Holds d \in \setOf{\every} \cup \reverse D \text{ and }\\ 
          &\Exists (k, d, u) \in S \SuchThat \speedOf(k) > 0},
    \end{align*}
    and let $\localTransitionFunction_v$\graffito{map $\localTransitionFunction_v$}\index[symbols]{deltav@$\localTransitionFunction_v$} be a map from $\domainOf(\localTransitionFunction_v)$ to $\powerSetOf(\Signals)$ such that
    \begin{equation*}
      \ForEach (D, S) \in \domainOf(\localTransitionFunction_v) \ForEach (k, d, u) \in \localTransitionFunction_v(D, S) \Holds d \in \setOf{\every} \cup D.
    \end{equation*}

    The quadruple $\mathcal{S} = \ntuple{\Kinds, \speedOf, \family{\Data_k}_{k \in \Kinds}, (\localTransitionFunction_e, \localTransitionFunction_v)}$\graffito{signal machine $\mathcal{S}$}\index[symbols]{Scalligraphic@$\mathcal{S}$} is called \define{signal machine}; each element $k \in \Kinds$ is called \define{kind}\graffito{kind $k$}\index[symbols]{k@$k$}; for each kind $k \in \Kinds$, the non-negative real number $\speedOf(k)$ is called \define{speed of $k$}\graffito{speed $\speedOf(k)$ of $k$}\index[symbols]{spdk@$\speedOf(k)$} and the set $\Data_k$ is called \define{data set of $k$}\graffito{data set $\Data_k$ of $k$}\index[symbols]{Dtk@$\Data_k$}; each element $s \in \Signals$ is called \define{signal}\graffito{signal $s$}\index[symbols]{s@$s$}; and the maps $\localTransitionFunction_e$ and $\localTransitionFunction_v$ are called \define{local transition function on edges} and \define{in vertices}\graffito{local transition functions $\localTransitionFunction_e$ and $\localTransitionFunction_v$ on edges and in vertices} respectively.
  \end{definition} 

  \begin{remark}
    The local transition function $\localTransitionFunction_e$ is used to handle events on edges but not in their endpoints. It gets the signals that are involved in the event and returns the resulting signals. Because in each event at least one moving signal is involved, the direction of this signal and the map that reverses orientation can be used by $\localTransitionFunction_e$ to determine the two possible directions the resulting signals may have.

    The local transition function $\localTransitionFunction_v$ is used to handle events in vertices. It gets the directions signals may take at the respective vertex and the signals that are involved in the event and returns the resulting signals.

    The local transition functions $\localTransitionFunction_e$ and $\localTransitionFunction_v$ are supposed to regard directions as black boxes that can merely be distinguished and whose orientation can be reversed. They must not determine edges or vertices by deconstructing directions, which is possible with the chosen representation of directions. If they did something like that, the signal machine would not be uniform.
  \end{remark} 

  \begin{remark}
    At the beginning of this section we fixed a general multigraph. This multigraph should be regarded as the blueprint of a multigraph. A signal machine depends only on that blueprint and not on any specific properties that a concrete choice of a multigraph may have. So, one and the same signal machine can be instantiated for any multigraph, each instantiation results in a machine on a concrete multigraph, and these instantiations are uniform in the chosen multigraphs. In other words, a signal machine is a map from the set of all multigraphs to the set of quadruples that describe instantiations of the machine on concrete multigraphs and this map depends only in a trivial way on its argument. 

    The quadruple that describes signal machines could be made independent of multigraphs by choosing a different representation for directions. They could for example be represented by integers or vectors or colours equipped with an involution to switch the orientation of directions.

    Even the global transition function, which is introduced below and describes the time evolution of a signal machine, could be made independent of multigraphs by representing them as patterns in high-dimensional Euclidean spaces (think of the drawing of a graph on a piece of paper). The directions are then vectors that are tangential to edges.

    In classical solutions of the firing squad synchronisation problem the regions to be synchronised are actually represented as patterns in integer lattices: The cells outside the region are in the same state, say $0$, and all cells inside the region are not in state $0$, more precisely, one cell inside the region is in a state that distinguishes it as the general, say $1$, and the other cells inside the region are all in the same state, say $2$.
  \end{remark}

  In the remainder of this section, let $\mathcal{S} = \ntuple{\Kinds, \speedOf, \family{\Data_k}_{k \in \Kinds}, (\localTransitionFunction_e, \localTransitionFunction_v)}$ be a signal machine.

  To describe the time evolution of our signal machine, the following notions are convenient.

  \begin{definition}
    Let $\Times$\graffito{set $\Times$}\index[symbols]{T@$\Times$} be the set $\R_{\geq 0}$, let $\extendedTimes$\graffito{set $\extendedTimes$}\index[symbols]{Tbar@$\extendedTimes$} be the set $\Times \cup \setOf{\infty}$, let $\States$\graffito{set $\States$}\index[symbols]{Q@$\States$} be the set $\powerSetOf(\Signals)$, and let $\Configurations$\graffito{set $\Configurations$}\index[symbols]{Cnf@$\Configurations$} be the set $Q^M$. Each element $t \in \extendedTimes$ is called \define{time}\graffito{time $t$}\index[symbols]{t@$t$} and the time $\infty$ is called \define{improper}\graffito{improper time}\index[symbols]{infinity@$\infty$}, each element $q \in \States$ is called \define{state}\graffito{state $q$}\index[symbols]{q@$q$}, and each element $c \in \Configurations$ is called \define{configuration}\graffito{configuration $c$}\index[symbols]{c@$c$}.
  \end{definition}

  The components of signals and some compounds of them are named in

  \begin{definition}
    Let $s = (k, d, u)$ be a signal of $\Signals$. The kind $k$ of $s$ is denoted by $\kindOf(s)$\graffito{kind $\kindOf(s)$ of $s$}\index[symbols]{knds@$\kindOf(s)$}; the speed $\speedOf(k)$ of $s$ is denoted by $\speedOf(s)$\graffito{speed $\speedOf(s)$ of $s$}\index[symbols]{spds@$\speedOf(s)$}; the (semi-)direction $d$ of $s$ is denoted by $\directionOf(s)$\graffito{(semi-)direction $\directionOf(s)$ of $s$}\index[symbols]{dirs@$\directionOf(s)$}; the velocity $(\speedOf(k), d)$ of $s$ is denoted by $\velocityOf(s)$\graffito{velocity $\velocityOf(s)$ of $s$}\index[symbols]{vels@$\velocityOf(s)$}; and the datum $u$ of $s$ is denoted by $\datumOf(s)$\graffito{datum $\datumOf(s)$ of $s$}\index[symbols]{dts@$\datumOf(s)$}.
  \end{definition}

  \begin{remark}
    For each time $t \in \Times$, the arrow $\velocityOf(s) \multipli t$ is equal to the arrow $(\speedOf(k) \multipli t, d)$, which can be interpreted as a directed distance.
  \end{remark}

  Signals of speed $0$ are named in

  \begin{definition}
    Let $s$ be a signal of $\Signals$. It is called \define{stationary}\graffito{stationary signal} if and only if $\speedOf(s) = 0$.
  \end{definition}

  When we stand on a point facing in a direction and from there we walk a fixed distance without making U-turns but otherwise making arbitrary choices at each vertex, then there is a finite number of points that we may reach and we reach them walking in some direction. The set of these points with and without target-directions is introduced in

  \begin{definition}
    Let $m$ be a point of $M$ and let $(\ell, d)$ be an arrow. The set of points with target-directions that can be reached from $m$ by a direction-preserving path of length $\ell$ with source-direction $d$ is
    \begin{equation*}
      \ReachOf_{(\ell, d)}^{\directed}(m) =
          \begin{aligned}[t]
            \setOf{(\targetOf(\mathfrak{p}), \directionOf_{\xend}(\mathfrak{p})) \suchThat{}
                 &\mathfrak{p} \in \continuumPaths_{\directionPreserving} \text{, } \length(\mathfrak{p}) = \ell \text{, }\\
                 &\directionOf_{\start}(\mathfrak{p}) = d \text{, and } \sourceOf(\mathfrak{p}) = m}
          \end{aligned}
      \mathnote{$\ReachOf_{(\ell, d)}^{\directed}(m)$}
    \end{equation*}
    and without target-directions it is
    \begin{equation*}
      \ReachOf_{(\ell, d)}(m) = \setOf{m' \suchThat \Exists d \in \Directions \SuchThat (m', d) \in \ReachOf_{(\ell, d)}^{\directed}(m)}.\mathnote{$\ReachOf_{(\ell, d)}(m)$} \qedhere 
    \end{equation*}
  \end{definition}

  An event occurs when a signal reaches a vertex or two signals coming from different points collide. The time of the next event is given a name in

  \begin{definition}
    Let $c$ be a configuration of $\Configurations$. The minimum time until a signal in $c$ reaches a vertex is
    \begin{equation*}
      t' = \inf_{m \in M} \inf_{\substack{s \in c(m)\\ \speedOf(s) > 0}} \inf \setOf{t \in \R_{> 0} \suchThat \ReachOf_{\velocityOf(s) \multipli t}(m) \cap \continuumVertices \neq \emptyset}.
    \end{equation*}
    The minimum time until at least two signals in $c$ collide is
    \begin{equation*}
      t'' = \inf_{\substack{m, m' \in M\\ m \neq m'}} \inf_{\substack{s \in c(m)\\ s' \in c(m')}} \inf \setOf{t \in \R_{> 0} \suchThat \ReachOf_{\velocityOf(s) \multipli t}(m) \cap \ReachOf_{\velocityOf(s') \multipli t}(m') \neq \emptyset}.
    \end{equation*}
    The minimum time until the next event(s) in $c$ occurs is \graffito{next event(s) time $t_0(c)$}$t_0(c) = \min\setOf{t', t''}$\index[symbols]{t0c@$t_0(c)$}.
  \end{definition}

  \begin{remark}
    A stationary signal at a vertex does never reach the vertex (it is already there) and two signals that already are at the same vertex do never collide (they may have collided when they got there but now they just are at the same vertex; if they have a non-zero velocity, then they will leave the vertex without interfering each other, and, if they have the same positive velocity, then they will travel alongside each other).
  \end{remark}

  \begin{remark}
    The next event time may be $0$, which means that events accumulate at time $0$, or $\infty$, which means that there are no more events in the future; note that $\inf\emptyset = \infty$. It is for example $0$ if there is a sequence of signals moving at the same velocity towards the same vertex each one being already a little closer to the vertex than the previous one. And it is for example $\infty$ if there are no signals at all or there are only stationary signals.
  \end{remark}

  If each signal moves with its velocity and upon reaching a vertex propagates to all incident edges except the one it came from (making copies of itself if necessary) and if collisions of signals are ignored, then the set of points an event occurs in at a time in the future is given a name in 

  \begin{definition} 
    Let $c$ be a configuration of $\Configurations$ and let $t$ be a time of $\Times$. The set of vertices that signals in $c$ reach at time $t$ (under the assumptions given in the introduction to the present definition) is
    \begin{equation*}
      M' = \begin{dcases*}
             \emptyset, &if $t = 0$,\\ 
             \bigcup_{\substack{m \in M\\ s \in c(m)\\ \speedOf(s) > 0}} \ReachOf_{\velocityOf(s) \multipli t}(m) \cap \continuumVertices, &otherwise.
           \end{dcases*}
    \end{equation*}
    The set of points that signals in $c$ collide in at time $t$ is
    \begin{equation*}
      M'' = \bigcup_{\substack{m, m' \in M\\ m \neq m'}} \bigcup_{\substack{s \in c(m)\\ s' \in c(m')}} \ReachOf_{\velocityOf(s) \multipli t}(m) \cap \ReachOf_{\velocityOf(s') \multipli t}(m').
    \end{equation*}
    The set of points that signals in $c$ are involved in an event in at time $t$ is $M_t(c) = M' \cup M''$\graffito{set $M_t(c)$ of points that signals in $c$ are involved in an event in at time $t$}\index[symbols]{Mtc@$M_t(c)$}.
  \end{definition}

  \begin{remark}
    For times $t$ before and including the time $t_0(c)$ of the next event, the definition of $M_t(c)$ is natural. And for other times, it is plausible with the explanation given before the definition and it is used to handle accumulations of events and accumulations of accumulations of events and so on.
  \end{remark}

  As above, if each signal moves with its velocity and upon reaching a vertex propagates to all incident edges except the one it came from (making copies of itself if necessary) and if collisions of signals are ignored, then starting our signal machine in a configuration $c$ and letting it run for a time $t$ without handling propagation of signals in vertices at time $t$ yields a new configuration $\gtfIgnoreCollisions(t)(c)$ as defined in 

  \begin{definition}
    \begin{align*}
      \gtfIgnoreCollisions \from \Times &\to (\Configurations \to \Configurations), \mathnote{map $\gtfIgnoreCollisions$ from $\Times$ to $\Configurations \to \Configurations$}\index[symbols]{boxplus@$\gtfIgnoreCollisions$}\\
      0 &\mapsto \identityMap_{\Configurations},\\ 
      t &\mapsto [c \mapsto [m \mapsto \begin{aligned}[t]
                                         \setOf{s \in \Signals \suchThat{} &\Exists m' \in M \Exists s' \in c(m') \SuchThat\\
                                                                           &(m, \directionOf(s)) \in \ReachOf_{\velocityOf(s') \multipli t}^{\directed}(m') \text{, }\\
                                                                           &\kindOf(s) = \kindOf(s') \text{, and }\\
                                                                           &\datumOf(s) = \datumOf(s')}]]. \qedhere
                                       \end{aligned}
    \end{align*}
  \end{definition}

  \begin{remark}
    For each time $t \in \Times$ and each configuration $c \in \Configurations$, if there is a signal in $c$ that reaches a vertex in time $t$ (or one of its duplicates does), then in the configuration $\gtfIgnoreCollisions(t)(c)$ that signal is at the vertex and its velocity is the one it had just before reaching the vertex. The direction of that velocity points away from the edge the signal came from and it does not point to any edge that is incident to the vertex. It is then up to the local transition function to decide what to do with the signal and, if it is not removed, what its direction shall be.
  \end{remark}

  \begin{remark}
    The map $\gtfIgnoreCollisions$ is used to determine future configurations before or right until the next event occurs and needs to be handled, and also to make crude predictions of future configurations beyond the next event time by propagating at vertices and ignoring collisions as explained above. These predictions will be used to handle accumulations of events and accumulations of accumulations of events and so on.
  \end{remark}

  Until the first events occur, signals move along edges without colliding. At the time the first events occur, a signal reached a vertex or two signals collided (on a vertex or edge) or multiple such events happened. An event in a vertex is handled by the local transition function $\localTransitionFunction_v$ and on an edge by $\localTransitionFunction_e$. This global behaviour can be described by a map that maps a configuration to the configuration right after the first events occurred and have been handled. This map is given in 

  \begin{definition}
      \index[symbols]{Deltadot0@$\gtfJump_0$}\begin{align*}
        \gtfJump_0 \from \Configurations &\to \Configurations, \mathnote{map $\gtfJump_0$ from $\Configurations$ to $\Configurations$}\\
        c &\mapsto
          \left\{
            \begin{aligned}
              &c, \text{ if $t_0(c) \in \setOf{0, \infty}$},\\
              &[m \mapsto
                \left\{
                  \begin{aligned}
                    &c'(m), &&\text{if $m \notin M_{t_0(c)}(c)$,}\\
                    &\localTransitionFunction_v(\directionOf(m), c'(m)), &&\text{if $m \in M_{t_0(c)}(c) \cap \continuumVertices$,}\\
                    &\localTransitionFunction_e(c'(m)), &&\text{if $m \in M_{t_0(c)}(c) \smallsetminus \continuumVertices$,}
                  \end{aligned}
                \right\}
              ],\\
              &&\llap{otherwise,\quad}
            \end{aligned}
          \right.\\
          &\phantom{\mapsto{}}\text{ where } c' = \gtfIgnoreCollisions(t_0(c))(c). \qedhere
      \end{align*}
  \end{definition}

  \begin{remark}
    The map $\gtfJump_0$ maps a configuration to itself if the next event time is $0$, which means that event times accumulated at $0$, or $\infty$, which means that there is no next event. And it maps a configuration to the configuration that is reached after the first events have been handled, by first using $\gtfIgnoreCollisions$ to determine the configuration in which the events occur and then handling all occurring events with $\localTransitionFunction_v$ and $\localTransitionFunction_e$.
  \end{remark}

  If the next event time is $0$, which we call \index{singularity}\graffito{singularity of order $-1$}\defineX{singularity of order $-1$}{singularity!of order $-1$} (see \cref{figure:singularity-of-order-minus-one}), then the machine is sometimes stuck, in the sense that there is no natural way to define what configuration the machine is in at any time in the future, and sometimes the machine can go forward in time, in the sense that there is a natural way to define what configuration the machine is in at least until some time in the future; the latter case is handled later and largely ignored for now. If the next event time is $\infty$, then the machine does nothing for eternity.

  Otherwise, the machine can at least proceed until the next event time and handle the occurring events, which we call \graffito{singularity of order $0$}\defineX{singularities of order $0$}{singularity!of order $0$}, and then the next event time may again be $0$, $\infty$, or something in between. It may happen that the next event times are never $0$ or $\infty$ but accumulate at some time in the future, which we call \graffito{singularity of order $j$, for $j \in \N_+$}\defineX{singularity of order $1$}{singularity!of order $j$, for $j \in \N_+$} (see \cref{figure:singularity-of-order-one}). In that case repeated applications of $\gtfJump_0$ never reach a configuration at that future time or a time beyond. But we can in a sense take the limit of the sequence of configurations that repeated applications of $\gtfJump_0$ yield. Yet, it may even happen that singularities of order $1$ accumulate at some time in the future, which we call \defineX{singularity of order $2$}{singularity!of order $j$, for $j \in \N_+$}. Again, we can in the same sense as before take the limit of the sequence of configurations at these singularities. It may continue this way ad infinitum. In precise terms this is done in

  \begin{figure}
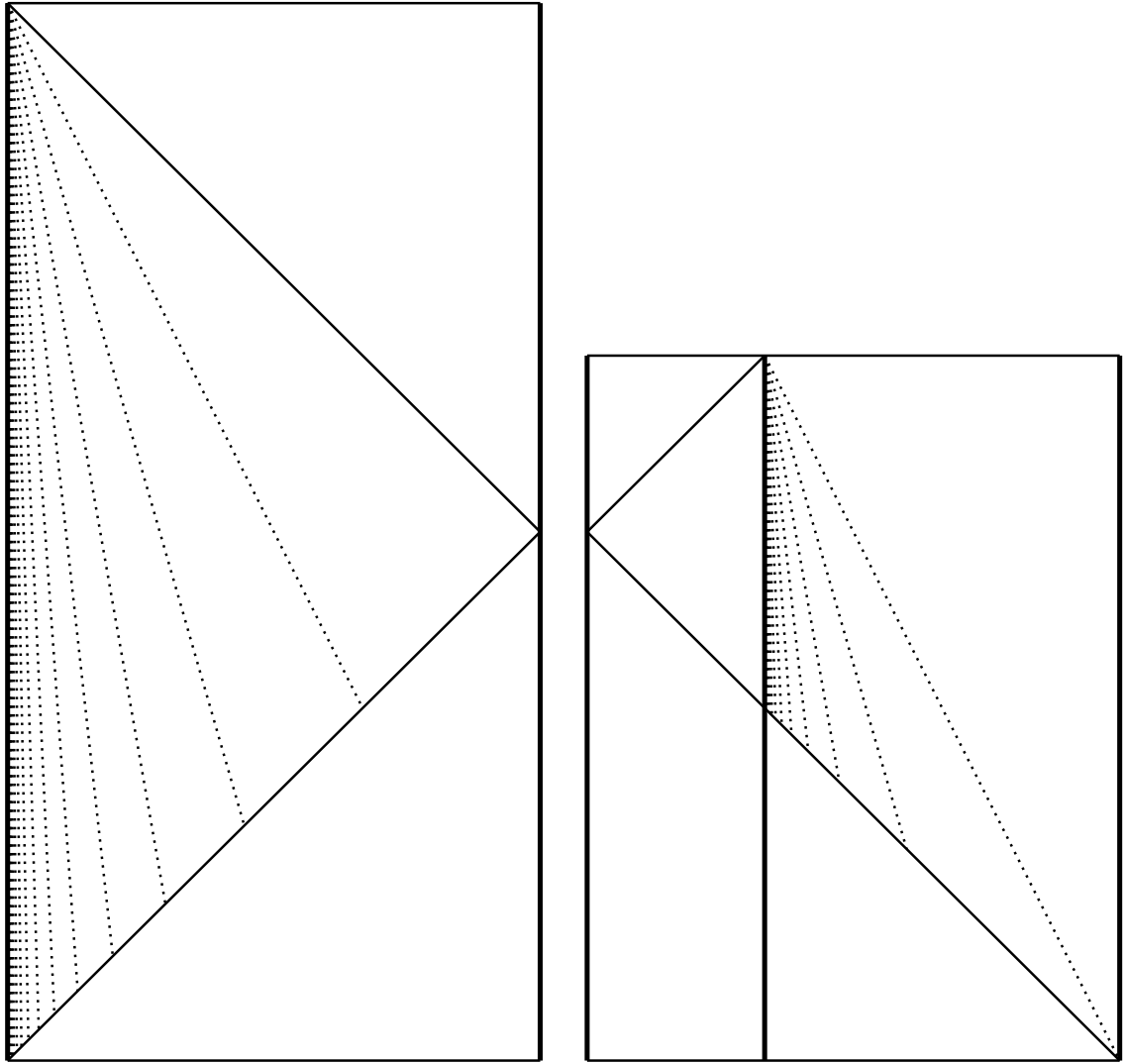

    \myfloatalign
    \begin{wide}
      \figureSingularities
      \caption{Both subfigures depict a space-time diagram of the time evolution of an unspecified signal machine, where space is drawn on the horizontal axis and time on the vertical axis evolving from top to bottom. In \cref{figure:singularity-of-order-one}, there are infinitely many signals of various speeds arbitrarily close to $0$ and slower than $1$ emanating from the left vertex, the fastest signal, which is the one of speed $1$, is reflected at the right vertex, and the reflected signal collides with all other signals at shorter and shorter time spans between collisions, resulting in a singularity of order $1$ at the last depicted time. In \cref{figure:singularity-of-order-minus-one}, there are infinitely many signals of various speeds arbitrarily close to $0$ and slower than $1$ emanating to the right from the middle vertex, and there is one signal of speed $1$ emanating to the left from the middle vertex, this signal is reflected at the left vertex, and at the time the reflected signal reaches the middle vertex is a singularity of order $-1$ because it collides with all signals that emanated to the right and are arbitrarily close to the vertex.}
      \label{figure:singularities}
    \end{wide}
  \end{figure}

  \begin{definition}
    The sequence $\sequence{t_{j - 1}^n}_{n \in \N_0}$, where the $n$ in $t_{j - 1}^n$ is an upper index and does not stand for exponentiation, the map $t_j$, and the map $\gtfJump_j$, for $j \in \N_+$, are defined by mutual induction as follows: The maps $t_0$ and $\gtfJump_0$ have already been defined and, for each positive integer $j$, let
    \begin{equation*}
      \sequence*{
        \begin{aligned}
          t_{j - 1}^n \from \Configurations &\to \extendedTimes, \index[symbols]{tj-1n@$t_{j - 1}^n$}\\ 
          c &\mapsto \sum_{i = 0}^{n - 1} t_{j - 1}(\gtfJump_{j - 1}^i(c)),
        \end{aligned}
      }_{n \in \N_0} \mathnote{map $t_{j - 1}^n$ from $\Configurations$ to $\extendedTimes$}
    \end{equation*} 
    (note that $t_{j - 1}^0 = 0$), let
    \begin{align*}
      t_j \from \Configurations &\to \extendedTimes, \mathnote{map $t_j$ from $\Configurations$ to $\extendedTimes$}\index[symbols]{tjsubscript@$t_j$}\\ 
      c &\mapsto \lim_{n \to \infty} t_{j - 1}^n(c),
    \end{align*}
    and let\index[symbols]{Deltadotj@$\gtfJump_j$}
    \begin{align*}
      \gtfJump_j \from \Configurations &\to \Configurations, \mathnote{map $\gtfJump_j$ from $\Configurations$ to $\Configurations$}\\
      c &\mapsto
        \left\{
          \begin{aligned}
            &c, \text{ if $t_j(c) \in \setOf{0, \infty}$},\\
            &[m \mapsto
              \left\{
                \begin{aligned}
                  &c'(m), &&\text{if $m \notin M_j(c)$,}\\
                  &\localTransitionFunction_v(\directionOf(m), c'(m)), &&\text{if $m \in M_j(c) \cap \continuumVertices$,}\\
                  &\localTransitionFunction_e(c'(m)), &&\text{if $m \in M_j(c) \smallsetminus \continuumVertices$,}
                \end{aligned}
              \right\}
            ],\\
            &&\llap{otherwise,\quad}
          \end{aligned}
        \right.\\
        &\phantom{\mapsto{}}\text{ where } c' = \liminf_{n \to \infty} \gtfIgnoreCollisions(t_j(c) - t_{j - 1}^n(c))(\gtfJump_{j - 1}^n(c)),\\
        &\phantom{\mapsto{}}\text{ and } M_j(c) = \liminf_{n \to \infty} M_{t_j(c) - t_{j - 1}^n(c)}(\gtfJump_{j - 1}^n(c)),
    \end{align*}
    where the first limit inferior is the pointwise limit inferior of sequences of set-valued maps and the second limit inferior is the limit inferior of sequences of sets. In greater detail, for each sequence $\sequence{c_n}_{n \in \N_0}$ of set-valued maps from $M$ to $Q$, the pointwise limit inferior of $\sequence{c_n}_{n \in \N_0}$ is the map $c \from M \to Q$, $m \mapsto \liminf_{n \to \infty} c_n(m)$ and is denoted by $\liminf_{n \to \infty} c_n$; and, for each sequence $\sequence{A_n}_{n \in \N_0}$ of subsets of $M$, the limit inferior of $\sequence{A_n}_{n \in \N_0}$ is the subset $\bigcup_{n \in \N_0} \bigcap_{k \geq n} A_k$ of $M$ and is denoted by $\liminf_{n \to \infty} A_n$.
  \end{definition} 

  \begin{remark}
    Let $c$ be a configuration of $\Configurations$. Then, for each positive integer $j$, we have $t_{j - 1}^0(c) = 0$ and $t_{j - 1}^1(c) = t_{j - 1}(c)$, and the sequence $\sequence{t_{j - 1}^n(c)}_{n \in \N_0}$ in $\extendedTimes$ is non-decreasing and hence converges in $\extendedTimes$. And, the sequence $\sequence{t_j(c)}_{j \in \N_0}$ in $\extendedTimes$ is non-decreasing and hence converges in $\extendedTimes$.
  \end{remark}

  \begin{remark} 
    Let the signal machine be in a configuration $c$ at time $0$ and let there be no future configuration whose next event time is $0$ or $\infty$. The latter is the case if and only if the sequences $\sequence{t_{j - 1}^n(c)}_{n \in \N_0}$, for $j \in \N_+$, and hence also the sequence $\sequence{t_j(c)}_{j \in \N_0}$, are strictly increasing sequences in $\Times$.

    Then, for each positive integer $j$ and each non-negative integer $n$, at time $t_{j - 1}^n(c)$ the machine is in configuration $\gtfJump_{j - 1}^n(c)$. And, the time $t_0^n(c)$ is the $n$-th time an event occurs (singularity of order $0$), the time $t_1^n(c)$ is the $n$-th time an accumulation of events occurs (singularity of order $1$), the time $t_2^n(c)$ is the $n$-th time an accumulation of accumulation of events occurs (singularity of order $2$), and so forth.

    Moreover, for each non-negative integer $j$, at time $t_j(c)$ the machine is in configuration $\gtfJump_j(c)$. And, the time $t_0(c)$ is the next time an event occurs, the time $t_1(c)$ is the next time an accumulation of events occurs, the time $t_2(c)$ is the next time an accumulation of accumulations of events occurs, and so forth. 

    Furthermore, for each positive integer $j$, the map $\gtfJump_j$ maps the configuration $c$ to the configuration that is reached after an accumulation of singularities of order $j - 1$, which is a singularity of order $j$. First, it calculates the configurations that accumulate, namely $\gtfJump_{j - 1}^n(c)$; secondly, for each of these configurations, it uses $\gtfIgnoreCollisions$ to determine the configuration that would be reached at the accumulation time if there were no further events, which is a crude prediction of the future that becomes better the greater $n$ is; thirdly, it calculates the pointwise limit inferior of these configurations, which is essentially the configuration that contains the signals that all but finitely many of the configurations have in common (in particular, if for a point $m$ the sequence of signals at $m$ become constant, then the limit at $m$ is that set of signals); lastly, it handles collisions.
  \end{remark}

  \begin{remark}
  \label{remark:limit-of-jumps-is-singularity-of-order-minus-1}
    Let the signal machine be in a configuration $c$ at time $0$ and let there be a future configuration whose next event time is $0$. Then, there is a least positive integer $j$ such that the sequence $\sequence{t_{j - 1}^n(c)}_{n \in \N_0}$ is eventually constant. And, there is a least non-negative integer $n$ such that $t_{j - 1}^n(c) = t_{j - 1}^{n + 1}(c)$. The time $t' = t_{j - 1}^n(c)$ is the first time at which the signal machine is in a configuration whose next event time is $0$ and this configuration is $c' = \gtfJump_{j - 1}^n(c)$.

    For each non-negative integer $n'$ such that $n' \geq n$, we have $t_{j - 1}^{n'}(c) = t'$ and $\gtfJump_{j - 1}^{n'}(c) = c'$. And, for each positive integer $j'$ such that $j' \geq j$, the time $t_{j'}(c)$ is equal to $t'$ and the configuration $\gtfJump_{j'}(c)$ is equal to $c'$, and the sequences $\sequence{t_{j'}^n(c)}_{n \in \N_0}$ and $\sequence{\gtfJump_{j'}^n(c)}_{n \in \N_0}$ are the constant sequences $\sequence{t'}_{n \in \N_0}$ and $\sequence{c'}_{n \in \N_0}$. In particular, the limit $\lim_{j \to \infty} t_j(c)$ is equal to $t'$ and the limit inferior $\liminf_{j \to \infty} \gtfJump_j(c)$ is equal to $c'$.
  \end{remark}

  \begin{remark}
    The limit of sequences of configurations and of sets of points does in general not exist. However, the limit inferior and the limit superior always exist. We decided not to use the limit, to avoid case distinctions that would have to be made. Instead, we decided to use the limit inferior; we could as well have decided to use the limit superior. Which of the two has the desired outcome depends on the specific use case.

    For the signal machine that solves the firing squad synchronisation problem that we construct in the next section and the configurations it is initialised with and the configurations it encounters during its time evolution, the encountered limit inferiors and superiors are actually always the same, which means that the limits exist, and hence the choice of limit inferior or superior is irrelevant in that use case.
  \end{remark}

  If the machine never assumes a configuration in which events accumulate at time $0$ and if non-negative singularities of ever higher orders ad infinitum do not accumulate, then the machine can be observed for eternity. Otherwise, it can for now only be observed for all times before $\lim_{j \to \infty} t_j(c)$, where $c$ is the initial configuration of the machine. This time is given a name in 

  \begin{definition}
    For each configuration $c \in \Configurations$, the non-negative real number or infinity $t_\infty(c) = \lim_{j \to \infty} t_j(c)$ is called \defineX{$\infty$-existence time of $c$}{existence time infinity@$\infty$-existence time of $c$}\graffito{$\infty$-existence time $t_\infty(c)$ of $c$}\index[symbols]{tinftyc@$t_\infty(c)$}, and the closed interval $\closedInterval{0, t_\infty(c)}$ is called \defineX{$\infty$-existence interval of $c$}{existence interval infinity@$\infty$-existence interval of $c$}\graffito{$\infty$-existence interval $\closedInterval{0, t_\infty(c)}$ of $c$}.
  \end{definition}

  Repeated applications of powers of the maps $\gtfJump_j$, for decreasing $j \in \N_0$, let us jump to and from configurations right after singularities of decreasing orders. The resulting map is given a name in

  \begin{definition}
    For each non-negative integer $j$, each non-negative integer $k$, and each finite sequence $\sequence{n_i}_{i = j}^k$ of non-negative integers, let
    \begin{equation*} 
      \gtfJump_{\sequence{n_i}_{i = j}^k}
      = \begin{dcases*}
          \identityMap_{\Configurations}, &if $j > k$,\\
          \gtfJump_j^{n_j} \after \gtfJump_{j + 1}^{n_{j + 1}} \after \dotsb \after \gtfJump_k^{n_k}, &otherwise, 
        \end{dcases*}
      \mathnote{map $\gtfJump_{\sequence{n_i}_{i = j}^k}$ from $\Configurations$ to $\Configurations$}
      \index[symbols]{Deltadotniijksequence@$\gtfJump_{\sequence{n_i}_{i = j}^k}$}
    \end{equation*}
    and let
    \begin{equation*}
      t_{\sequence{n_i}_{i = j}^k} = \sum_{i = j}^k t_i^{n_i} \after \gtfJump_{\sequence{n_\ell}_{\ell = i + 1}^k}.
      \mathnote{map $t_{\sequence{n_i}_{i = j}^k}$ from $\Configurations$ to $\extendedTimes$}
      \index[symbols]{tniijksequence@$t_{\sequence{n_i}_{i = j}^k}$} \qedhere
    \end{equation*}
  \end{definition}

  \begin{remark}
    Let the signal machine be in a configuration $c$ at time $0$ and let there be no future configuration whose next event time is $0$ or $\infty$. The map $\gtfJump_{\sequence{n_i}_{i = j}^k}$ applied to $c$, first applies $\gtfJump_k^{n_k}$ to jump from $c$ to the configuration right after the $n_k$-th time a singularity of order $k$ occurs, secondly it applies $\gtfJump_{k - 1}^{n_{k - 1}}$ to jump from that configuration, namely $\gtfJump_{\sequence{n_i}_{i = k}^k}(c)$, to the configuration right after the $n_{k - 1}$-th time a singularity of order $k - 1$ occurs (counting from the time at which the machine is in configuration $\gtfJump_{\sequence{n_i}_{i = k}^k}(c)$), and so forth until it finally applies $\gtfJump_j^{n_j}$ to jump from the configuration $\gtfJump_{\sequence{n_i}_{i = j + 1}^k}(c)$ to the configuration right after the $n_j$-th time a singularity of order $j$ occurs (counting from the time at which the machine is in configuration $\gtfJump_{\sequence{n_i}_{i = j + 1}^k}(c)$), where in the case that $j = 0$, a singularity of order $0$ is nothing but an event.

    The time it takes the machine to get from $c$ to the configuration $\gtfJump_k^{n_k}(c)$ is $t_k^{n_k}(c)$, the time it takes to get from $\gtfJump_k^{n_k}(c)$ to $\gtfJump_{k - 1}^{n_{k - 1}}(\gtfJump_k^{n_k}(c))$ is $t_{k - 1}^{n_{k - 1}}(\gtfJump_k^{n_k}(c))$, and so forth; in total, the time it takes to get from $c$ to $\gtfJump_{\sequence{n_i}_{i = j}^k}(c)$ is $t_{\sequence{n_i}_{i = j}^k}(c)$.
  \end{remark}

  To compute the configuration the machine is in at the $\infty$-existence time of the initial configuration, we can use the maps $\gtfJump_j$ for increasing $j$ to jump from singularities of non-negative lower orders to singularities of ever higher orders, which in the case there are any singularities of order $-1$ comes to a halt at the first such singularity at the $\infty$-existence time and in the other case yields in a sense the limit configuration. And, to compute the configuration at time $t$ before the $\infty$-existence time, we can use one of the maps $\gtfJump_{\sequence{n_i}_{i = 0}^k}$ to jump to the configuration right after the last event before time $t$ and then we can use the map $\gtfIgnoreCollisions$ to jump from there to time $t$. The resulting map describes the time evolution of the signal machine before or at $\infty$-existence times and it is given in

  \begin{definition}
    For each time $t \in \extendedTimes$, the set of configurations whose $\infty$-existence interval contains $t$ is
    \begin{equation*}
      \Configurations_t^\infty = \setOf{c \in \Configurations \suchThat t \leq t_\infty(c)}.
      \mathnote{set $\Configurations_t^\infty$} 
      \index[symbols]{Cnftinfinity@$\Configurations_t^\infty$}
    \end{equation*}
    Let
    \begin{align*} 
      \gtfNonNegativeSingularities \from \extendedTimes &\to \bigcup_{t \in \extendedTimes} (\Configurations_t^\infty \to \Configurations), \mathnote{map $\gtfNonNegativeSingularities$ from $\extendedTimes$ to $\bigcup_{t \in \extendedTimes} (\Configurations_t^\infty \to \Configurations)$}\index[symbols]{boxminus@$\gtfNonNegativeSingularities$}\\ 
      t &\mapsto
        \left[
          \begin{aligned}
            \Configurations_t^\infty &\to \Configurations,\\
            c &\mapsto
              \begin{dcases*}
                \liminf_{j \to \infty} \gtfJump_j(c), &if $t = t_\infty(c)$,\\ 
                \gtfIgnoreCollisions(t - t_{\sequence{n_i}_{i = 0}^k}(c))(\gtfJump_{\sequence{n_i}_{i = 0}^k}(c)), &otherwise,
              \end{dcases*}\\
              &\phantom{\mapsto{}}\text{ for the least } k \in \N_0 \text{ and the } \sequence{n_i}_{i = 0}^k \in \N_0^{k + 1}\\
              &\phantom{\mapsto{}}\text{ with } t \in \leftClosedAndRightOpenInterval{t_{\sequence{n_i}_{i = 0}^k}(c), t_{(n_0 + 1, n_1, n_2, \dotsc, n_k)}(c)}.
          \end{aligned}
        \right]
    \end{align*} 
    Note that the least index $k$ and the finite sequence $\sequence{n_i}_{i = 0}^k$ that occur above are uniquely determined by the time $t$ and the configuration $c$. 
  \end{definition}

  \begin{remark}
    In the second case in the definition of $\gtfNonNegativeSingularities$, we have $t \in \leftClosedAndRightOpenInterval{t_k^{n_k}(c), t_k^{n_k + 1}(c)}$, and $t - t_k^{n_k}(c) \in \leftClosedAndRightOpenInterval{t_{k - 1}^{n_{k - 1}}(\gtfJump_k^{n_k}(c)), t_{k - 1}^{n_{k - 1} + 1}(\gtfJump_k^{n_k}(c))}$, and so forth.
  \end{remark}

  \begin{remark}
    For each configuration $c \in \Configurations$, the map $\gtfNonNegativeSingularities(\blank)(c)$ is defined on the closed interval $\closedInterval{0, t_\infty(c)}$.
  \end{remark}

  \begin{remark}
    Let $t$ be a time of $\Times$ and let $c$ be a configuration of $\Configurations_t^\infty$ such that $t \neq t_\infty(c)$. If events do not accumulate for the initial configuration $c$, then $\gtfNonNegativeSingularities(t)(c) = \gtfIgnoreCollisions(t - t_0^{n_0}(c))(\gtfJump_0^{n_0}(c))$, for the $n_0 \in \N_0$ with $t \in \leftClosedAndRightOpenInterval{t_0^{n_0}(c), t_0^{n_0 + 1}(c)}$; in words, we apply $\gtfJump_0$ repeatedly, jumping from event to event, until we reach the configuration $\gtfJump_0^{n_0}(c)$ at time $t_0^{n_0}(c)$ with the property that the next event (if there even is one) occurs after $t$, at which point we use $\gtfIgnoreCollisions$ to move signals along edges for the remaining time $t - t_0^{n_0}(c)$.

    If events do accumulate for $c$ but singularities of order $1$ do not accumulate, then we apply $\gtfJump_1$ repeatedly, jumping from singularity to singularity, until we reach the configuration $\gtfJump_1^{n_1}(c)$ at time $t_1^{n_1}(c)$ with the property that the next singularity (if there even is one) occurs after $t$, at which point we apply $\gtfJump_0$ repeatedly, jumping from event to event, until we reach the configuration $\gtfJump_0^{n_0}(\gtfJump_1^{n_1}(c))$ at time $t_0^{n_0}(c) + t_1^{n_1}(c)$ with the property that the next event (if there even is one) occurs after $t$, at which point we use $\gtfIgnoreCollisions$ to move signals along edges for the remaining time $t - t_0^{n_0}(c) - t_1^{n_1}(c)$.

    And so forth.
  \end{remark}

  \begin{remark}
  \label{remark:gtf-at-infinity-existence-time-sometimes-yields-the-configuration-at-a-singulairty-of-order-minus-1}
    Let the signal machine be in a configuration $c$ at time $0$. If there is a singularity of order $-1$ in the future, then, according to \cref{remark:limit-of-jumps-is-singularity-of-order-minus-1}, the time $t_\infty(c)$ is the first time at which the signal machine is at such a singularity and the corresponding configuration is $\gtfNonNegativeSingularities(t_\infty(c))(c)$. Otherwise, the time $t_\infty(c)$, which may be the improper time $\infty$, is the time just after all singularities and the corresponding configuration is $\gtfNonNegativeSingularities(t_\infty(c))(c)$. In either case, in what is to come, for simplicity, if $t_\infty(c)$ is finite, then we talk as if there is a singularity of order $-1$ at time $t_\infty(c)$. 
  \end{remark}


  At the time of a singularity of order $1$, there are infinitely many events that occur just \emph{before} that time and arbitrarily close to it, and the problem is to define the configuration at the time of the singularity. At the time of a singularity of order $-1$, there are infinitely many events that occur just \emph{after} that time and arbitrarily close to it, and the problem is to define the configurations at all times after the singularity. Analogous problems exist for accumulations of singularities of order $1$ or $-1$, and accumulations of accumulations of singularities of order $1$ or $-1$, and so forth. For singularities of positive orders these problems have been solved above but not for singularities of order $-1$ and its accumulations.

  At a singularity of order $-1$ at time $0$, to compute the configuration at a small enough time $t$, first, we make crude predictions of the future with $\gtfIgnoreCollisions$ by jumping past the singularity to future times $\varepsilon$ ignoring events, secondly, we extrapolate these predictions to the time $t$ with $\gtfNonNegativeSingularities$ by letting the machine evolve them until the time $t$, and, lastly, we take the limit inferior of these predictions as $\varepsilon$ tends to $0$. This does not work for all singularities of order $-1$ regardless of how small we choose $t$ (for example if at each point in time of a time span after and including $0$ an event occurs), but it does work for the singularities of order $-1$ that occur in our quasi-solution of the firing squad synchronisation problem. 

  If the $\infty$-existence time of the current configuration is $\infty$, then the machine can be observed for eternity. Otherwise, it can be observed until the $\infty$-existence time using $\gtfNonNegativeSingularities$ and from there at least until the least time until which all the crude predictions mentioned above can be observed; this time span is named in

  \begin{definition}
    \begin{align*}
      t_{-1} \from \Configurations &\to \extendedTimes, \mathnote{map $t_{-1}$ from $\Configurations$ to $\extendedTimes$}\index[symbols]{t-1@$t_{-1}$}\\
      c &\mapsto
        \begin{dcases*}
          \infty, &if $t_\infty(c) = \infty$,\\
          \inf_{\varepsilon \in \R_{> 0}}((\varepsilon + t_\infty(c)) + t_\infty(c_\varepsilon)), &otherwise,
        \end{dcases*}\\
        &\phantom{\mapsto{}}\text{ where } c_\varepsilon = \gtfIgnoreCollisions(\varepsilon)(\gtfNonNegativeSingularities(t_\infty(c))(c)). \qedhere
    \end{align*}
  \end{definition}

  How the machine evolves until and beyond $t_\infty$ and at most until $t_{-1}$ is given in

  \begin{definition}
    For each time $t \in \extendedTimes$, let
    \begin{equation*}
      \Configurations_t^{-1} = \setOf{c \in \Configurations \suchThat t \leq t_{-1}(c)}
      \mathnote{set $\Configurations_t^{-1}$} 
      \index[symbols]{Cnft-1@$\Configurations_t^{-1}$}
    \end{equation*}
    and let
    \begin{align*}
      \gtfNegativeSingularityOfOrderMinusOne \from \extendedTimes &\to \bigcup_{t \in \extendedTimes} (\Configurations_t^{-1} \to \Configurations), \mathnote{map $\gtfNegativeSingularityOfOrderMinusOne$ from $\extendedTimes$ to $\bigcup_{t \in \extendedTimes} (\Configurations_t^{-1} \to \Configurations)$}\index[symbols]{boxast@$\gtfNegativeSingularityOfOrderMinusOne$}\\
      t &\mapsto
        \left[
          \begin{aligned}
            \Configurations_t^{-1} &\to \Configurations,\\
            c &\mapsto
              \begin{dcases*}
                \gtfNonNegativeSingularities(t)(c), &if $t \leq t_\infty(c)$,\\
                \liminf_{\varepsilon \downTo 0} \gtfNonNegativeSingularities(t - (\varepsilon + t_\infty(c)))(c_\varepsilon), &otherwise,
              \end{dcases*}\\
              &\phantom{\mapsto{}}\text{ where } c_\varepsilon = \gtfIgnoreCollisions(\varepsilon)(\gtfNonNegativeSingularities(t_\infty(c))(c)),
          \end{aligned}
        \right]
    \end{align*}
    where the limit inferior is the pointwise limit inferior of set-valued maps, in greater detail, for each family $\family{c_\varepsilon}_{\varepsilon \in \R_{> 0}}$ of set-valued maps from $M$ to $Q$, the pointwise limit inferior $\liminf_{\varepsilon \downTo 0} c_\varepsilon$ is the map $c \from M \to Q$, $m \mapsto \bigcup_{r \in \R_{> 0}} \bigcap_{s \geq r} c_{1/s}(m)$.
  \end{definition}

  \begin{remark}
    For each configuration $c \in \Configurations$, we have $t_\infty(c) \leq t_{-1}(c)$; for each time $t \in \extendedTimes$, we have $\Configurations_t^\infty \subseteq \Configurations_t^{-1}$; and for each configuration $c \in \Configurations$ and each positive real number $\varepsilon$, we have $t_{-1}(c) - (\varepsilon + t_\infty(c)) \leq t_\infty(c_\varepsilon)$.
  \end{remark}

  \begin{remark}
    Let the signal machine be in a configuration $c$ at time $0$. If $t_\infty(\gtfNonNegativeSingularities(t_\infty(c))(c)) = 0$, then there is a singularity of order $-1$ at time $t_\infty(c)$. Otherwise, there is not. In the latter case, as one would hope, $t_{-1}(c) = t_\infty(\gtfNonNegativeSingularities(t_\infty(c))(c))$, and, for each time $t \in \closedInterval{t_\infty(c), t_{-1}(c)}$, we have $\gtfNegativeSingularityOfOrderMinusOne(t)(c) = \gtfNonNegativeSingularities(t - t_\infty(c))(\gtfNonNegativeSingularities(t_\infty(c))(c))$. 
  \end{remark}

  Unlike for a singularity of a non-negative order, for a singularity of order $-1$ it is not possible to jump to the time right after the singularity as there is no such time. However, we may jump over the singularity at time $t_\infty$ to the time $t_{-1}$. This is done by the map 

  \begin{definition}
    \index[symbols]{Deltadot-1@$\gtfJump_{-1}$}\begin{align*}
      \gtfJump_{-1} \from \Configurations &\to \Configurations, \mathnote{map $\gtfJump_{-1}$ from $\Configurations$ to $\Configurations$}\\
      c &\mapsto
        \begin{dcases*}
          c, &if $t_\infty(c) = \infty$,\\
          \gtfNegativeSingularityOfOrderMinusOne(t_{-1}(c))(c), &otherwise. \qedhere
        \end{dcases*}
    \end{align*}
  \end{definition}

  Like for singularities of non-negative orders, such of order $-1$ may accumulate, which is a \index{singularity}\defineX{singularity of order $-2$}{singularity!of order $-j$, for $j \in \Z_{\geq 2}$}\graffito{singularity of order $-j$, for $j \in \Z_{\geq 2}$}, such of order $-2$ my accumulate, which is a \defineX{singularity of order $-3$}{singularity!of order $-j$, for $j \in \Z_{\geq 2}$}, and so forth. The map to jump over singularities of order $-1$ has already been given and the maps to jump over singularities of smaller orders are introduced in

  \begin{definition}
    The sequence $\sequence{t_{-(j - 1)}^n}_{n \in \N_0}$, where the $n$ in $t_{-(j - 1)}^n$ is an upper index and does not stand for exponentiation, the map $t_{-j}$, and the map $\gtfJump_{-j}$, for $j \in \Z_{\geq 2}$, are defined by mutual induction as follows: The maps $t_{-1}$ and $\gtfJump_{-1}$ have already been defined and, for each integer $j$ such that $j \geq 2$, let
    \begin{equation*}
      \sequence*{
        \begin{aligned}
          t_{-(j - 1)}^n \from \Configurations &\to \extendedTimes, \index[symbols]{tj-1nz@$t_{-(j - 1)}^n$}\\
          c &\mapsto \sum_{i = 0}^{n - 1} t_{-(j - 1)}(\gtfJump_{-(j - 1)}^i(c)),
        \end{aligned}
      }_{n \in \N_0} \mathnote{map $t_{-(j - 1)}^n$ from $\Configurations$ to $\extendedTimes$}
    \end{equation*}
    (note that $t_{-(j - 1)}^0 = 0$), let
    \begin{align*}
      t_{-j} \from \Configurations &\to \extendedTimes, \mathnote{map $t_{-j}$ from $\Configurations$ to $\extendedTimes$}\index[symbols]{tjsubscriptz@$t_{-j}$}\\
      c &\mapsto \lim_{n \to \infty} t_{-(j - 1)}^n(c),
    \end{align*}
    and let\index[symbols]{Deltadotj-@$\gtfJump_{-j}$}
    \begin{align*}
      \gtfJump_{-j} \from \Configurations &\to \Configurations, \mathnote{map $\gtfJump_{-j}$ from $\Configurations$ to $\Configurations$}\\
      c &\mapsto
        \left\{
          \begin{aligned}
            &c, \text{ if $t_{-j}(c) = \infty$},\\
            &\liminf_{n \to \infty} \gtfIgnoreCollisions(t_{-j}(c) - t_{-(j - 1)}^n(c))(\gtfJump_{-(j - 1)}^n(c)),\\
            &&\llap{otherwise.} \qedhere
          \end{aligned}
        \right.
    \end{align*}
  \end{definition}

  The machine can be observed for all times before $\lim_{j \to \infty} t_{-j}(c)$, where $c$ is the initial configuration of the machine. This time is given a name in

  \begin{definition}
    For each configuration $c \in \Configurations$, the non-negative real number or infinity $t_{-\infty}(c) = \lim_{j \to \infty} t_{-j}(c)$ is called \graffito{$(-\infty)$-existence time $t_{-\infty}(c)$ of $c$}\defineX{$(-\infty)$-existence time of $c$}{existence time infinity minus@$(-\infty)$-existence time of $c$}\index[symbols]{tinftycminus@$t_{-\infty}(c)$}, and the closed interval $\closedInterval{0, t_{-\infty}(c)}$ is called \graffito{$(-\infty)$-existence interval $\closedInterval{0, t_{-\infty}(c)}$ of $c$}\defineX{$(-\infty)$-existence interval of $c$}{existence interval infinity minus@$(-\infty)$-existence interval of $c$}.
  \end{definition}

  Like for singularities of non-negative orders, repeated applications of powers of the maps $\gtfJump_{-j}$, for decreasing $j \in \N_+$, let us jump over singularities of decreasing negative orders down to order $-1$. The resulting map is given a name in

  \begin{definition}
    For each positive integer $j$, each non-negative integer $k$, and each finite sequence $\sequence{n_i}_{i = -j}^{-k}$ of non-negative integers that is indexed from $-j$ down to $-k$, let
    \begin{equation*}
      \gtfJump_{\sequence{n_i}_{i = -j}^{-k}}
      = \begin{dcases*}
          \identityMap_{\Configurations}, &if $-j < -k$,\\
          \gtfJump_{-j}^{n_{-j}} \after \gtfJump_{-j - 1}^{n_{-j - 1}} \after \dotsb \after \gtfJump_{-k}^{n_{-k}}, &otherwise, 
        \end{dcases*}
      \mathnote{map $\gtfJump_{\sequence{n_i}_{i = -j}^{-k}}$ from $\Configurations$ to $\Configurations$}
      \index[symbols]{Deltadotniijksequenceminus@$\gtfJump_{\sequence{n_i}_{i = -j}^{-k}}$}
    \end{equation*}
    and let
    \begin{equation*}
      t_{\sequence{n_i}_{i = -j}^{-k}} = \sum_{i = -j}^{-k} t_i^{n_i} \after \gtfJump_{\sequence{n_\ell}_{\ell = i - 1}^{-k}}. \qedhere
      \mathnote{map $t_{\sequence{n_i}_{i = -j}^{-k}}$ from $\Configurations$ to $\extendedTimes$}
      \index[symbols]{tniijksequenceminus@$t_{\sequence{n_i}_{i = -j}^{-k}}$}
    \end{equation*}
  \end{definition}

  Like for singularities of non-negative orders, to compute the configuration the machine is in at the $(-\infty)$-existence time of the initial configuration, we can use the maps $\gtfJump_{-j}$ for increasing $j$ to jump from singularities of negative lower orders to singularities of ever higher orders ad infinitum. And, to compute the configuration at time $t$ before the $(-\infty)$-existence time, we can use one of the maps $\gtfJump_{\sequence{n_i}_{i = -1}^{-k}}$ to jump over all the singularities before time $t$ such that the next jump over a singularity of order $-1$ would be beyond time $t$ and then we can use the map $\gtfNegativeSingularityOfOrderMinusOne$ to jump from there to time $t$ whereby we may cross a singularity of order $-1$. The resulting map describes the time evolution of the signal machine before or at $(-\infty)$-existence times, which for the purposes of this treatise is the complete time evolution, and this map is given in

  \begin{definition}
    For each time $t \in \extendedTimes$, the set of configurations whose $(-\infty)$-existence interval contains $t$ is
    \begin{equation*}
      \Configurations_t^{-\infty} = \setOf{c \in \Configurations \suchThat t \leq t_{-\infty}(c)}.
      \mathnote{set $\Configurations_t^{-\infty}$} 
      \index[symbols]{Cnftinfinityminus@$\Configurations_t^{-\infty}$}
    \end{equation*}
    The map
    \begin{align*}
      \globalTransitionFunction \from \extendedTimes &\to \bigcup_{t \in \extendedTimes} (\Configurations_t^{-\infty} \to \Configurations), \mathnote{global transition function $\globalTransitionFunction$ from $\extendedTimes$ to $\bigcup_{t \in \extendedTimes} (\Configurations_t^{-\infty} \to \Configurations)$}\index[symbols]{boxdot@$\globalTransitionFunction$}\\ 
      t &\mapsto
        \left[
          \begin{aligned}
            \Configurations_t^{-\infty} &\to \Configurations,\\
            c &\mapsto
              \left\{
                \begin{aligned}
                  &\liminf_{j \to \infty} \gtfJump_{-j}(c), \text{ if $t = t_{-\infty}(c)$},\\
                  &\gtfNegativeSingularityOfOrderMinusOne(t - t_{\sequence{n_i}_{i = -1}^{-k}}(c))(\gtfJump_{\sequence{n_i}_{i = -1}^{-k}}(c)), \text{ otherwise},
                \end{aligned}
              \right.\\
              &\phantom{\mapsto{}}\text{ for the least } k \in \N_+ \text{ and the } \sequence{n_i}_{i = -1}^{-k} \in \N_0^k\\
              &\phantom{\mapsto{}}\text{ with } t \in \leftClosedAndRightOpenInterval{t_{\sequence{n_i}_{i = -1}^{-k}}(c), t_{(n_{-1} + 1, n_{-2}, n_{-3}, \dotsc, n_{-k})}(c)},
          \end{aligned}
        \right]
    \end{align*}
    is called \define{global transition function}.
  \end{definition} 

  \begin{remark}
    Jérôme Olivier Durand-Lose introduces and studies another notion of signal machines in his paper \enquote{\citetitle*{durand-lose:2008}}\cite{durand-lose:2008}. The notable differences are the following: While our machines are defined over continuum graphs, his machines are defined over the real number line; while our machines may have infinitely many different kinds of signals, his machines may only have signals of a finite number of kinds (which he calls meta-signals); while in a configuration of our machines there may exist infinitely many signals (even at the same point), in configurations of his machines there may only exist finitely many signals; while the time evolution of our machines can be observed beyond singularities of any order, the time evolution of his machines already stops before singularities of order $1$, that is, before accumulations of collisions (as there are no vertices, other events do not exist for his machines).
  \end{remark}


  \section{Firing Squad Signal Machines}
  \label{section:firing-squad-solution}

  In this section, let $\Graph = \ntuple{\Vertices, \Edges, \eendsOf}$ be a non-trivial, finite, and connected undirected multigraph, let $\weightOf$ be an edge weighting of $\Graph$, let $M$ be a continuum representation of $\Graph$, and let $\general$ be a vertex of $M$, which we call \define{general}\graffito{general vertex $\general$}. When we say \emph{path}, we either mean \emph{path in $\Graph$} or \emph{path in $M$}; when we say \emph{longest path}, we either mean \emph{maximum-weight path in $\Graph$} or \emph{longest path in $M$}; and, when we say \emph{path}, we often mean \emph{non-empty direction-preserving path from vertex to vertex}; in all cases it should be clear from the context what is meant.

  From a broad perspective, the signal machine we construct in this section performs the following tasks: It cuts the graph such that the graph turns into a virtual tree; it starts synchronisation of edges as soon as possible and freezes it as late as possible; it determines the midpoints of all non-empty direction-preserving paths from vertices to vertices; it determines which midpoints are the ones of the longest paths; starting from the midpoints of the longest paths, it traverses midpoints of shorter and shorter paths and upon reaching midpoints of edges, it thaws synchronisation of the respective edges; all edges finish synchronisation at the same time with the creation of fire signals that lie dense in the graph (see \cref{figure:determinedMidpointOfOneEdgeOrOfTwoEdges,figure:MazoyerWithOneAndTwoEdges,figure:fsspWithTwoEdgesAndGeneralInBetweenAndAtTheLeft,figure:fsspWithThreeEdgesInARowAndGeneralAtTheSecondVertexFromTheLeft,figure:fsspWithThreeEdgesInARowAndGeneralAtTheSecondVertexFromTheLeft,figure:fsspWithThreeEdgesIncidentToTheSameVertexAndGeneralAtTheVertexAndAtTheLeafOfTheShortestEdge}). A more detailed account is given in

  \begin{figure}
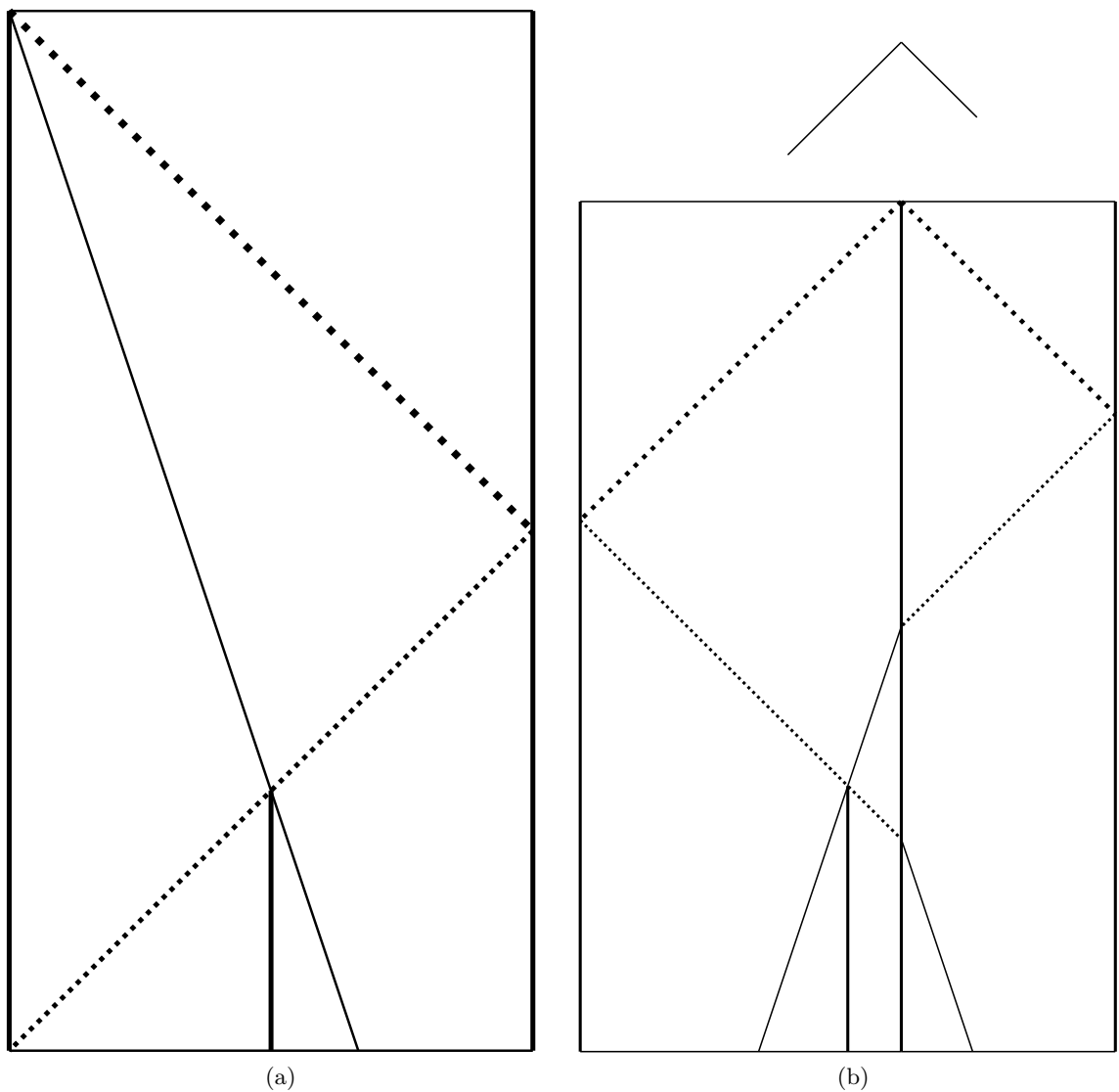

    \begin{wide}
      \mbox{
        \subfloat[]{
          \resizebox{(\linewidth-1em)/2}{!}{\figureMidpointOneEdge}
          \label{figure:determineMidpointOfOneEdge}
        }
        \subfloat[]{
          \resizebox{(\linewidth-1em)/2}{!}{\figureMidpointTwoEdges}
          \label{figure:determineMidpointOfTwoEdges}
        }
      }
      \caption{Both subfigures depict a space-time diagram of the time evolution of the signal machine that we construct in \cref{section:firing-squad-solution}. The diagram in \cref{figure:determineMidpointOfOneEdge} illustrates how, beginning from one end of an edge, the midpoint of the edge is found, which works analogously for paths instead of edges, and the diagram in \cref{figure:determineMidpointOfTwoEdges} illustrates how, beginning from the inner vertex of a path consisting of two edges, the midpoint of the path is found, which works analogously for longer paths; above the right space-time diagram is a scaled-down depiction of the two-edged path. Vertices are thick and solid lines; find-midpoint signals of speed $1$ are thick and dotted lines; reflected find-midpoint signals of speed $1$ are densely dotted lines; slowed-down find-midpoint signals of speed $1/3$ are solid lines; stationary midpoint signals are thick and solid lines.} 
      \label{figure:determinedMidpointOfOneEdgeOrOfTwoEdges}
    \end{wide}
  \end{figure}

  \begin{figure}
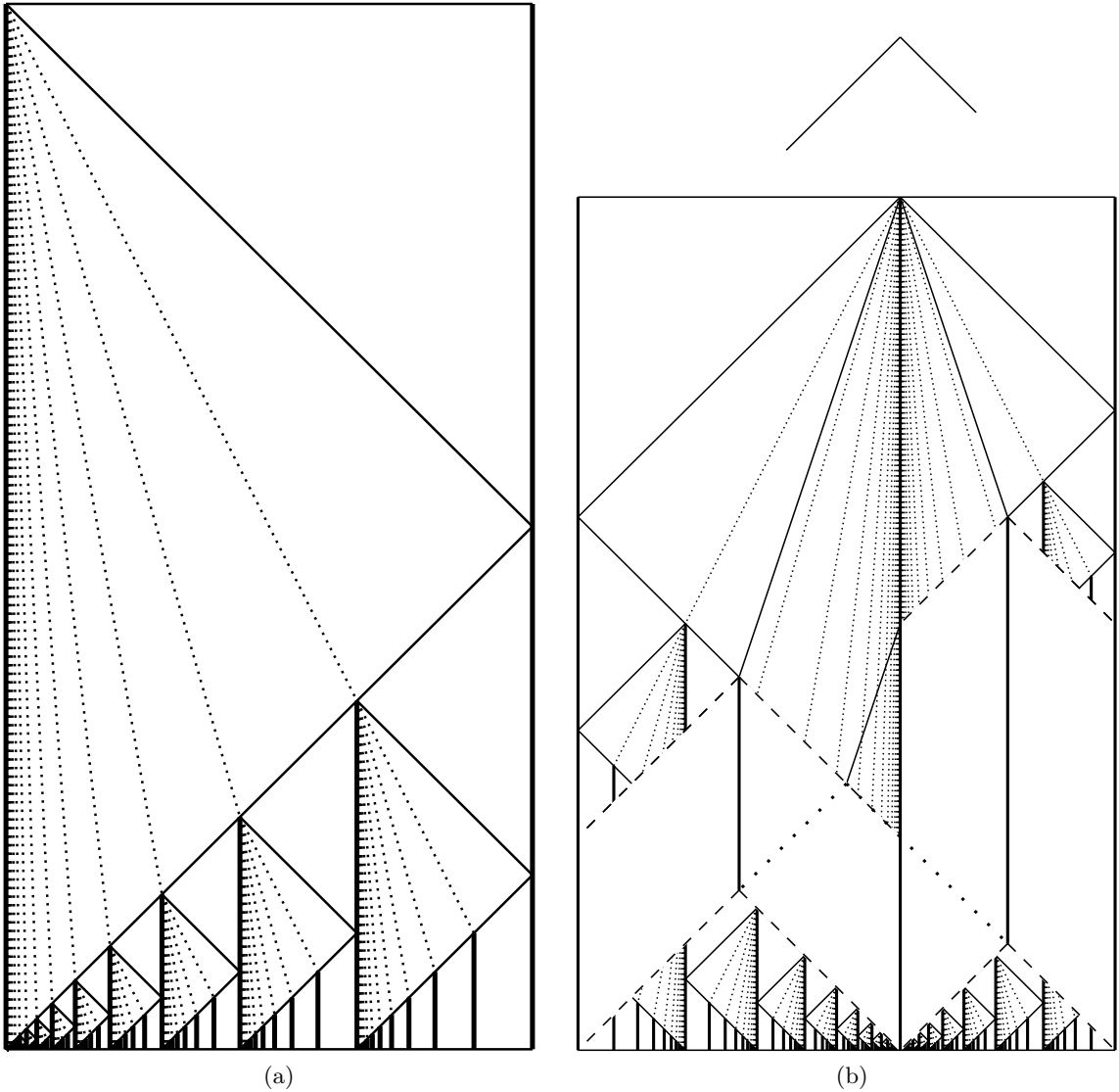

    \begin{wide}
      \mbox{
        \subfloat[]{
          \resizebox{(\linewidth-1em)/2}{!}{\figureMazoyer}
          \label{figure:Mazoyer}
        }
        \subfloat[]{
          \resizebox{(\linewidth-1em)/2}{!}{\figureFSSPMazoyerWithTwoEdges}
          \label{figure:MazoyerWithTwoEdges}
        }
      }
      \caption{Both subfigures depict a space-time diagram of the time evolution of the signal machine that we construct in \cref{section:firing-squad-solution}. The diagram in \cref{figure:Mazoyer} illustrates how, beginning from one end of an edge, the edge is synchronised by recursively dividing it into two parts, where one part is two-thirds the superpart's length and the other part is one-third the superpart's length, and creating fire signals when reflected divide signals reach the boundaries they originated at, which happens for all at the same time and these fire signals lie dense in the edge. And, the diagram in \cref{figure:MazoyerWithTwoEdges} illustrates how, beginning from the inner vertex of a path consisting of two edges, the path is synchronised by synchronising both edges, freezing the synchronisation of individual edges as late as possible and thawing it as early as possible such that both edges finish at the same time. Note that in both subfigures, at each boundary, there is a singularity of order $1$ at the last depicted time, and in \cref{figure:MazoyerWithTwoEdges}, there are additional singularities of order $1$ at time twice the longer edge's length and at time twice the shorter edge's length, and of order $-1$ at various points of the freeze and thaw signals. Only the most relevant signals are depicted and these only for the most relevant time spans. Initiate signals, divide signals of type $0$, and find-midpoint signals, all of speed $1$, are solid lines; reflected divide signals and reflected find-midpoint signals, both of speed $1$, are solid lines; divide signals of type $n \in \N_+$, which have speed $(2/3)^n / (2 - (2/3)^n)$, are densely dotted lines; slowed-down find-midpoint signals of speed $1/3$ are solid lines; stationary boundary and midpoint signals are thick and solid lines; freeze signals of speed $1$ are dashed lines; thaw signals of speed $1$ that do not thaw synchronisation of the edge they are on are thick and loosely dotted lines, and the other thaw signals of speed $1$ are dashed lines.} 
      \label{figure:MazoyerWithOneAndTwoEdges}
    \end{wide}
  \end{figure}

  \begin{figure}
    \begin{wide}
      \mbox{
        \subfloat[]{
          \resizebox{(\linewidth-1em)/2}{!}{\figureFSSPWithTwoEdgesAndGeneralInBetween}
          \label{figure:fsspWithTwoEdgesAndGeneralInBetween}
        }
        \subfloat[]{
          \resizebox{(\linewidth-1em)/2}{!}{\figureFSSPWithTwoEdgesAndGeneralAtTheLeft}
          \label{figure:fsspWithTwoEdgesAndGeneralAtTheLeft}
        }
      }
      \caption{Both subfigures depict the space-time diagram of the time evolution of the signal machine that we construct in \cref{section:firing-squad-solution} for the scaled-down depicted tree with two edges, beginning with the initial configuration for the firing squad synchronisation problem --- which is the configuration in which initiate signals, find-midpoint signals, and slowed-down find-midpoint signals emanate from the general onto all incident edges --- and ending with the final configuration of the problem --- which is the configuration in which the fire signals lie dense in the multigraph ---, where in \cref{figure:fsspWithTwoEdgesAndGeneralInBetween} the general is the vertex that is incident to both edges and in \cref{figure:fsspWithTwoEdgesAndGeneralAtTheLeft} it is the leaf of the longer edge. Only the most relevant signals are depicted and these only for the most relevant time spans. Initiate signals and find-midpoint signals, which in the depicted cases always travel alongside each other, are thick and dotted lines; reflected find-midpoint signals are densely dotted lines; slowed-down find-midpoint signals are solid lines; midpoint signals are thick and solid lines; freeze signals are dashed lines; thaw signals that do not thaw synchronisation of the edge they are on are thick and loosely dotted lines, and the other thaw signals are dashed lines. The synchronisation of individual edges, before it is frozen and after it is thawed, is schematically represented by hatch patterns.} 
      \label{figure:fsspWithTwoEdgesAndGeneralInBetweenAndAtTheLeft}
    \end{wide}
  \end{figure}

  \begin{figure}
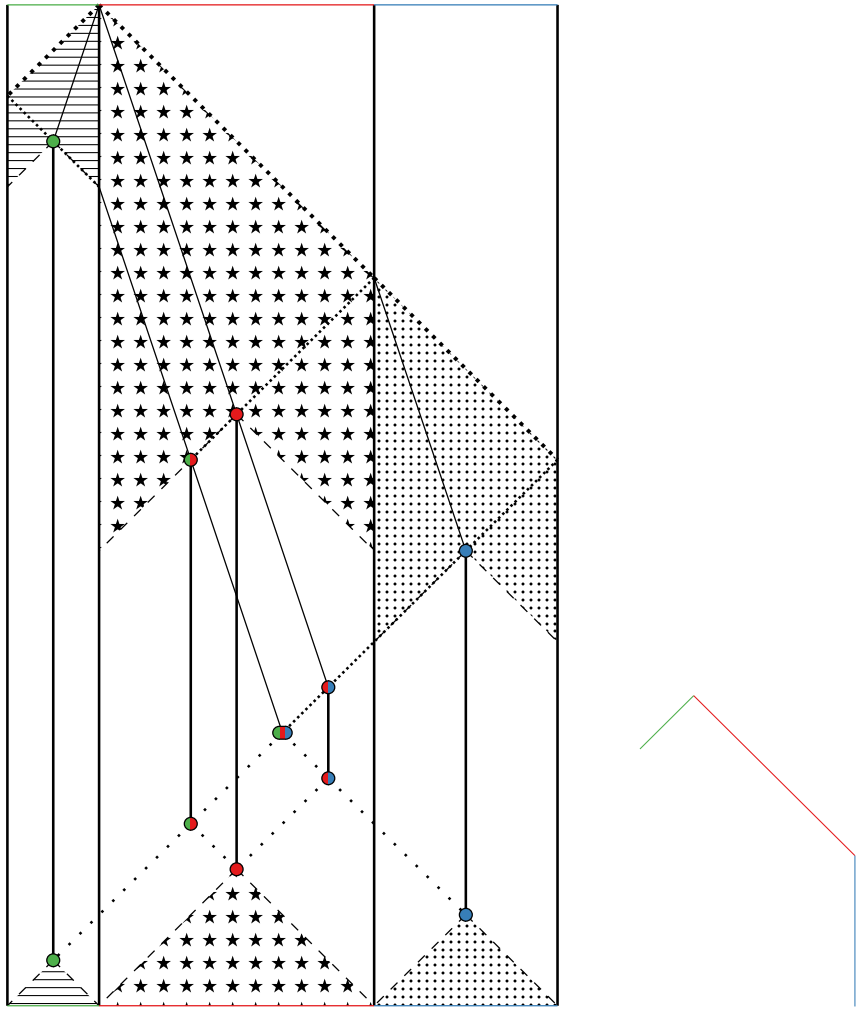

    \begin{wide}
      \resizebox{(\linewidth-1em)/2}{!}{\figureFSSPWithThreeEdgesInARowAndGeneralAtTheSecondVertexFromTheLeft}
      \qquad
      \figureTheTreeThatIsSynchronised
      \caption{The figure depicts the space-time diagram of the time evolution of the signal machine that we construct in \cref{section:firing-squad-solution} for the tree with three edges depicted on the right, beginning with the initial configuration for the firing squad synchronisation problem --- which is the configuration in which initiate signals, find-midpoint signals, and slowed-down find-midpoint signals emanate from the general onto all incident edges --- and ending with the final configuration of the problem --- which is the configuration in which the fire signals lie dense in the multigraph ---, where the general is the second vertex from the left. Only the most relevant signals are depicted and these only for the most relevant time spans. Initiate signals and find-midpoint signals are thick and dotted lines; reflected find-midpoint signals are densely dotted lines; slowed-down find-midpoint signals are solid lines; midpoint signals are thick and solid lines; freeze signals are dashed lines; thaw signals that do not thaw synchronisation of the edge they are on are thick and loosely dotted lines, and the other thaw signals are dashed lines. When a midpoint is found and when the corresponding midpoint signal collides with a matching thaw signal, the path it is the midpoint of is represented by the colours of the path's edges in a little disk or rounded rectangle. The synchronisation of individual edges, before it is frozen and after it is thawed, is schematically represented by hatch patterns.}
      \label{figure:fsspWithThreeEdgesInARowAndGeneralAtTheSecondVertexFromTheLeft}
    \end{wide}
  \end{figure}

  \begin{figure}
    \begin{wide}
      \mbox{
        \subfloat[]{
          \resizebox{(\linewidth-1em)/2}{!}{\figureFSSPWithThreeEdgesIncidentToTheGeneralVertex} 
          \label{figure:fsspWithThreeEdgesIncidentToTheSameVertexAndGeneralAtTheVertex}
        }
        \subfloat[]{
          \resizebox{(\linewidth-1em)/2}{!}{\figureFSSPWithThreeEdgesIncidentToTheSameVertexAndTheGeneralIsNotOnTheLongestPath}
          \label{figure:fsspWithThreeEdgesIncidentToTheSameVertexAndGeneralAtTheShortestEdge}
        }
      }
      \caption{Both subfigures depict the space-time diagram of the time evolution of the signal machine that we construct in \cref{section:firing-squad-solution} for the scaled-down depicted tree with three edges, beginning with the initial configuration for the firing squad synchronisation problem --- which is the configuration in which initiate signals, find-midpoint signals, and slowed-down find-midpoint signals emanate from the general onto all incident edges --- and ending with the final configuration of the problem --- which is the configuration in which the fire signals lie dense in the multigraph. In \cref{figure:fsspWithThreeEdgesIncidentToTheSameVertexAndGeneralAtTheVertex} the general is the one that is incident to all edges; and in \cref{figure:fsspWithThreeEdgesIncidentToTheSameVertexAndGeneralAtTheShortestEdge} it is the leaf of the shortest edge. Only the most relevant signals are depicted and these only for the most relevant time spans. Each edge has a colour, leaves are coloured according to the edge they are incident to, and signals and hatch patterns are coloured according to the edge they are on. The shortest edge is depicted twice, once coloured red and overlaying the green edge, and once coloured violet and overlaying the blue edge; this makes it easier to follow signals that travel from or onto the shortest edge. Initiate signals and find-midpoint signals, which in the depicted cases always travel alongside each other, are thick and dotted lines; reflected find-midpoint signals are densely dotted lines; slowed-down find-midpoint signals are solid lines; midpoint signals are thick and solid lines; freeze signals are dashed lines; thaw signals that do not thaw synchronisation of the edge they are on are thick and loosely dotted lines, and the other thaw signals are dashed lines. When a midpoint is found and when the corresponding midpoint signal collides with a matching thaw signal, the path it is the midpoint of is represented by the colours of the path's edges in a little disk. The synchronisation of individual edges, before it is frozen and after it is thawed, is schematically represented by hatch patterns.}
      \label{figure:fsspWithThreeEdgesIncidentToTheSameVertexAndGeneralAtTheVertexAndAtTheLeafOfTheShortestEdge}
    \end{wide}
  \end{figure}

  \begin{remark} 
    \begin{aenumerate}
      \item Turn the graph into a tree: Initiate signals of speed $1$ spread from the general throughout the graph without making U-turns and vanish at leaves. When they collide on an edge but not in one of its ends, the edge is cut in two at the point of collision by two stationary leaf signals (also called \emph{virtual leaves}), one for each of the two new ends. When they collide in a vertex coming from all incident edges, all edges are cut off by leaf signals; and when the do not come from all incident edges, all but one of the edges the signals come from are cut off. In the latter case, the initiate signal that comes from the edge that is not cut off spreads onto the edges from which no initiate signals came.

            In this way all cycles are eventually broken up and the graph is turned into a virtual tree. Because the virtual leaves are created as soon as possible and because they are treated just like normal leaves, we may and will assume in the description of the other tasks that the graph is a tree.
      \item Start and freeze synchronisation of edges (see \cref{figure:MazoyerWithOneAndTwoEdges}): When initiate signals reach a vertex, for each incident edge, synchronisation of the edge is started from the vertex immediately and freezing of this synchronisation is started from the midpoint of the edge after $3/2$ times the edge's length time units by sending freeze signals of speed $1$ to both ends of the edge and finishes after twice times the edge's length time units.

            Each edge is synchronised by recursively dividing it into two parts, one having two-third its length and the other having one-third its length. This division procedure becomes finer and finer, in other words, divisions accumulate, the closer the time evolution gets to the time $2$ times the edge's length after of the edge's synchronisation started.

            The division of one part is performed by sending divide signals of speed $(2/3)^n / (2 - (2/3)^n)$, for $n \in \N_0$, from one boundary onto the part, reflecting the divide signal of speed $1$ at the other boundary, and, upon collision of the reflected divide signal with a divide signal, letting the reflected divide signal move on, removing the involved divide signal, creating a stationary boundary signal and sending divide signals of the above speeds onto the newly created part the reflected divide signal comes from. 
      \item Determine the midpoints of all direction-preserving paths from vertices to vertices (see \cref{figure:determinedMidpointOfOneEdgeOrOfTwoEdges}): The general and initiate signals that reach vertices send find-midpoint signals of speed $1$ and slowed-down find-midpoint signals of speed $1/3$ in all directions to determine the midpoints of all paths that contain the general and the reached vertex respectively. These signals spread throughout the graph without making U-turns \emph{memorising the paths they take} and (fast) find-midpoint signals are additionally reflected. 

            Reflected find-midpoint signals of speed $1$ take the paths back they took before being reflected \emph{memorising the path to the vertex they were reflected at and the remaining path they have to take}. Upon finishing their path at the vertex they originated at, slowed-down find-midpoint signals of speed $1/3$ are sent onto all edges except the one the reflected signal comes from, \emph{they memorise the path from the origin vertex to the reflection vertex of the reflected find-midpoint signal}, and they spread throughout the graph without making U-turns \emph{memorizing the paths they take} and vanish at leaves.

            When a reflected find-midpoint signal collides with a slowed-down find-midpoint signal that originated at the same vertex, the point of collision is the midpoint of the concatenation of the paths the two signals took after being reflected, that is, the path from the vertex the reflected find-midpoint signal was reflected at over the point of collision over the vertex both signals originated at to the vertex the reflected find-midpoint signal that spawned the slowed-down find-midpoint signal was reflected at. Each midpoint is designated by a stationary midpoint signal that \emph{memorises the path it is the midpoint of along with its position on the path}.

            When a reflected find-midpoint signal collides with a reflected find-midpoint signal that originated at the same vertex and the point of collision is the origin vertex itself, both reflection vertices have the same distance to the origin vertex and this vertex is the midpoint of the concatenation of the paths the two signals took after being reflected. 

            To get a clearer picture of how midpoints are determined, let us focus on only a few signals and let us ignore boundary cases: When an initiate signal reaches a vertex, one find-midpoint signal is sent along one incident edge, and another find-midpoint signal is sent along another incident edge. Both signals travel along edges and upon reaching a vertex they are either reflected and travel back or they take one of the incident edges that leads them further away. At the latest, they are reflected upon reaching a leaf. One of the reflected find-midpoint signals returns first to the vertex it originated at, is slowed down there, and the slowed down signal travels towards the other (reflected) find-midpoint signal. At some time in the future, the slowed-down find-midpoint signal collides with the other, now reflected, find-midpoint signal and the point of collision is the midpoint of the concatenation of the paths the two signals took after being reflected. 
      \item Determine the midpoints of the longest direction-preserving paths (see \cref{figure:fsspWithTwoEdgesAndGeneralInBetweenAndAtTheLeft,figure:fsspWithThreeEdgesInARowAndGeneralAtTheSecondVertexFromTheLeft}): The midpoints of the longest paths are eventually found. But this is not sufficient, they also need to be recognised as such. To that end, each reflected find-midpoint signal and each slowed-down find-midpoint signal carries a boolean that indicates whether the path it took from the vertex it was reflected at may be the subpath of a longest path that has the same source (or target), and whether the signal would be the first one to find the longest path's midpoint (recall that a slowed-down find-midpoint signal was either spawned by the general vertex or an initiate signal that reached a vertex, in which case we can think of this vertex as the one the slowed-down find-midpoint signal was reflected at; or it was originally a reflected find-midpoint signal and knows where that signal was reflected at). Let us call signals that carry the boolean $\yes$ \emph{marked} and the others not.

            At each vertex that is not a leaf, a stationary count signal memorises the directions from which marked reflected find-midpoint signals that originated at the vertex or from which marked slowed-down find-midpoint signals with any origin have already returned. Because longest paths always end at leaves, for each leaf, find-midpoint signals are marked when they are the first ones to be reflected at the leaf and not otherwise. When a marked signal reaches a vertex, it stays marked, if, including itself, from all but one direction have marked signals already returned (which is the case if and only if it is the last signal to return from its direction and the penultimate signal to return among such signals from either direction), and it is unmarked, otherwise.

            This means that a marked reflected find-midpoint signal is unmarked, if there are at least two other signals that started out at the same time at the same vertex but take longer to return (because they have a longer way; the combined paths that two of the other signals take may be a longest path), and it stays marked, if all other signals except for one that all started out at the same time at the same vertex returned earlier (because they had a shorter way).

            And it means that a marked slowed-down find-midpoint signal is unmarked, if it is unclear which of the incident edges belong to longest paths, and it stays marked, otherwise, which is the case if from precisely one direction no slowest reflected find-midpoint signal that originated at the vertex has returned yet. This is not only pessimistic, meaning that we do not falsely consider midpoints of paths that are not among the longest as such, but also correct, meaning that we still find the midpoints of all longest paths in time (see \cref{subsection:midpoints-of-maximum-weight-paths-are-recognised}).

            When a marked reflected find-midpoint signal reaches its origin and is the penultimate such signal to do so, it turns into a marked slowed-down find-midpoint signal that travels in the one direction from which no marked signal has returned yet and the point at which this signal collides with the one signal that has not yet returned will be the midpoint of a longest path, if both signals are still marked at the time of collision.
      \item Traverse midpoints, thaw synchronisation of edges, and fire (see \cref{figure:fsspWithThreeEdgesInARowAndGeneralAtTheSecondVertexFromTheLeft,figure:fsspWithThreeEdgesIncidentToTheSameVertexAndGeneralAtTheVertexAndAtTheLeafOfTheShortestEdge}): The midpoints of the longest paths are found at the same time, at which, from each such midpoint, two thaw signals of speed $1$ are sent that travel along the midpoint's path towards both of its ends. When a thaw signal collides with the midpoint of a path such that one of the two subpaths from the midpoint to either end of the path coincides with the remaining path the thaw signal travels along, an additional thaw signal is created that travels along the other subpath. On its way from the midpoint of the last edge of the path a thaw signal travels on to the end of the path, the thaw signal thaws all frozen signals it collides with. All thaw signals reach the ends of their paths at the same time, which is also the time all edges finish synchronisation with the creation of stationary fire signals that lie dense in the graph. \qedhere
    \end{aenumerate} 
  \end{remark}

  We introduce a typographic convention in

  \begin{definition}
    Each word of letters of the Latin alphabet that is written in typewriter font shall denote the word itself and shall for example not be the name of a variable.
  \end{definition}

  We introduce a boolean algebra in

  \begin{definition} 
    Let $\booleans = \setOf{\no, \yes}$\graffito{booleans $\booleans = \setOf{\no, \yes}$}\index[symbols]{Bblackboard@$\booleans$}. Each element $b \in \booleans$ is called \define{boolean}\index[symbols]{b@$b$}. The map
    \begin{align*}
      \neg \from \booleans &\to \booleans, \mathnote{negation $\neg$}\index[symbols]{negate@$\neg$}\\
      \no &\mapsto \yes,\\
      \yes &\mapsto \no,
    \end{align*}
    is called \define{negation}. The map
    \begin{align*}
      \land \from \booleans \times \booleans &\to \booleans, \mathnote{conjunction $\land$}\index[symbols]{wedge@$\land$}\\
      (b, b') &\mapsto \begin{dcases*}
        \no, &if $\no \in \setOf{b, b'}$,\\
        \yes, &otherwise,
      \end{dcases*}
    \end{align*}
    is called \define{conjunction}.
  \end{definition}

  We introduce finite lists of directions in

  \begin{definition}
    Let $\Directions^*$ be the set $\setOf{w \from \discreteInterval{1}{n} \to \Directions \suchThat n \in \N_0}$\graffito{set $\Directions^*$}\index[symbols]{Dirstarsuperscript@$\Directions^*$}. Each element $w \in \Directions^*$ is called \define{word over $\Directions$}\graffito{word $w$ over $\Directions$}\index[symbols]{w@$w$}; for each word $w \in \Directions^*$, the non-negative integer $\lengthOf{w} = \cardinalityOf{\domainOf(w)}$ is called \define{length of $w$}\graffito{length $\lengthOf{w}$ of $w$}\index[symbols]{absolutew@$\lengthOf{w}$}; the word $\emptyWord \from \emptyset \to \Directions$ is called \define{empty}\graffito{empty word $\emptyWord$}\index[symbols]{lambda@$\emptyWord$}; the map
    \begin{align*}
      \concat \from \Directions^* \times \Directions^* &\to \Directions^*, \mathnote{concatenation $\bullet$}\index[symbols]{bullet@$\concat$}\\
      (w, w') &\mapsto \left[
                         \begin{aligned}
                           \discreteInterval{1}{\lengthOf{w} + \lengthOf{w'}} &\to \Directions,\\
                           i &\mapsto \begin{dcases*}
                                        w(i), &if $i \leq \lengthOf{w}$,\\
                                        w'(i - \lengthOf{w}), &otherwise,
                                      \end{dcases*}
                         \end{aligned}
                       \right]
    \end{align*}
    is called \define{concatenation}.
  \end{definition}

  \begin{remark}
    The empty word is the only word of length $0$ and it is the neutral element of $\concat$.
  \end{remark}

  The signal machine we construct in this section has infinitely many kinds of signals, which are explained in \cref{remark:explanation-of-signals}, and given names, speeds, and data sets in

  \begin{definition}
    Let
    \begin{equation*}
      \Kinds = \setOf{\initiateKind, \leafKind, \countKind, \midpointKind, \findMidpointKind, \reflectedFindMidpointKind, \slowedDownFindMidpointKind, \freezeKind, \thawKind} \cup \parens[\big]{\bigcup_{n \in \N_0} \setOf{\divideKind_n, \frozenDivideKind_n}} \cup \setOf{\reflectedDivideKind, \boundaryKind, \fireKind, \frozenFireKind}, \mathnote{set $\Kinds$} 
    \end{equation*}
    let
    \begin{align*}
      \speedOf \from \Kinds &\to \R_{\geq 0}, \mathnote{map $\speedOf$}\\
            \initiateKind &\mapsto 1,\\
                \leafKind &\mapsto 0,\\
               \countKind &\mapsto 0,\\
            \midpointKind &\mapsto 0,\\
            \findMidpointKind &\mapsto 1,\\
        \reflectedFindMidpointKind &\mapsto 1,\\
            \slowedDownFindMidpointKind &\mapsto \frac{1}{3},\\
              \freezeKind &\mapsto 1,\\
                \thawKind &\mapsto 1,\\
            \divideKind_n &\mapsto \frac{\parens*{\frac{2}{3}}^n}{2 - \parens*{\frac{2}{3}}^n}, \text{ for } n \in \N_0,\\
      \frozenDivideKind_n &\mapsto 0, \text{ for } n \in \N_+,\\
          \reflectedDivideKind &\mapsto 1,\\
            \boundaryKind &\mapsto 0,\\
                \fireKind &\mapsto 0,\\
          \frozenFireKind &\mapsto 0,
    \end{align*}
    and let
    \begin{gather*}
      \NoData = \setOf{0}, \mathnote{set $\NoData$}\\ 
      \Data_{\initiateKind} = \NoData, \mathnote{family $\family{\Data_k}_{k \in \Kinds}$}\\
      \Data_{\leafKind} = \Directions,\\
      \Data_{\countKind} = \powerSetOf(\Directions),\\
      \Data_{\midpointKind} = \setOf{\setOf{w, w'} \subseteq \Directions^* \suchThat w \neq w'},\\
      \Data_{\findMidpointKind} = \Directions^*,\\
      \Data_{\reflectedFindMidpointKind} = \Directions^* \times \Directions^* \times \booleans,\\
      \Data_{\slowedDownFindMidpointKind} = \Directions^* \times \Directions^* \times \booleans,\\
      \Data_{\freezeKind} = \NoData,\\
      \Data_{\thawKind} = \Directions^* \times \booleans,\\
      \Data_{\divideKind_n} = \NoData, \text{ for } n \in \N_0,\\
      \Data_{\frozenDivideKind_n} = \Directions, \text{ for } n \in \N_+,\\
      \Data_{\reflectedDivideKind} = \NoData,\\
      \Data_{\boundaryKind} = \NoData,\\
      \Data_{\fireKind} = \NoData,\\
      \Data_{\frozenFireKind} = \NoData. \qedhere
    \end{gather*}
  \end{definition}

  The kinds together with their speeds and data sets determine the possible signals, which are recalled and given abbreviations in

  \begin{definition}
    The set of signals is
    \begin{equation*}
      \Signals = \setOf{(k, d, u) \suchThat k \in \Kinds \text{, } (\speedOf(k), d) \in \Arrows \text{, and } u \in \Data_k}. \mathnote{set $\Signals$}
    \end{equation*}
    Let
    \begin{gather*}
      \initiateSignal{d} = (\initiateKind, d, 0), \text{ for } d \in \Directions, \mathnote{abbreviations of signals like $\initiateSignal{d}$, $\reflectedFindMidpointSignal{w_o}{d}{w_r}{b}$ and $\frozenFireSignal$}\\
      \leafSignal{d} = (\leafKind, \every, d), \text{ for } d \in \Directions,\\
      \countSignal{D} = (\countKind, \every, D), \text{ for } D \in \Data_{\countKind},\\
      \midpointSignal{w}{w'} = (\midpointKind, \every, \setOf{w, w'}), \text{ for } \setOf{w, w'} \in \Data_{\midpointKind},\\
      \findMidpointSignal{w_o}{d} = (\findMidpointKind, d, w_o), \text{ for } w_o \in \Data_{\findMidpointKind},\\
      \reflectedFindMidpointSignal{w_o}{d}{w_r}{b} = (\reflectedFindMidpointKind, d, (w_o, w_r, b)), \text{ for } d \in \Directions \text{ and } (w_o, w_r, b) \in \Data_{\reflectedFindMidpointKind},\\
      \slowedDownFindMidpointSignal{d}{w_o}{w_r}{b} = (\slowedDownFindMidpointKind, d, (w_o, w_r, b)), \text{ for } d \in \Directions \text{ and } (w_o, w_r, b) \in \Data_{\slowedDownFindMidpointKind},\\
      \freezeSignal{d} = (\freezeKind, d, 0), \text{ for } d \in \Directions,\\
      \thawSignal{d}{w}{b} = (\thawKind, d, (w, b)), \text{ for } d \in \Directions \text{ and } (w, b) \in \Data_{\thawKind},\\
      \divideSignal{n}{d} = (\divideKind_n, d, 0), \text{ for } n \in \N_0 \text{ and } d \in \Directions,\\
      \frozenDivideSignal{n}{d} = (\frozenDivideKind_n, \every, d), \text{ for } n \in \N_+ \text{ and } d \in \Directions,\\
      \reflectedDivideSignal{d} = (\reflectedDivideKind, d, 0), \text{ for } d \in \Directions,\\
      \boundarySignal = (\boundaryKind, \every, 0),\\
      \fireSignal = (\fireKind, \every, 0),\\
      \frozenFireSignal = (\frozenFireKind, \every, 0). \qedhere
    \end{gather*}
  \end{definition}

  What signals of various kinds do when they reach a vertex or collide with one another, what the data they carry means, and what we call them is given a glimpse at in

  \begin{remark} 
  \label{remark:explanation-of-signals}
    \begin{aenumerate}
      \item Each signal of kind $\initiateKind$ has speed $1$; at each vertex it reaches it spreads in all directions that lead away from where it comes from, it initiates synchronisation of all incident edges except the one it comes from by sending divide signals onto them, it initiates the search for midpoints of paths that contain the vertex by sending (slowed-down) find-midpoint signals in all directions, and it initiates one component of the search for the longest paths of the graph by marking slowed-down find-midpoint signals if the vertex is a leaf; it carries no data; and it is called \graffito{initiate signal}\define{initiate signal}\index{signal!initiate}. The very first initiate signals spread from the general in all directions.
      \item Each signal of kind $\leafKind$ is stationary, designates a virtual leaf, carries the direction that leads onto the edge that is incident to the virtual leaf, and is called \define{leaf signal}\index{signal!leaf}\graffito{leaf signal}. Such signals are created when initiate signals collide in a vertex (or on an edge), which means that there is a cycle in the graph, and this cycle is broken up by virtually terminating the involved edge(s) with leaf signals. Each leaf signal is treated like a leaf in the following way: When signals collide with each other and with leaf signals, for each involved leaf signal, the collision of the signals that move in the opposite direction than the one the leaf signal carries is handled as if those signals collided in a leaf. Because leaf signals are created at points at the same time or before any other signal reaches them, the graph looks like a tree for all other signals.
      \item Each signal of kind $\countKind$ is stationary; is positioned at a vertex that is not a leaf; memorises the directions from which find-midpoint signals that originated at the vertex, were reflected, and may be on longest paths and would be the first ones to find their midpoints have already returned, in other words, it memorises the directions from which the slowest find-midpoint signals that originated at the vertex and travelled alongside initiate signals before they were reflected at a leaf have already returned; and is called \define{count signal}\index{signal!count}\graffito{count signal}. When an initiate signal reaches a vertex that is not a leaf, a count signal is created. 

            Note that, for the data set of count signals, instead of the infinite set $\powerSetOf(\Directions)$, we could have used the finite set $\powerSetOf(\setOf{1, 2, \dotsc, k})$, where $k$ is the maximum degree of the graph or the upper bound of the maximum degrees of the graphs to be considered and the numbers represent directions that lead away from vertices. The first choice of $k$ would make the signal machine depend on the graph, which is unconventional, whereas the latter would not and would also fit to the fact that solutions of firing mob synchronisation problems are usually considered for graphs whose maximum degrees are uniformly bounded by a constant.
      \item Each signal of kind $\midpointKind$ is stationary, designates the midpoint of a path, carries the directions that lead from its position to both ends of the path, and is called \define{midpoint signal}\index{signal!midpoint}\graffito{midpoint signal}. Such a signal is created when a reflected find-midpoint signal collides with a slowed-down find-midpoint signal that originated at the same vertex, or when two reflected find-midpoint signals that originated at the same vertex collide with each other, which only happens at the origin vertex itself. See \cref{figure:determinedMidpointOfOneEdgeOrOfTwoEdges}.
      \item Each signal of kind $\findMidpointKind$ has speed $1$, at each vertex it reaches it spreads in all directions that lead away from where it comes from (in the sense that, in each such direction, a signal of its kind is sent) and it is also reflected (in the sense that a reflected find-midpoint signal is sent in the direction from where it comes from), it carries the directions that lead from its position to the vertex the signal originated at, and it is called \define{find-midpoint signal}\index{signal!find-midpoint}\graffito{find-midpoint signal}. When an initiate signal reaches a vertex, for each incident edge, a find-midpoint signal whose origin is the vertex is created that travels onto the edge.
      \item Each signal of kind $\reflectedFindMidpointKind$ has speed $1$; is the reflection of a find-midpoint signal at a vertex, travels back along the path this signal took before it was reflected and slows down when it reaches the vertex the find-midpoint signal originated at; carries the directions that lead from its position to the vertex the find-midpoint signal originated at, the directions that lead from its position to the vertex the find-midpoint signal was reflected at, and a boolean that indicates whether the path described by its position and both directions, which leads from the reflection vertex to the origin vertex, may be the subpath of a longest path that has the same source (or target), and the boolean also indicates whether the signal would be the first one to find the longest path's midpoint; and is called \define{reflected find-midpoint signal}\index{signal!reflected find-midpoint}\index{signal!reflected find-midpoint!marked}\graffito{(marked) reflected find-midpoint signal} and, if the boolean it carries is $\yes$, it is called \define{marked}.

            As has already been pointed at, when a find-midpoint signal reaches a vertex, a reflected find-midpoint signal is created that travels onto the edge the find-midpoint signal comes from. If the vertex is a leaf and the find-midpoint signal is one of the first signals to reach it, which is precisely the case if the signal reaches the leaf together with an initiate signal, then its reflection is marked, and otherwise, not. The reasons are that both ends of a longest path are leaves and that a find-midpoint signal that is not among the first signals to reach one end of a longest path would not find its midpoint after another one has already found it.

            When a marked reflected find-midpoint signal reaches a vertex that is not a leaf, the count signal at the vertex memorises the direction the marked signal comes from, and the signal stays marked, if the memory of the count signal contains each but one direction that leads away from the vertex, and it is unmarked, otherwise. Why is that? Each vertex that is not the general is reached precisely once by an initiate signal, at which point find-midpoint signals are sent in all directions; for each direction, the marked reflected find-midpoint signal to return from that direction is memorised, which is the slowest one or, in other words, the last one or the one that had the longest way (note that although only one find-midpoint signal is sent in a direction, multiple reflected find-midpoint signals may return from that direction); the penultimate marked reflected find-midpoint signal to return may come from one edge of a longest path that runs through the vertex and hence it stays marked (note that the other edge of the longest path that is incident to the vertex would be the one from which the marked signal has not yet returned); the signals that return before the penultimate one are too fast to be on a longest path and the last signal to return has already collided with the slowed-down penultimate signal that returned before it (if they do not return at the same time) and hence they are unmarked.
      \item Each signal of kind $\slowedDownFindMidpointKind$ has speed $1/3$; is the slow-down of a reflected find-midpoint signal at the vertex the find-midpoint signal originated at; at each vertex it reaches it spreads in all directions that lead away from where it comes from; it carries the directions that lead from its position to the vertex the find-midpoint signal originated at, the directions that lead from the origin vertex to the vertex the find-midpoint signal was reflected at, and a boolean that indicates whether the path described by its position and both directions, which leads from the reflection vertex over the origin vertex to its position, may be the subpath of a longest path that has the same source (or target), and the boolean also indicates whether the signal would be the first one to find the longest path's midpoint; and it is called \define{slowed-down find-midpoint signal}\index{signal!slowed-down find-midpoint}\index{signal!slowed-down find-midpoint!marked}\graffito{(marked) slowed-down find-midpoint signal} and, if the boolean it carries is $\yes$, it is called \define{marked}.

            As has already been pointed at, when a reflected find-midpoint signal reaches the vertex the find-midpoint signal originated at, for each incident edge except the one the reflected find-midpoint signal comes from, a slowed-down find-midpoint signal is created that travels onto the edge. Additionally, when an initiate signal reaches a vertex, for each incident edge, a slowed-down find-midpoint signal is created that travels onto the edge.

            The latter case is in the following senses the boundary or limiting case of the former: Imagine that the vertex the initiate signal reaches is in fact two vertices that are infinitesimally close; then a find-midpoint signal is created at one vertex, this signal immediately reaches the infinitesimally close other vertex, there it is reflected, the reflected find-midpoint signal immediately reaches the infinitesimally close other vertex, and there it is slowed down. Or, analogously, imagine the limit of the cascade of the creation of a find-midpoint signal, its reflection, and slow-down for shorter and shorter distances between the vertex the find-midpoint signal originates at and the one it is reflected at; then in the limit the find-midpoint signal and its reflection vanish and only the slowed-down find-midpoint signal remains.

            When a marked slowed-down find-midpoint signal reaches a vertex that is not a leaf, the count signal at the vertex memorises the direction the marked signal comes from, and the signal stays marked, if the memory of the count signal contains each but one direction that leads away from the vertex, and it is unmarked, otherwise. Why is that? The slowed-down signal reaches the vertex at the same time and from the same direction as the slowest reflected find-midpoint signal from that direction that originated at the vertex. The latter signal is however not marked, because it did not travel alongside initiate signals before it was reflected at a leaf (the reason is that if it had travelled alongside initiate signals, then it would have been reflected at the same time as the find-midpoint signal whose reflection turned into the marked slowed-down find-midpoint signal and hence, because the paths from the reflection leaves to the vertex they reach together have the same lengths, the find-midpoint signal that reaches it slowed down would have taken longer, and therefore the signals would not reach the vertex at the same time).

            Therefore, the memory of the count signal contains each but one direction if and only if from each but one direction the slowest reflected find-midpoint signals that originated at the vertex have already returned. If this is the case, then the incident edge belonging to that direction may be the edge of a longest path that runs through the vertex and the slowed-down find-midpoint signal may be the one to collide with the not yet returned signal somewhere on or beyond the edge precisely at the midpoint of the longest path. If from more than one direction signals are overdue, then the paths running through each pair of these directions are longer than the paths running through any of these directions and the direction the slowed-down find-midpoint signal comes from. And, if all signals have already returned, then they have already collided with the slowed-down signal and found the midpoints of the longest paths whose determination involves the slowed-down signal if there are any. 
      \item Each signal of kind $\freezeKind$ has speed $1$, is created at the midpoint of an edge, moves towards one end of the edge, and freezes synchronisation of the edge, carries no data, and is called \define{freeze signal}\index{signal!freeze}\graffito{freeze signal}. When an initiate signal reaches a vertex, for each incident edge, a find-midpoint signal and a slowed-down find-midpoint signal are created that travel onto the edge, the former is reflected at the other end of the edge and collides with the latter at the midpoint of the edge, at which point two freeze signals are created that travel to both ends of the edge. See \cref{figure:fsspWithTwoEdgesAndGeneralInBetweenAndAtTheLeft}
      \item Each signal of kind $\thawKind$ has speed $1$, is created at the midpoint of a path, travels along the path towards one of end of the path, creates a new signal of its kind when it collides with the midpoint signal that designates the midpoint of a path such that one of the two subpaths from the midpoint to either end of the path (a \graffito{half-path}\define{half-path}\index{path!half-}) coincides with the path it takes itself and the new signal travels along the other half-path, and thaws synchronisation of an edge if it collided with or was created at the midpoint signal that designates the midpoint of the last edge of the path it takes, carries the directions of the path it takes and a boolean that indicates whether it thaws synchronisation of the edge it is on or not, and is called \define{thaw signal}\index{signal!thaw}\graffito{thaw signal}.

            The first thaw signals are created simultaneously at the midpoints of longest paths. For each such midpoint, two thaw signals are created, one that travels along the path to one end of the path and the other that travels to the other end of the path. When a thaw signal collides with the midpoint of a path whose one half-path coincides with the remaining path the thaw signal travels along, an additional thaw signal is created that travels along the other half-path. On its way from the midpoint of the last edge of the path a thaw signal travels on to the end of the path, the thaw signal thaws all frozen signals it collides with.

            In this way, starting at the midpoints of longest paths, thaw signals traverse the midpoints of shorter and shorter paths, reach the ends of their paths at the same time, and thaw synchronisation of edges, which finishes at the same time. See \cref{figure:fsspWithTwoEdgesAndGeneralInBetweenAndAtTheLeft,figure:fsspWithThreeEdgesInARowAndGeneralAtTheSecondVertexFromTheLeft,figure:fsspWithThreeEdgesIncidentToTheSameVertexAndGeneralAtTheVertexAndAtTheLeafOfTheShortestEdge}
      \item Each signal of kind $\divideKind_0$ has speed $1$, moves from one boundary (which may be one end of an edge or a boundary signal) to the next boundary (which may be the other end of the edge or a boundary signal) and is reflected there, carries no data, and is called \define{divide signal of type $0$}\graffito{divide signal of type $0$}. When an initiate signal reaches a vertex, for each incident edge, a divide signal of type $0$ is created that travels onto the edge. And, when a divide signal of any type collides with a reflected divide signal, a divide signal of type $0$ is created that travels in the same direction as the (non-reflected) divide signal.
      \item Each signal of kind $\divideKind_n$, for $n \in \N_+$, has speed $(2/3)^n / (2 - (2/3)^n)$, moves from one boundary (which may be one end of an edge or a boundary signal) towards the next boundary (which may be the other end of the edge or a boundary signal) but never reaches it and can be frozen, carries no data, and is called \graffito{divide signal of type $n$}\define{divide signal of type $n$}. When an initiate signal reaches a vertex, for each incident edge, and for each $n \in \N_+$, a divide signal of type $n$ is created that travels onto the edge. And, when a divide signal of any type collides with a reflected divide signal, for each $n \in \N_+$, a divide signal of type $n$ is created that travels in the same direction as the (non-reflected) divide signal.

            Note that although $\divideKind_n$, for $n \in \N_+$, are different kinds, events that involve signals of these kinds are handled the same way, in other words, signals of these kinds are not differentiated by the two local transition functions of the signal machine. The only reason they are different kinds is because we need them to have different speeds and by definition all signals of the same kind have the same speed.
      \item Each signal of kind $\frozenDivideKind_n$, for $n \in \N_+$, has speed $0$, is a frozen divide signal of type $n$, carries the direction the non-frozen divide signal had, and is called \define{frozen divide signal of type $n$}\graffito{frozen divide signal of type $n$}. When a freeze signal collides with or is created at the same time as a divide signal of type $n \in \N_+$, the divide signal is frozen.
      \item Each signal of kind $\reflectedDivideKind$ has speed $1$, is the reflection of a divide signal of type $0$, creates a boundary signal when it collides with a divide signal of type $n$, creates a fire signal when it reaches the end of the edge it traverses, carries no data, and is called \graffito{reflected divide signal}\define{reflected divide signal}\index{signal!reflected divide}. When a divide signal of type $0$ reaches a boundary (which may be a vertex or a boundary signal), a reflected divide signal is created that travels in the opposite direction.
      \item Each signal of kind $\boundaryKind$ is stationary, designates a boundary for the synchronisation of an edge, carries no data, and is called \graffito{boundary signal}\define{boundary signal}\index{signal!boundary}. Such signals are created when divide signals collide with reflected divide signals.

            On each edge, the interplay of divide signals, reflected divide signals, and boundary signals has the following effect: At first a divide signal of type $1$ collides with a reflected divide signal that originated at the same end of the edge. This collision results in the creation of a boundary signal that divides the edge into two parts. The length of the part from the origin vertex to the boundary signal is $2/3$ times the length of the edge and the length of the part from the boundary signal to the other end of the edge is $1/3$ times the length of the edge.

            In the same manner as the edge itself, the $(1/3)$-part is recursively divided further and further. In the $(2/3)$-part, a signal of type $2$ collides with the reflected divide signal from before. This collision results in the creation of a boundary signal that divides the $(2/3)$-part into two subparts. One has $(2/3) \cdot (2/3)$ the length of the edge and the other has $(1/3) \cdot (2/3)$ the length of the edge. The $((1/3) \cdot (2/3))$-part is recursively divided further and further. The $((2/3) \cdot (2/3))$-part is divided into a $((2/3) \cdot (2/3) \cdot (2/3))$-part and a $((1/3) \cdot (2/3) \cdot (2/3))$-part and so forth.

            If there were no freeze and thaw signals, after twice the time the edge is long --- which is precisely the time it took the divide signal of type $0$ to reach the other end of the edge, to be reflected there, and to return to the end it originated at --- the boundary signals together are dense on the edge, which means that each point on the edge is arbitrarily close to a boundary signal, and, at this point in time, each boundary signal collides with a reflected divide signal, which results in the creation of fire signals that designate that synchronisation has finished. 

            However, because the synchronisation of each edge is started at different times and takes different times depending on how far away the edge is from the general and how long the edge is, synchronisation of each edge is frozen at the last possible moment --- the freezing starts from the midpoint of the edge $3/2$ times the edge's length many time units after synchronisation of the edge was initiated --- and it is thawed such that all edges finish synchronisation at the same time --- the thawing starts from the midpoint of the edge $1/3$ times the edge's length many time units before the total synchronisation finishes, which is the sum of the radius of the graph and its diameter. Recall that the radius is the longest distance from the general to another vertex and that the diameter is the longest distance between two vertices.

            Note that, for each edge, collisions of divide signals with reflected divide signals and with boundary signals accumulate at the times the two freeze signals and the two thaw signals reach the ends of the edge. These accumulations are singularities of order $1$. See \cref{figure:MazoyerWithOneAndTwoEdges}.
      \item Each signal of kind $\fireKind$ is stationary, designates that synchronisation has finished and can be frozen, carries no data, and is called \define{fire signal}\index{signal!fire}\graffito{fire signal}. Such signals are created when reflected divide signals collide with boundary signals or reach vertices.
      \item Each signal of kind $\frozenFireKind$ is stationary, is a frozen fire signal, carries no data, and is called \define{frozen fire signal}\index{signal!frozen fire}\graffito{frozen fire signal}. On each edge, one of the two freeze signals reaches an end of the edge at the same time as the reflected divide signal, at which point a frozen fire signal is created; it is thawed at the same time at which all other fire signals are created, which happens on all edges at the same time and the fire signals lie dense in the multigraph. \qedhere 
    \end{aenumerate}
  \end{remark}

  In the forthcoming definitions of maps we make extensive use of pattern matching. To make the exposition concise and readable we introduce some pattern matching conventions in

  \begin{definition}
    \begin{aenumerate}
      \item In\graffito{order matters} the case that patterns of multiple rules overlap, the rule that occurs first is the one to use. For example, the map $f \from \Z \to \Z$, $0 \mapsto 1$, $1 \mapsto 0$, $z \mapsto z$, maps $0$ to $1$ and $1$ to $0$ and each other integer to itself.
      \item In\graffito{wildcard $\blank$}\index[symbols]{underscore@$\blank$} the case that we do not care to name some part of a matched structure, we write $\blank$ instead of a name for the part. For example, the pattern $(E, \blank, \initiateSignal{d})$ matches each triple whose last component is an initiate signal, gives the first component the name $E$, does not give a name to the second component, and gives the direction of the initiate signal the name $d$. And, the pattern $\findMidpointSignal{w_o}{\blank}$ matches each find-midpoint signal, gives the directions to the vertex the signal originated at the name $w_o$, but gives no name to the direction of the signal.
      \item To\graffito{same names express equality} express equality of different parts of a matched structure, we give those parts the same name. For example, the pattern $\setOf{\reflectedDivideSignal{\reverse d}, \divideSignal{n}{d}}$ matches each set that consists of a reflected divide signal and a divide signal of any type such that the direction of the reflected divide signal is the opposite of the direction of the divide signal, gives the direction of the divide signal the name $d$, and gives its type the name $n$.
      \item To\graffito{$@$-notation} name both a structure and its parts, we use a Haskell-like $@$-notation. For example, the pattern $s@\reflectedFindMidpointSignal{\emptyWord}{d}{\emptyWord}{b}$ matches each reflected find-midpoint signal whose directions to the vertex the signal originated at and to the vertex the signal was reflected at are empty, gives the direction of the signal the name $d$, gives the boolean that indicates whether the signal may be on a longest path and would be the first to find its midpoint the name $b$, and gives the signal itself the name $s$. \qedhere
    \end{aenumerate}
  \end{definition}

  Some of the maps we define below are actually partial maps. We represent them by (total) maps as specified in

  \begin{definition}
    Let\graffito{representations of partial maps using the bottom symbol $\bot$} $X$ and $X'$ be two sets, let $Y$ be a subset of $X$, let $\bot$ be an element that is not in $X'$, which we call \define{bottom}, and let $f$ be a map from $X$ to $X' \cup \setOf{\bot}$ such that, for each element $x \in X$, we have $f(x) = \bot$ if and only if $x \notin Y$. The map $f$ represents a partial map whose domain of definition is $Y$, whose domain is $X$, and whose codomain is $X'$.

    In the following, for maps like $f$, we do not explicitly specify the domain $Y$ of definition (it is the set $X \smallsetminus f^{-1}(\bot)$) and we implicitly assume that $\bot$ does not occur in the codomain $X'$ of the represented partial map. 
  \end{definition}

  To define the local transition functions, we begin with definitions for special cases and use those to gradually arrive at definitions for the general case. For trees and without freezing and thawing, the map $\localTransitionFunction_{v, 1}^{\tree}$ handles the event that precisely one signal reaches a vertex, and the maps $\localTransitionFunction_{v, 2}^{\tree}$ and $\localTransitionFunction_{e, 2}^{\tree}$ handle the event that precisely two signals collide in a vertex and an edge respectively (see \cref{definition:for-trees:single-signals-reach-a-vertex-and-pairs-collide}). For trees, the maps $\localTransitionFunction_v^{\tree}$ and $\localTransitionFunction_e^{\tree}$ handle events involving arbitrarily many signals by considering unordered pairs of signals and applying $\localTransitionFunction_{v, 1}^{\tree}$, $\localTransitionFunction_{v, 2}^{\tree}$, and $\localTransitionFunction_{e, 2}^{\tree}$, and by also freezing and thawing signals if needed (see \cref{definition:for-trees:local-transition-functions}). For virtual trees, which means that edges of the graph have been virtually cut by leaf signals to remove circles, the maps $\localTransitionFunction_v^{\virtualTree}$ and $\localTransitionFunction_e^{\virtualTree}$ handle events by partitioning signals at virtual cuts into those belonging to one or the other leaf and applying $\localTransitionFunction_v^{\tree}$ and $\localTransitionFunction_e^{\tree}$ (see \cref{definition:for-virtual-trees:local-transition-functions}). For general graphs, the maps $\localTransitionFunction_v$ and $\localTransitionFunction_e$ handle events by virtually cutting the graph, which eventually creates a virtual tree, and applying $\localTransitionFunction_v^{\virtualTree}$ and $\localTransitionFunction_e^{\virtualTree}$ (see \cref{definition:for-graphs:local-transition-functions}).

  Most of the forthcoming definitions and parts of them are annotated with intuitive explanations of what they mean. For example, after each rule that handles a specific kind of event, it is explained what kind of event in the time evolution of the signal machine is handled, how it is handled, and sometimes why.

  How events for trees and without freezing and thawing, with one or two signals involved are handled is given in

  \begin{definition}
  \label{definition:for-trees:single-signals-reach-a-vertex-and-pairs-collide}
    The following map tells whether two words of directions are both empty:
    \begin{align*}
      \areEmpty \from \Directions^* \times \Directions^* &\to \booleans, \mathnote{map $\areEmpty$}\\
      (w, w') &\mapsto \begin{dcases*}
                         \yes, &if $\lengthOf{w} = 0$ and $\lengthOf{w'} = 0$,\\
                         \no, &otherwise.
                       \end{dcases*}
    \end{align*}

    For trees and without freezing and thawing, the case that precisely two signals collide on an edge but not in one of its ends is handled by the following map, which maps colliding unordered pairs of signals for which a collision rule is specified to the resulting signals and all other sets of signals to $\bot$:
    \begin{align*} 
      \localTransitionFunction_{e, 2}^{\tree} \from \powerSetOf(\Signals) &\to \powerSetOf(\Signals) \cup \setOf{\bot}, \mathnote{map $\localTransitionFunction_{e, 2}^{\tree}$}\\
      \setOf{\divideSignal{0}{d}, \boundaryKind}
          &\mapsto \setOf{\reflectedDivideSignal{\reverse d}, \boundarySignal}, 
          \intertext{(If a divide signal of type $0$ collides with a boundary signal, then reflect the divide signal.)}
      \setOf{\reflectedDivideSignal{\reverse d}, \divideSignal{n}{d}}
          &\mapsto \setOf{\boundarySignal} \cup \setOf{\divideSignal{n'}{d} \suchThat n' \in \N_0},
          \intertext{(If a reflected divide signal collides with a divide signal of any type, then create a boundary signal and send divide signals of all types in the direction of the original divide signal.)}
      \setOf{s@\reflectedFindMidpointSignal{\emptyWord}{d}{\emptyWord}{b}, s'@\slowedDownFindMidpointSignal{\reverse d}{\emptyWord}{\emptyWord}{b'}}
          &\mapsto\\
                   &&\llap{$\begin{dcases*}
                     \setOf{s, s', \midpointSignal{\reverse d}{d}, \freezeSignal{\reverse d}, \freezeSignal{d}}, &if $b = \no$ or $b' = \no$,\\
                     \emptyset, &otherwise,
                   \end{dcases*}$}
          \intertext{(If a reflected find-midpoint signal collides with a slowed-down find-midpoint signal that originated at the same end of an edge and only travelled on this edge, then the point of collision is the midpoint of the edge and, if the graph has at least two edges, then let the signals move on, designate the point by a midpoint signal, and send freeze signals to both ends of the edge to freeze synchronisation of the edge, and otherwise, do not create and send any signals.)}
      \setOf{s@\reflectedFindMidpointSignal{w_o}{d}{w_r}{b}, s'@\slowedDownFindMidpointSignal{\reverse d}{w_o}{w_r'}{b'}}
          &\mapsto\\
                   &&\llap{$\begin{dcases*}
                     \setOf{s, s', \midpointSignal{(\reverse d) \concat w_r}{d \concat w_o \concat w_r'}}, &if $b = \no$ or $b' = \no$,\\
                     \setOf{\thawSignal{\reverse d}{w_r}{\no}, \thawSignal{d}{w_o \concat w_r'}{\no}}, &otherwise,
                   \end{dcases*}$}
          \intertext{(If a reflected find-midpoint signal collides with a slowed-down find-midpoint signal that originated at the same vertex, then the point of collision is the midpoint of the shortest path in the virtual tree from the vertex the reflected find-midpoint signal was reflected at to the vertex the signal that spawned the slowed-down find-midpoint signal was reflected at and, if this midpoint is not the midpoint of a longest path, then let the signals move on and designate the point by a midpoint signal, and otherwise, send thaw signals along that longest path to both its ends.)}
      \setOf{\thawSignal{d}{w}{\no}, \midpointSignal{d \concat w}{d' \concat w'}}
          &\mapsto \setOf{\thawSignal{d}{w}{\areEmpty(w, w')}, \thawSignal{d'}{w'}{\areEmpty(w, w')}},
          \intertext{(If a thaw signal collides with a midpoint signal that designates the midpoint of a path whose one half-path coincides with the path the thaw signal is going to take, then send an additional thaw signal along the other half-path and, if this path is just a directed edge, then make the thaw signals thaw synchronisation of the edge.)}
      \blank &\mapsto \bot.
          \shortintertext{(If none of the above happened, then indicate that by returning bottom.)}
    \end{align*}

    A signal that reaches a vertex is in a leaf if and only if there is precisely one direction that leads away from the vertex. And a signal that reaches a vertex is the penultimate one to do so if and only if the number of signals that have already returned including the signal itself is one less than the number of directions that lead away from the vertex. The two maps that express this in an abstract way using booleans are
    \begin{align*} 
      \isInLeaf \from \powerSetOf(\Directions) &\to \booleans, \mathnote{map $\isInLeaf$}\\
      E &\mapsto \begin{dcases*}
                   \no,  &if $\cardinalityOf{E} \neq 1$,\\
                   \yes, &otherwise,
                 \end{dcases*}
    \end{align*}
    and
    \begin{align*}
      \isPenultimate \from \powerSetOf(\Directions) \times \N_0 &\to \booleans, \mathnote{map $\isPenultimate$}\\
      (E, n) &\mapsto \begin{dcases*}
                        \no,  &if $\cardinalityOf{E} - 1 \neq n$,\\
                        \yes, &otherwise.
                      \end{dcases*}
    \end{align*}

    For trees and without freezing and thawing, the case that precisely one signal reaches a vertex is handled by the following map, which maps quadruples --- consisting of first, the set of directions that lead away from the vertex; secondly, the number of the directions from which the slowest reflected find-midpoint signals that originated at the vertex have already returned or have just arrived; thirdly, a boolean that indicates whether the signal is among the first ones to reach the vertex; and lastly, the signal that reaches the vertex --- to the resulting signals:
    \begin{align*}
      \localTransitionFunction_{v, 1}^{\tree} \from \powerSetOf(\Directions) \times \N_0 \times \booleans \times \Signals &\to \powerSetOf(\Signals), \mathnote{map $\localTransitionFunction_{v, 1}^{\tree}$}\\
      (\blank, \blank, \blank, \divideSignal{0}{d})
          &\mapsto \setOf{\reflectedDivideSignal{\reverse d}},
          \intertext{(If a divide signal of type $0$ reaches a vertex, then reflect it.)}
      (\blank, \blank, \blank, \reflectedDivideSignal{d})
          &\mapsto \setOf{\fireSignal},
          \intertext{(If a reflected divide signal reaches a vertex, then create a fire signal.)}
      (E, \blank, \blank, \initiateSignal{d})
          &\mapsto \begin{aligned}[t]
                     &\parens[\big]{\bigcup_{e \in E \smallsetminus \setOf{\reverse d}} \setOf{\initiateSignal{e}}}\\
                     &\cup \parens[\big]{\bigcup_{e \in E \smallsetminus \setOf{\reverse d}} \setOf{\divideSignal{n}{e} \suchThat n \in \N_0}}\\
                     &\cup \parens[\big]{\bigcup_{e \in E} \setOf{\findMidpointSignal{\emptyWord}{e}, \slowedDownFindMidpointSignal{e}{\emptyWord}{\emptyWord}{\isInLeaf(E)}}},
                   \end{aligned}
          \intertext{(If an initiate signal reaches a vertex, then send initiate signals onto all incident edges except the one the original initiate signal comes from, send divide signals of all types onto all incident edges except the one the original initiate signal comes from, and send find-midpoint and slowed-down find-midpoint signals onto all incident edges to find all midpoints of paths that contain the vertex, where, in the case that the vertex is a leaf, the slowed-down find-midpoint signal is marked, where the mark means that it may be on a longest path and would be the first to find its midpoint.)}
      (E, \blank, b, \findMidpointSignal{w_o}{d})
          &\mapsto \begin{aligned}[t]
                     &\setOf{\reflectedFindMidpointSignal{w_o}{\reverse d}{\emptyWord}{b \land \isInLeaf(E)}}\\
                     &\cup \parens[\big]{\bigcup_{e \in E \smallsetminus \setOf{\reverse d}} \setOf{\findMidpointSignal{(\reverse d) \concat w_o}{e}}},
                   \end{aligned}
          \intertext{(If a find-midpoint signal reaches a leaf, then reflect it and, if it is one of the first signals to reach the leaf, then also mark it as a signal that may be on a longest path and would be the first to find its midpoint. And, if a find-midpoint signal reaches a vertex that is not a leaf, then reflect it and send find-midpoint signals onto all incident edges except the one the original signal comes from.)}
      (E, n, \blank, \reflectedFindMidpointSignal{\emptyWord}{d}{w_r}{b})
          &\mapsto \bigcup_{e \in E \smallsetminus \setOf{\reverse d}} \setOf{\slowedDownFindMidpointSignal{e}{\emptyWord}{(\reverse d) \concat w_r}{b \land \isPenultimate(E, n)}}, 
          \intertext{(If a reflected find-midpoint signal reaches the vertex it originated at, then send slowed-down find-midpoint signals onto all incident edges except the one the original signal comes from and, if the original signal is marked and is the penultimate marked signal that originated at and has returned to the vertex, then also mark the slowed-down signals as signals that may be on a longest path and would be the first to find their midpoints.)}
      (E, n, \blank, \reflectedFindMidpointSignal{e \concat w_o}{d}{w_r}{b})
          &\mapsto \setOf{\reflectedFindMidpointSignal{w_o}{e}{(\reverse d) \concat w_r}{b \land \isPenultimate(E, n)}},
          \intertext{(If a reflected find-midpoint signal reaches a vertex, then it takes the way back it took before it was reflected and, if it is not one of the penultimate marked signals that reaches the vertex, then it is unmarked.)}
      (E, n, \blank, \slowedDownFindMidpointSignal{d}{w_o}{w_r}{b})
          &\mapsto \bigcup_{e \in E \smallsetminus \setOf{\reverse d}} \setOf{\slowedDownFindMidpointSignal{e}{(\reverse d) \concat w_o}{w_r}{b \land \isPenultimate(E, n)}},
          \intertext{(If a slowed-down find-midpoint signal reaches a vertex, then send slowed-down find-midpoint signals onto all incident edges except the one the original signal comes from, and mark these signals, if the original signal is marked --- in which case it arrives at the same time and from the same direction as the slowest reflected find-midpoint signal that originated at the vertex, which however is not marked because it arrived too late at the leaf it was reflected at --- and from exactly one direction the slowest find-midpoint signal that originated at the vertex has not returned yet.)}
      (\blank, \blank, \blank, \freezeSignal{d})
          &\mapsto \emptyset,
          \intertext{(If a freeze signal reaches the end of the edge it freezes, then it vanishes.)}
      (\blank, \blank, \blank, \thawSignal{d}{\emptyWord}{\yes})
          &\mapsto \emptyset,
          \intertext{(If a thaw signal reaches the end of its path, then it vanishes.)}
      (\blank, \blank, \blank, \thawSignal{d}{e \concat w}{\no})
          &\mapsto \setOf{\thawSignal{e}{w}{\no}},
          \intertext{(If a thaw signal reaches a vertex, then it takes the direction that makes it stay on its path.)}
      (E, \blank, \blank, (s, d, y)) 
          &\mapsto \left\{
                     \begin{aligned}
                       &\setOf{(s, d, y)}, \text{ if $\speedOf(s) = 0$},\\
                       &\bigcup_{e \in E \smallsetminus \setOf{\reverse d}} \setOf{(s, e, y)}, \text{ otherwise}. 
                     \end{aligned}
                   \right.
          \shortintertext{(A stationary signal in a vertex stays there, and if a non-stationary signal reaches a vertex, then copies of it are sent onto each edge except the one the signal comes from.)}
    \end{align*}

    For trees and without freezing and thawing, the case that precisely two signals collide in a vertex is handled by the following map, which maps triples --- consisting of first, the set of directions that lead away from the vertex; secondly, the number of the directions from which the slowest reflected find-midpoint signals that originated at the vertex have already returned or have just arrived; and lastly, the set of colliding signals that is supposed to consist of precisely two signals --- to the resulting signals, if a collision rule is specified, and to $\bot$, otherwise:
      \begin{align*}
        \localTransitionFunction_{v, 2}^{\tree} \from \powerSetOf(\Directions) \times \N_0 \times \powerSetOf(\Signals) &\to \powerSetOf(\Signals) \cup \setOf{\bot}, \mathnote{map $\localTransitionFunction_{v, 2}^{\tree}$}\\
        (\blank, \blank, \setOf{\reflectedDivideSignal{d}, \boundarySignal})
            &\mapsto \setOf{\fireSignal},
            \intertext{(If a reflected divide signal collides with a boundary signal, then create a fire signal.)}
        (E, n, \setOf{s@\reflectedFindMidpointSignal{(\reverse d') \concat w_o'}{d}{w_r}{b}, s'@\slowedDownFindMidpointSignal{d'}{w_o'}{w_r'}{b'}})
            &\mapsto\\
                    &&\llap{$\begin{dcases*}
                       \left.\begin{aligned}
                         &\localTransitionFunction_{v, 1}^{\tree}(E, n, \no, s)\\
                         &\quad\cup \localTransitionFunction_{v, 1}^{\tree}(E, n, \no, s')\\
                         &\quad\cup \setOf{\midpointSignal{(\reverse d) \concat w_r}{(\reverse d') \concat w_o' \concat w_r'}},
                       \end{aligned}\right\}
                       &if $b = \no$ or $b' = \no$,\\ 
                       \setOf{\thawSignal{\reverse d}{w_r}{\no}, \thawSignal{\reverse d'}{w_o' \concat w_r'}{\no}}, &otherwise,
                     \end{dcases*}$}
            \intertext{(If a reflected find-midpoint signal collides with a slowed-down find-midpoint signal that originated at the same vertex, then the vertex of collision is the midpoint of the shortest path in the virtual tree from the vertex the reflected find-midpoint signal was reflected at to the vertex the signal that spawned the slowed-down find-midpoint signal was reflected at and, if this midpoint is not the midpoint of a longest path, then treat the original signals as if they reached the vertex alone and designate the point by a midpoint signal, and otherwise, send thaw signals along that longest path to both its ends.)}
        (E, n, \setOf{s@\reflectedFindMidpointSignal{\emptyWord}{d}{w_r}{b}, s'@\reflectedFindMidpointSignal{\emptyWord}{d'}{w_r'}{b'}})
            &\mapsto\\
                     &&\llap{$\begin{dcases*}
                       \left.\begin{aligned}
                         &\localTransitionFunction_{v, 1}^{\tree}(E, n, \no, s)\\
                         &\quad\cup \localTransitionFunction_{v, 1}^{\tree}(E, n, \no, s')\\
                         &\quad\cup \setOf{\midpointSignal{(\reverse d) \concat w_r}{(\reverse d') \concat w_r'}},
                       \end{aligned}\right\}
                       &if $b = \no$ or $b' = \no$ or $n \neq \cardinalityOf{E}$,\\ 
                       \setOf{\thawSignal{\reverse d}{w_r}{\no}, \thawSignal{\reverse d'}{w_r'}{\no}}, &otherwise,
                     \end{dcases*}$}
            \intertext{(If two reflected find-midpoint signals that originated at the same vertex collide with each other, then the vertex of collision is the vertex the signals originated at and it is the midpoint of the shortest path in the virtual tree between the vertices the signals were reflected at and, if this midpoint is not the midpoint of a longest path, then treat the original signals as if they reached the vertex alone and designate the point by a midpoint signal, and otherwise, send thaw signals along that longest path to both its ends.)}
        (\blank, \blank, \setOf{\thawSignal{d}{w}{\no}, \midpointSignal{w@(d' \concat w')}{d'' \concat w''}})
            &\mapsto \setOf{\thawSignal{d'}{w'}{\no}, \thawSignal{d''}{w''}{\no}}, 
            \intertext{(If a thaw signal collides with a midpoint signal that designates the midpoint of a path whose one half-path coincides with the path the thaw signal is going to take, then send an additional thaw signal along the other half-path.)}
        \blank &\mapsto \bot.
            \shortintertext{(If none of the above happened, then indicate that by returning bottom.) \qedhere}
      \end{align*}
  \end{definition}

  How events for trees are handled is given in

  \begin{definition}
  \label{definition:for-trees:local-transition-functions}
    The maps
    \begin{equation*}
      \left\{
        \begin{aligned}
          \xi \from \Signals &\to \Signals,\\
          \divideSignal{n}{d} &\mapsto \frozenDivideSignal{n}{d},\\
          \fireSignal &\mapsto \frozenFireSignal,\\
          s &\mapsto s,
        \end{aligned}
      \right\}
      \text{ and }
      \left\{
        \begin{aligned} 
          \chi \from \Signals &\to \Signals,\\
          \frozenDivideSignal{n}{d} &\mapsto \divideSignal{n}{d},\\
          \frozenFireSignal &\mapsto \fireSignal,\\
          s &\mapsto s,
        \end{aligned}
      \right\}
      \mathnote{maps $\xi$ and $\chi$}
    \end{equation*}
    freeze and thaw signals that can be frozen and thawed respectively.

    The map
    \begin{align*}
      \nu \from \powerSetOf(\Signals) \times \powerSetOf(\Signals) &\to \powerSetOf(\Signals), \mathnote{map $\nu$}\\
      (S, S') &\mapsto \begin{dcases*}
                         \xi(S'), &if $\Exists \freezeSignal{d} \in S \cup S'$\\
                                  &and $\notExists \thawSignal{d}{\emptyWord}{\yes} \in S \cup S'$,\\
                         \chi(S'), &if $\Exists \thawSignal{d}{\emptyWord}{\yes} \in S \cup S'$\\
                                   &and $\notExists \freezeSignal{d} \in S \cup S'$,\\
                         S', &otherwise,
                       \end{dcases*}
    \end{align*}
    takes a set of old signals and a set of new signals and freezes the new signals, if the old or new signals contain a freeze signal but not a thaw signal that thaws the synchronisation of an edge; thaws the new signals, if the old or new signals contain a thaw signal that thaws the synchronisation of an edge but not a freeze signal; and does nothing, otherwise. 

    The maps
    \begin{align*}
      \zeta_2 \from \powerSetOf(\Signals) &\to \powerSetOf(\powerSetOf(\Signals)), \mathnote{map $\zeta_2$}\\
      S &\mapsto \setOf{\setOf{s, s'} \subseteq S \suchThat s \neq s' \text{ and } \localTransitionFunction_{e, 2}^{\tree}(\setOf{s, s'}) \neq \bot},
    \end{align*}
    and
    \begin{align*}
      \eta_2 \from \powerSetOf(\Directions) \times \powerSetOf(\Signals) &\to \powerSetOf(\powerSetOf(\Signals)), \mathnote{map $\eta_2$}\\
      (D, S) &\mapsto \begin{aligned}[t]
                        \{
                          \setOf{s, s'} \subseteq S \suchThat{} &s \neq s' \text{ and }\\
                                                                &\localTransitionFunction_{v, 2}^{\tree}(D, x, \setOf{s, s'}) \neq \bot
                        \},
                      \end{aligned}
    \end{align*}
    both take a set of signals and return the set of unordered pairs of distinct signals from the given set for which a collision rule is specified in $\localTransitionFunction_{e, 2}^{\tree}$ and $\localTransitionFunction_{v, 2}^{\tree}$ respectively.

    %

    The map
    \begin{align*} 
      \localTransitionFunction_e^{\tree} \from \powerSetOf(\Signals) &\to \powerSetOf(\Signals), \mathnote{map $\localTransitionFunction_e^{\tree}$}\\
      S &\mapsto \nu(S,
                   \parens[\big]{ S \smallsetminus \bigcup_{P \in \zeta_2(S)} P } \cup
                   \parens[\big]{ \bigcup_{P \in \zeta_2(S)} \localTransitionFunction_{e, 2}^{\tree}(P) }
                 ),
    \end{align*}
    handles collisions of signals on edges by leaving signals for which no pairwise collision rule with any other signal is specified in $\localTransitionFunction_{e, 2}^{\tree}$ as is, by applying $\localTransitionFunction_{e, 2}^{\tree}$ to each unordered pair of distinct signals for which a collision rule is specified, and by applying $\nu$ to freeze or thaw signals if there are or were any freeze or thaw signals.

    The map
    \begin{align*} 
      \kappa \from \powerSetOf(\Signals) &\to \powerSetOf(\Directions), \mathnote{map $\kappa$}\\
      S &\mapsto \begin{aligned}[t]
                   &\parens{\bigcup_{\countSignal{D} \in S} D}\\
                   &\cup \setOf{d \in \Directions \suchThat \Exists w_r \in \Directions^* \SuchThat \reflectedFindMidpointSignal{\emptyWord}{d}{w_r}{\yes} \in S}\\
                   &\cup \setOf{d \in \Directions \suchThat \Exists w_o \in \Directions^* \Exists w_r \in \Directions^* \SuchThat \slowedDownFindMidpointSignal{d}{w_o}{w_r}{\yes} \in S}\\
                 \end{aligned}
    \end{align*} 
    takes a set of signals that are at a vertex and returns the directions from which the slowest reflected find-midpoint signals that originated at the vertex have already returned or have just arrived. The directions from which signals have already returned is memorised by a count signal, and the other directions are the ones from which marked reflected find-midpoint signals that originated at the vertex or marked slowed-down find-midpoint signals with any origin have just arrived (the slowest ones of the latter kind always arrive at the same time and from the same direction as the slowest but unmarked reflected find-midpoint signal that originated at the vertex arrive from that direction). Note that although we take the union of the memories of all count signals, there is actually gonna be no such signal in leaves (in which case the union is $\emptyset$) and precisely one such signal in each vertex that is not a leaf (in which case the union is the memory stored by this signal). 

    The map
    \begin{align*} 
      \varkappa \from \powerSetOf(\Signals) &\to \booleans, \mathnote{map $\varkappa$}\\
      S &\mapsto \begin{dcases*}
                   \yes, &if $\Exists d \in \Directions \SuchThat \initiateSignal{d} \in S$,\\
                   \no, &otherwise,
                 \end{dcases*}
    \end{align*}
    tells whether a set of signals contains any initiate signals or not. It is used by $\localTransitionFunction_v^{\tree}$ to tell whether a find-midpoint signal that reaches a leaf is among the first ones to do so, which is the case if and only if the signal travels alongside an initiate signal.

    The map
    \begin{align*}
      \localTransitionFunction_v^{\tree} \from{} &\powerSetOf(\Directions) \times \powerSetOf(\Signals) \to \powerSetOf(\Signals), \mathnote{map $\localTransitionFunction_v^{\tree}$}\\
      &(D, S) \begin{aligned}[t]
                {}\mapsto \nu(S,
                  &\left\{\begin{aligned} 
                    &\emptyset, &&\text{if $\isInLeaf(D) = \yes$},\\ 
                    &\setOf{\countSignal{\kappa(S)}}, &&\text{otherwise},
                  \end{aligned}\right\}\\
                  &\cup \localTransitionFunction_{v, 1}^{\tree}\parens[\big]{D, \cardinalityOf{\kappa(S)}, \varkappa(S), S' \smallsetminus \bigcup_{P \in \eta_2(D, S')} P}\\
                  &\cup \bigcup_{P \in \eta_2(D, S')} \localTransitionFunction_{v, 2}^{\tree}(D, \cardinalityOf{\kappa(S)}, P)
                ),\\
                &\text{ where } S' = S \smallsetminus \setOf{\countSignal{D} \in S \suchThat D \subseteq \Directions}.
              \end{aligned}
    \end{align*}
    handles events in vertices by updating the memory of the directions from which the slowest reflected find-midpoint signals that originated at the vertex have already returned or have just arrived, by applying $\localTransitionFunction_{v, 1}^{\tree}$ to non-count signals for which no pairwise collision rule with any other non-count signal is specified in $\localTransitionFunction_{v, 2}^{\tree}$, by applying $\localTransitionFunction_{v, 2}^{\tree}$ to each unordered pair of distinct non-count signals for which a collision rule is specified, and by applying $\nu$ to freeze or thaw signals if there are or were any freeze or thaw signals.
  \end{definition}

  How events for virtual trees, where virtual leaves already exist, are handled is given in

  \begin{definition}
  \label{definition:for-virtual-trees:local-transition-functions}
    When colliding signals in the graph and the involved directions are partitioned with respect to the virtual tree, some of the components may be degenerated and applying $\localTransitionFunction_e^{\tree}$ and $\localTransitionFunction_v^{\tree}$ to them may have unwanted effects. Such boundary cases are properly handled by the maps 
    \begin{align*}
      \localTransitionFunction_e^{\boundaryCases} \from \powerSetOf(\Signals) &\to \powerSetOf(\Signals), \mathnote{map $\localTransitionFunction_e^{\boundaryCases}$}\\
      S &\mapsto \begin{dcases*}
                   \emptyset, &if $\cardinalityOf{S} \leq 1$,\\
                   \localTransitionFunction_e^{\tree}(S), &otherwise,
                 \end{dcases*}
    \end{align*}
    and
    \begin{align*}
      \localTransitionFunction_v^{\boundaryCases} \from \powerSetOf(\Directions) \times \powerSetOf(\Signals) &\to \powerSetOf(\Signals), \mathnote{map $\localTransitionFunction_v^{\boundaryCases}$}\\
      (D, S) &\mapsto \begin{dcases*}
                        \emptyset, &if $D = \emptyset$ or $S = \emptyset$,\\
                        \localTransitionFunction_v^{\tree}(D, S), &otherwise.
                      \end{dcases*}
    \end{align*}

    To handle an event involving the signals $S$ in a virtual tree we do the following: We partition the signals into leaf signals, namely $S_{\leafKind}$, for each leaf signal $\leafSignal{d} \in S_{\leafKind}$, the signals coming from the direction $\reverse d$, namely $S_d$, and all other signals, namely $S_o$. And we denote the set of directions that no leaf signal has by $D'$. Intuitively, $S_{\leafKind}$ is the set of virtual leaves, $S_d$ is the set of signals that reach the virtual leaf $\leafSignal{d}$, and $D'$ is the set of directions that do not lead away from any virtual leaf. In the configurations we will encounter, in the case of a collision on an edge, there are either no leaf signals, in which case $S_o = S$, or there are two leaf signals, in which case $S_o = \emptyset$ (because there are no stationary signals in virtual leaves besides the leaf signals themselves) and hence either the signals simply collide, or one subset of signals reaches one virtual leaf and the other subset reaches the other virtual leaf. And, in the case of a collision in a vertex, signals coming from some edges may reach a virtual leaf and signals coming from other edges may reach the virtual vertex (we call it virtual because some of its original edges have been cut off; the directions onto the ones that have not been cut off are those in the set $D'$). Note that some collisions are not collisions in the virtual tree, because the signals came from different directions of virtual cuts. The maps that do what we just explained are 
    \begin{align*}
      \localTransitionFunction_e^{\virtualTree} \from \powerSetOf(\Signals) &\to \powerSetOf(\Signals), \mathnote{map $\localTransitionFunction_e^{\virtualTree}$}\\
      S &\mapsto \begin{aligned}[t]
                   S_{\leafKind}
                   &\cup \localTransitionFunction_e^{\boundaryCases}(S_o)\\
                   &\cup \parens[\big]{\bigcup_{d \in X} \localTransitionFunction_v^{\boundaryCases}(\setOf{d}, S_d)},
                 \end{aligned}
    \end{align*}
    and
    \begin{align*}
      \localTransitionFunction_v^{\virtualTree} \from \powerSetOf(\Directions) \times \powerSetOf(\Signals) &\to \powerSetOf(\Signals), \mathnote{map $\localTransitionFunction_v^{\virtualTree}$}\\
      (D, S) &\mapsto \begin{aligned}[t]
                        S_{\leafKind}
                        &\cup \localTransitionFunction_v^{\boundaryCases}(D', S_o)\\
                        &\cup \parens[\big]{\bigcup_{d \in X} \localTransitionFunction_v^{\boundaryCases}(\setOf{d}, S_d)},
                      \end{aligned}
    \end{align*}
    where $X = \setOf{d \in \Directions \suchThat \leafSignal{d} \in S}$, the set of directions that lead away from virtual leaves, $S_{\leafKind} = \setOf{\leafSignal{d} \suchThat d \in X}$, the set of virtual leaves, $\family{S_d}_{d \in X} = \family{\setOf{s \in S \suchThat \directionOf(s) = \reverse d}}_{d \in X}$, for each direction of a virtual leaf, the set of signals that reach the virtual leaf corresponding to the direction moving towards it, $S_o = S \smallsetminus (S_{\leafKind} \cup (\bigcup_{d \in X} S_d))$, the signals that are not virtual leaves and that do not reach a virtual leaf, and $D' = D \smallsetminus X$, the set of directions that do not lead away from virtual leaves. 
  \end{definition}

  How events for graphs are handled is given in

  \begin{definition}
  \label{definition:for-graphs:local-transition-functions}
    The map
    \begin{align*}
      \mu \from \powerSetOf(\Signals) &\to \powerSetOf(\Signals), \mathnote{map $\mu$}\\
      S &\mapsto \reverse \setOf{d \in \Directions \suchThat \initiateSignal{d} \in S},
    \end{align*}
    takes a set of signals and returns the set of the reverses of the directions that initiate signals have.

    The map
    \begin{align*}
      \varphi_{\leafKind}^e \from \powerSetOf(\Signals) &\to \powerSetOf(\Signals), \mathnote{map $\varphi_{\leafKind}^e$}\\
      S &\mapsto \begin{dcases*}
                   \emptyset, &if $\cardinalityOf{\mu(S)} \leq 1$,\\
                   \setOf{\leafSignal{d} \suchThat d \in \mu(S)}, &otherwise,
                 \end{dcases*}
    \end{align*}
    takes a set of signals and returns the empty set, if there is at most one initiate signal, and the set that consists of a virtual leaf for each initiate signal, otherwise.

    The map
    \begin{align*}
      &\varphi_{\leafKind}^v \from \powerSetOf(\Directions) \times \powerSetOf(\Signals) \to \powerSetOf(\Signals), \mathnote{map $\varphi_{\leafKind}^v$}\\
      &(D, S) \mapsto \left\{
                        \begin{aligned} 
                          &\emptyset, \text{ if $\cardinalityOf{\mu(S)} \leq 1$},\\
                          &\setOf{\leafSignal{d} \suchThat d \in D}, \text{ if $\cardinalityOf{\mu(S)} \geq 2$ and $\mu(S) = D$,}\\
                              &&\llap{and $\isInLeaf(D) = \no$},\\ 
                          &\setOf{\leafSignal{d} \suchThat d \in D \smallsetminus \setOf{d_x}}, \text{ if $\cardinalityOf{\mu(S)} \geq 2$ and $\mu(S) \neq D$},\\
                              &&\llap{for some $d_x \in D \smallsetminus \mu(S)$}, 
                        \end{aligned}
                      \right.
    \end{align*}
    takes a set of directions and a set of signals and returns the empty set, if there is at most one initiate signal, or the set that consists of a virtual leaf for each initiate signal, if there are at least two initiate signals and initiate signals reached the vertex from all incident edges, or the set that consists of a virtual leaf for each initiate signal but one, otherwise. Note that right now the choice of $d_x$ is non-deterministic; however, if the finite set of directions carried a total order, then we could deterministically choose for example the smallest direction; or if continuum representations of graphs were embedded in high-dimensional Euclidean spaces and a Cartesian coordinate system was chosen such that the occurring directions are unit vectors, then the lexicographic order is a total order on the set of directions.

    The map
    \begin{align*}
      \localTransitionFunction_e \from \domainOf(\localTransitionFunction_e) &\to \powerSetOf(\Signals), \mathnote{map $\localTransitionFunction_e$}\\
      S &\mapsto \localTransitionFunction_e^{\virtualTree}(S \cup \varphi_{\leafKind}^e(S))
    \end{align*}
    handles collisions on edges by creating two virtual leaves, if two initiate signals collide, which cuts the edge virtually, and then applying $\localTransitionFunction_e^{\virtualTree}$ to the maybe new set of signals.

    The map
    \begin{align*}
      \localTransitionFunction_v \from \domainOf(\localTransitionFunction_v) &\to \powerSetOf(\Signals), \mathnote{map $\localTransitionFunction_v$}\\
      (D, S) &\mapsto \localTransitionFunction_v^{\virtualTree}(D, S \cup \varphi_{\leafKind}^v(D, S))
    \end{align*}
    handles events in vertices by creating a virtual leaf for each incident edge, if the vertex is not a leaf and initiate signals reached the vertex from all incident edges, or by creating a virtual leaf for each incident edge from which an initiate signal arrived except for one such edge, otherwise, and then applying $\localTransitionFunction_v^{\virtualTree}$ to the maybe new set of signals.

    Intuitively, initiate signals are used to turn the graph into a virtual tree by cutting edges at points where such signals collide. These cuts create virtual leaves which are represented by leaf signals. More precisely: When two initiate signals collide on an edge, it is cut by two leaf signals, one for each of the two directions. And when at least two initiate signals collide in a vertex, there are two cases: If initiate signals arrive from all directions, then each incident edge is cut by a leaf signal; otherwise, the incident edges from which initiate signals arrive are cut except for one such edge --- the initiate signal from this excluded edge will spread to all edges that have not been cut. 
  \end{definition}

  \begin{main-theorem} 
  \label{theorem:signal-machine-is-quasi-solution-of-the-firing-mob-synchronisation-problem}
    The signal machine \graffito{signal machine $\mathcal{S}$}$\mathcal{S} = \ntuple{\Kinds, \speedOf, \family{\Data_k}_{k \in \Kinds}, (\localTransitionFunction_e, \localTransitionFunction_v)}$ is a time-optimal quasi-solution of the firing mob synchronisation problem over continuum representations of weighted, non-trivial, finite, and connected undirected multigraphs in the following sense: For each representation $M$ of such a graph, each vertex $\general \in M$, for the time $t = r + d$, where $r = \sup_{m \in M} \distanceOf(\general, m)$ is the radius of $M$ with respect to $\general$ and $d = \sup_{m, m' \in M} \distanceOf(m, m')$ is the diameter of $M$, for the instantiation of $\mathcal{S}$ for $M$, for the configuration $c \in \Configurations$ such that $c(\general) = \bigcup_{d \in \directionOf(\general)} \setOf{\initiateSignal{d}} \cup \setOf{\divideSignal{n}{d} \suchThat n \in \N_0} \cup \setOf{\findMidpointSignal{\emptyWord}{d}, \slowedDownFindMidpointSignal{d}{\emptyWord}{\emptyWord}{\isInLeaf(\directionOf(\general))}}$ and $c\restrictedTo_{M \smallsetminus \setOf{\general}} \equiv \emptyset$, the points in the configuration $\globalTransitionFunction(t)(c)$ at which a fire signal occurs lies dense in $M$ with respect to the metric $\distanceOf$, and no fire signals occur in any of the configurations $\globalTransitionFunction(s)(c)$, for $s \in \R_{\geq 0}$ with $s < t$.
  \end{main-theorem}

  \begin{proof-sketch}
    A proof is sketched in \cref{section:proof-sketch}.
  \end{proof-sketch}

  \begin{remark}
  \label{remark:data-sets-of-kinds-can-be-chosen-to-be-finite}
    For each positive integer $k$, under the restriction to multigraphs whose maximum degree is bounded by $k$, for each such multigraph, because directions only need to be locally unique (compare \cref{remark:efficient-representation-of-directions}), the set $\setOf{1, 2, \dotsc, 2 k}$ can be chosen as the set of directions, which makes the data sets of the kinds $\leafKind$ and $\frozenDivideKind$ finite and independent of the multigraph, the finite set $\powerSetOf(\setOf{1, 2, \dotsc, k})$ can be chosen as the data set of the kind $\countKind$, and, depending on the diameter $d$ of the multigraph, the finite set of words over $\Directions$ with maximum length $d$ can be chosen as the sets of words over $\Directions$ that occur in the data sets of the kinds $\midpointKind$, $\findMidpointKind$, $\reflectedFindMidpointKind$, $\slowedDownFindMidpointKind$, and $\thawKind$ --- altogether, the data sets of all kinds can be chosen to be finite but some depend on the multigraph.
  \end{remark}

  \begin{corollary}
    A discretisation of the signal machine $\mathcal{S}$ is a time-optimal cellular automaton quasi-solution of the firing mob synchronisation problem over non-trivial, finite, and connected undirected multigraphs.
  \end{corollary}

  \begin{proof-sketch}
    Let $\Graph$ be a non-trivial, finite, and connected undirected multigraph. It is made up of paths whose source and target vertices are not of degree $2$ and whose other vertices are of degree $2$. Each such path together with its inverse can be regarded as an undirected \emph{uber-edge} whose weight is the length of the path and whose ends are the source and target vertices of the path or its inverse. We call vertices that are not of degree $2$ \emph{uber-vertices} and vertices that are of degree $2$ \emph{under-vertices}.

    Signals jump from vertices to vertices along edges. They \emph{collide} when they jump simultaneously from different vertices onto the same vertex or along the same edge but in different directions, or when signals jump onto vertices on which stationary signals reside. Collisions in uber-vertices are handled as collisions in vertices, collisions in under-vertices are handled as collisions on edges (namely on the uber-edges that contain the under-vertices), and collisions in the midst of edges are handled as collisions on edges (namely on the uber-edges that contain the edges). Signals \emph{reach a vertex} when they jump onto an uber-vertex, but not when they jump onto an under-vertex (because the latter just means that they travel along an uber-edge).

    When signals collide in a vertex, be it a uber- or under-vertex, the resulting signals are on the vertex. But when signals collide in the midst of an edge, the resulting signals must be distributed onto both its ends depending on their directions. This last case makes it rather cumbersome to write down the local transition functions explicitly. Vertices must be virtually divided into multiple parts: One part that plays the role of the vertex itself and, for each incident edge, an additional part that plays the role of the midpoint of the edge together with the corresponding part of the other end of the edge. And signals must be cleverly distributed onto these parts depending on their direction and how they came into being, and collisions of signals must also be cleverly handled taking the parts the involved signals came from and are on into account.

    This discretisation of the signal machine is actually a cellular automaton over the multigraph with appropriate dummy neighbours that are in a dead state (think for example of the multigraph as being embedded in a coloured $S$-Cayley graph with sufficient maximum degree and of the vertices that do not belong the multigraph as being in a dead state).
  \end{proof-sketch}

  \begin{remark}
    Jacques Mazoyer showed in 1987 that all infinitely many divide signals of type $n$, for $n \in \N_0$, that emanate from the same point can be generated by a cellular automaton with only finitely many states (see \cite{mazoyer:1987}). And, as illustrated in \cref{remark:data-sets-of-kinds-can-be-chosen-to-be-finite}, under the restriction to multigraphs whose maximum degrees are uniformly bounded by a constant, the data sets of all kinds can be chosen to be finite but some depend on the multigraph. Therefore, depending on the multigraph, the discretisation of $\mathcal{S}$ is a cellular automaton with a finite number of states.
  \end{remark}

  \begin{open-problem}
    Are there time-optimal signal machine and cellular automaton solutions of the firing mob synchronisation problem over non-trivial, finite, and connected undirected multigraphs whose maximum degrees are uniformly bounded by a constant? Or, more specifically, is it possible to adapt the signal machine $\mathcal{S}$ (and thereby its discretisation) such that the data sets of all kinds can be chosen to be finite and independent of the multigraph (and thereby making its discretisation have a finite set of states), for example by reducing the number of midpoints that are and need to be determined?
  \end{open-problem}

  \section{Proof Sketch of the Main Theorem}
  \label{section:proof-sketch}

  In this section, we sketch a proof of \cref{theorem:signal-machine-is-quasi-solution-of-the-firing-mob-synchronisation-problem}. To that end, let $\Graph = \ntuple{\Vertices, \Edges, \eendsOf}$ be a non-trivial, finite, and connected undirected multigraph, let $\weightOf$ be an edge weighting of $\Graph$, let $M$ be a continuum representation of $\Graph$, identify vertices of $\Graph$ and $M$ and direction-preserving paths from vertices to vertices of $\Graph$ and $M$, and let $\general$ be a vertex of $M$, which we call \define{general}\graffito{general vertex $\general$}. Furthermore, let $\mathcal{S}$ be the signal machine and let $c$ be the initial configuration of the firing mob synchronisation problem from \cref{theorem:signal-machine-is-quasi-solution-of-the-firing-mob-synchronisation-problem}, and, whenever we talk about time evolution, we mean the one of $\mathcal{S}$ that is in the configuration $c$ at time $0$, for example, \emph{at time $t$} either means \emph{in configuration $\globalTransitionFunction(t)(c)$} or \emph{essentially in configuration $\globalTransitionFunction(t)(c)$ but before events have been handled}.

  To proof \cref{theorem:signal-machine-is-quasi-solution-of-the-firing-mob-synchronisation-problem}, we need to ascertain that the signal machine performs the following tasks: First, it cuts the multigraph such that the multigraph turns into a virtual tree and looks like a tree to all other tasks; secondly, it starts synchronisation of edges and freezes it in time; thirdly, it determines the midpoints of all non-empty direction-preserving paths from vertices to vertices in time; fourthly, it determines which midpoints are the ones of the longest paths; fifthly, starting from the midpoints of the longest paths, it traverses midpoints of shorter and shorter paths and upon reaching midpoints of edges, it thaws synchronisation of the respective edges; sixthly, all edges finish synchronisation at time $r + d$ with the creation of fire signals that lie dense in the graph, where $r$ is the radius of the graph with respect to the general and $d$ is the diameter of the graph.

  That the first task is performed is evident from the definitions of $\localTransitionFunction_e$, $\localTransitionFunction_v$, $\localTransitionFunction_e^{\virtualTree}$, and $\localTransitionFunction_v^{\virtualTree}$. The only subtlety here is that besides leaf signals there cannot be stationary signals in virtual leaves or, more precisely, at points that are virtually cut by leaf signals, because stationary signals carry the semi-direction $\every$, which is insufficient to associate them with one or the other leaf signal as is done for non-stationary signals. This is no problem because the other tasks do not place stationary signals in (virtual) leaves. Therefore, we assume from now on, without loss of generality, that the multigraph $\Graph$ is a \emph{tree}.

  That the second task is performed can be seen from a careful examination of the definitions of $\localTransitionFunction_e^{\tree}$, $\localTransitionFunction_v^{\tree}$, $\localTransitionFunction_{e, 2}^{\tree}$, $\localTransitionFunction_{v, 1}^{\tree}$, and $\localTransitionFunction_{v, 2}^{\tree}$, where from the third task it is used that midpoints of edges are found in time to start the freezing process. 

  That the third, fourth, and fifth and sixth parts are performed is proven in \cref{subsection:midpoints-are-determined}, \cref{subsection:midpoints-of-maximum-weight-paths-are-recognised}, and \cref{subsection:thaw-signals-traverse-midpoints-and-thaw-synchronisation-of-edges-just-in-time}. 

  \subsection{Midpoints are Determined}
  \label{subsection:midpoints-are-determined}

  The midpoint of a path in a multigraph is the midpoint of its embedding in the continuum representation of the multigraph as introduced in

  \begin{definition} 
    Let $p$ be a path in $\Graph$. The point $\midpoint_p = \continuumRepresentationOf{p}(\weightOf(p) / 2)$ is called \define{midpoint of $p$}\graffito{midpoint $\midpoint_p$ of $p$}\index[symbols]{mpfraktur@$\midpoint_p$}.
  \end{definition}

  \begin{remark}
    The midpoint of the empty path in $v$ is the vertex $v$ itself.
  \end{remark}

  \begin{remark}
    Let $p$, $q$, and $q'$ be three paths such that $\targetOf(p) = \sourceOf(q) = \sourceOf(q')$, $\midpoint_{p \concat q} \in \imageOf p$, and $\weightOf(q) \geq \weightOf(q')$. Then, $\midpoint_{p \concat q'} \in \imageOf p$ and
    \begin{equation*}
      \distanceOf(\midpoint_{p \concat q}, \midpoint_{p \concat q'})
      = \weightOf(p \concat q) / 2 - \weightOf(p \concat q') / 2
      = \weightOf(q) / 2 - \weightOf(q') / 2.
    \end{equation*}

    Analogously, let $q$, $q'$, and $p$ be three paths such that $\targetOf(q) = \targetOf(q') = \sourceOf(p)$, $\midpoint_{q \concat p} \in \imageOf p$, and $\weightOf(q) \geq \weightOf(q')$. Then, $\midpoint_{q' \concat p} \in \imageOf p$ and
    \begin{equation*}
      \distanceOf(\midpoint_{q \concat p}, \midpoint_{q' \concat p})
      = \weightOf(q \concat p) / 2 - \weightOf(q' \concat p) / 2
      = \weightOf(q) / 2 - \weightOf(q') / 2. \qedhere
    \end{equation*} 
  \end{remark}

  When the midpoints of non-empty direction-preserving paths are found is stated in

  \begin{lemma}
  \label{lemma:the-time-at-which-the-midpoint-of-a-path-is-found}
    Let $p$ be a non-empty direction-preserving path in $\Graph$. The midpoint signal that designates the midpoint of $p$ is created at $\midpoint_p$ at time $t_p = \max\setOf{\distanceOf(\general, \sourceOf(p)),\allowbreak \distanceOf(\general, \targetOf(p))} + \weightOf(p)/2$.
  \end{lemma}

  \begin{proof-sketch}
    First, let the general $\general$ be the source or target of $p$. Then, a find-midpoint signal with origin $\general$ of speed $1$ and a slowed-down find-midpoint signal with origin $\general$ of speed $1/3$ travel from $\general$ towards the other end of $p$. The find-midpoint signal is reflected at the other end at time $(2/2) \cdot \weightOf(p)$ and this reflection collides with the slowed-down find-midpoint signal at the midpoint of $p$ at time $(3/2) \cdot \weightOf(p)$ creating a midpoint signal for $p$ (because both signals have the same origin). Note that the time of collision is equal to $t_p$ (see \cref{figure:determineMidpointOfOneEdge}).

    Secondly, let the general $\general$ lie on $p$ without being its source or target. Then, two find-midpoint signals with origin $\general$ travel from $\general$ to the ends of $p$; the source of $p$ is reached at time $\distanceOf(\general, \sourceOf(p))$ and the target at time $\distanceOf(\general, \targetOf(p))$. When such a signal reaches its end, it is reflected and travels back. If $\distanceOf(\general, \sourceOf(p)) = \distanceOf(\general, \targetOf(p))$, then this distance is equal to $\weightOf(p)/2$, the reflected signals collide at time $\weightOf(p)$ at the midpoint of $p$ creating a midpoint signal for $p$; note that the time of collision is equal to $t_p$. Otherwise, the reflected signal that is nearer to $\general$ reaches this vertex first, where the signal is slowed down and travels towards the other reflected signal with which it collides at the midpoint of $p$ at time $t_p$ (see \cref{figure:determineMidpointOfTwoEdges}).

    Lastly, let the general $\general$ not lie on $p$. Then, an initiate signal travels from $\general$ to the nearest vertex $v$ on $p$, where it creates two find-midpoint signals with origin $v$ that travel to the ends of $p$, are reflected at these ends and travel back, one is slowed-down upon reaching $v$, and the slowed-down signal collides with the reflected signal in the midpoint of $p$ at time $\distanceOf(\general, v) + \max\setOf{\distanceOf(v, \sourceOf(p)),\allowbreak \distanceOf(v, \targetOf(p))} + \weightOf(p)/2$ creating a midpoint signal for $p$. Note that the time of collision is equal to $t_p$.
  \end{proof-sketch}

  \begin{remark} 
    When reflected/slowed-down find-midpoint signals with different origins collide, nothing happens, the signals just move on. Hence, no points are falsely found to be midpoints.
  \end{remark}

  That the midpoints of maximum-weight direction-preserving paths are identical and found at time $r + d$ is shown in

  \begin{lemma}
  \label{lemma:the-time-at-which-the-midpoints-of-longest-paths-are-found}
    All maximum-weight direction-preserving paths in $\Graph$ have the same midpoint $\hat{\midpoint}$ and the midpoint signals that designate the midpoints of such paths are created at $\hat{\midpoint}$ at time $r + d/2$, where $r = \max_{v \in \Vertices} \distanceOf(\general, v)$ is the radius of $\Graph$ with respect to $\general$ and $d = \max_{v, v' \in \Vertices} \distanceOf(v, v')$ is the diameter of $\Graph$. Note that $r$ is equal to the radius $\sup_{m \in M} \distanceOf(\general, m)$ of $M$ with respect to $\general$ and $d$ is equal to the diameter $\sup_{m, m' \in M} \distanceOf(m, m')$ of $M$.
  \end{lemma}

  \begin{proof} 
    We prove both statements by contradiction.

    First, suppose that there are two maximum-weight direction-pre\-serv\-ing paths $\hat{p}$ and $\hat{p}'$ in $\Graph$ that do not have the same midpoint. Then, there is a non-empty direction-preserving path $\mathfrak{p}_{\midpoint}$ in $M$ from the midpoint of $\hat{p}$ to the one of $\hat{p}'$. And, there is a direction-preserving subpath $\mathfrak{p}$ of $\continuumRepresentationOf{\hat{p}}$ in $M$ from one end of $\hat{p}$ to its midpoint whose target-direction is not the reverse of the source-direction of $p_{\midpoint}$. And, there is a direction-preserving subpath $\mathfrak{p'}$ of $\continuumRepresentationOf{\hat{p}'}$ in $M$ from the midpoint of $\hat{p}'$ to one of its ends whose source-direction is not the reverse of the target-direction of $p_{\midpoint}$. The concatenation of $\mathfrak{p}$, $\mathfrak{p}_{\midpoint}$, and $\mathfrak{p}'$ is a direction-preserving path from vertex to vertex in $M$ whose length is equal to $d/2 + \length(\mathfrak{p}_{\midpoint}) + d/2 > d$. It corresponds to a direction-preserving path $p$ in $\Graph$ whose weight is greater than $d$, which contradicts that $d$ is the diameter of $\Graph$. Therefore, all maximum-weight direction-preserving paths in $\Graph$ have the same midpoint, which we denote by $\hat{\midpoint}$.

    Secondly, suppose that there is a maximum-weight direction-pre\-serv\-ing path $\hat{p}$ in $\Graph$ such that $\max\setOf{\distanceOf(\general, \sourceOf(\hat{p})),\allowbreak \distanceOf(\general, \targetOf(\hat{p}))} < r$. Let $v$ be a vertex of $\Graph$ such that $\distanceOf(\general, v) = r$, let $p_v$, $p_{\sourceOf}$, and $p_{\targetOf}$ be the direction-preserving paths in $\Graph$ from $v$, $\sourceOf(\hat{p})$, and $\targetOf(\hat{p})$ to $\general$, let $v'$ be the vertex on $p_v$ and $p_{\sourceOf}$ or on $p_v$ and $p_{\targetOf}$ that is the furthest from $\general$, and let $v''$ be the vertex on $\hat{p}$ that is the nearest to $\general$. Then, the weight of $p_v$ is $r$, the one of $p_{\sourceOf}$ is $\distanceOf(\general, \sourceOf(\hat{p})) < r$, and the one of $p_{\targetOf}$ is $\distanceOf(\general, \targetOf(\hat{p})) < r$. If $v'$ lies on the subpath of $p_{\sourceOf}$ from $\sourceOf(\hat{p})$ to $v''$ (which is equal to the subpath of $\hat{p}$ with the same ends), then let $p$ be the direction-preserving path from $v$ over $v'$ over $v''$ to $\targetOf(\hat{p})$; if $v'$ lies on the subpath of $p_{\targetOf}$ from $\targetOf(\hat{p})$ to $v''$ (which is equal to the subpath of the inverse of $\hat{p}$ with the same ends), then let $p$ be the direction-preserving path from $v$ over $v'$ over $v''$ to $\sourceOf(\hat{p})$; and otherwise, let $p$ be the direction-preserving path from $v$ over $v'$ over $v''$ to $\targetOf(\hat{p})$ (we could have chosen $\sourceOf(\hat{p})$ as well). See \cref{figure:the-time-at-which-the-midpoints-of-longest-paths-are-found} for a schematic representation of the three cases.
    \begin{figure}
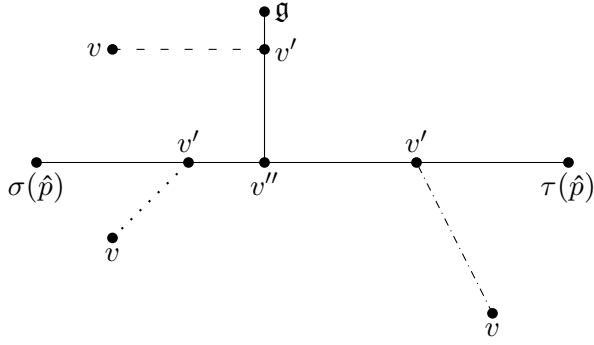

      \myfloatalign
      \figureTheTimeAtWhichTheMidpointsOfLongestPathsAreFound
      \caption{Schematic representation of the set-up of the proof of \cref{lemma:the-time-at-which-the-midpoints-of-longest-paths-are-found} with the three cases of where $v'$ may be located. The first case corresponds to the vertices $v$ and $v'$ that are incident to the dotted path, the second case to the dash-dotted path, and the third case to the dashed path.}
      \label{figure:the-time-at-which-the-midpoints-of-longest-paths-are-found}
    \end{figure}

    In the first case, because the subpaths of $p_v$ and $p_{\sourceOf}$ from $v''$ to $\general$ coincide and the weight of $p_v$ is greater than the weight of $p_{\sourceOf}$, the weight of the subpath of $p_v$ from $v$ over $v'$ to $v''$ (which is equal to the subpath of $p$ from $v$ over $v'$ to $v''$) is greater than the weight of the subpath of $p_{\sourceOf}$ from $\sourceOf(\hat{p})$ over $v'$ to $v''$ (which is equal to the subpath of $\hat{p}$ from $\sourceOf(\hat{p})$ over $v'$ to $v''$) and hence, because the subpaths of $p$ and $\hat{p}$ from $v''$ to $\targetOf(\hat{p})$ coincide, the weight of $p$ is greater than the weight of $\hat{p}$. In the second case, it follows analogously that the weight of $p$ is greater than the weight of $\hat{p}$. And in the third case, because the subpaths of $p_v$ and $p_{\sourceOf}$ from $v'$ to $\general$ coincide and the weight of $p_v$ is greater than the weight of $p_{\sourceOf}$, the weight of the subpath of $p_v$ from $v$ to $v'$ is greater than the weight of the subpath of $p_{\sourceOf}$ from $\sourceOf(\hat{p})$ over $v''$ to $v'$, hence the weight of the subpath of $p$ from $v$ over $v'$ to $v''$ is greater than the weight of the subpath of $\hat{p}$ from $\sourceOf(\hat{p})$ to $v''$, and therefore, the weight of $p$ is greater than the weight of $\hat{p}$.

    In either case, the inequality $\weightOf(p) > \weightOf(\hat{p})$ contradicts that $\hat{p}$ is a maximum-weight path. Therefore, for each maximum-weight direction-preserving path in $\Graph$, we have $\max\setOf{\distanceOf(\general, \sourceOf(\hat{p})),\allowbreak \distanceOf(\general, \targetOf(\hat{p}))} = r$. It follows from \cref{lemma:the-time-at-which-the-midpoint-of-a-path-is-found} that the midpoint signals that designate the midpoints of maximum-weight direction-preserving paths in $\Graph$ are created at $\hat{\midpoint}$ at time $r + d/2$.
  \end{proof}

  It follows that the midpoints of non-empty non-maximum-weight direction-preserving paths are found before the ones of maximum-weight direction-preserving paths as shown in

  \begin{corollary}
    Let $p$ be a non-empty direction-preserving path in $\Graph$. The midpoint signal that designates the midpoint of $p$ is created at $\midpoint_p$ before time $r + d/2$.
  \end{corollary}

  \begin{proof}
    This is a direct consequence of \cref{lemma:the-time-at-which-the-midpoint-of-a-path-is-found,lemma:the-time-at-which-the-midpoints-of-longest-paths-are-found}, because $\distanceOf(\general, \sourceOf(p)) \leq r$, $\distanceOf(\general, \targetOf(p)) \leq r$, and $\weightOf(p) < d$.
  \end{proof}

  \subsection{Midpoints of Maximum-Weight Paths are Recognised as Such}
  \label{subsection:midpoints-of-maximum-weight-paths-are-recognised}

  That the midpoints of maximum-weight direction-preserving paths are recognised as such is sketched in

  \begin{remark}
  \label{remark:midpoints-of-maximum-weight-paths-are-recognised}
    The first two reflected find-midpoint signals, or the first two reflected and slowed-down find-midpoint signals to collide that originated at the same vertex and were reflected at the ends of a maximum-weight direction-preserving path are marked at the time of collision and hence recognise that the midpoint of the path they collide at is the one of a maximum-weight direction-preserving path.
  \end{remark}

  \begin{proof-sketch} 
    Let $\hat{p}$ be a maximum-weight direction-preserving path in $\Graph$ (see \cref{figure:remark:midpoints-of-maximum-weight-paths-are-recognised}). Then, the ends $\hat{v}_1$ and $\hat{v}_2$ of $\hat{p}$ are leaves. And, among the first find-midpoint signals to reach the ends of $\hat{p}$ are the two that originated at the vertex $\hat{v}$ on $\hat{p}$ that is nearest to $\general$, and, because they travel alongside initiate signals, their reflections $s$ and $s'$ at the ends of $\hat{p}$ are marked. When one of them reaches a vertex on its way back it stays marked, because from all directions excluding from the direction it is headed but including the direction it is coming from have the marked reflected find-midpoint signals that originated at the vertex just or already returned, which is memorised by a count signal that is located at the vertex; the reason that such marked signals have already returned is that otherwise there would be a direction-preserving path with more weight than $\hat{p}$ that would be the concatenation the maximal subpath of $\hat{p}$ that lies in the direction the signal is headed and a path that begins with one of the edges from which no marked signal has returned yet. Note that the signals $s$ and $s'$ travel back alongside the marked reflected find-midpoint signals that originated at the vertices the signals $s$ and $s'$ passed by before they were reflected and that therefore, when they reach vertices on their way back, the count signals are just updated and hence up-to-date.
    \begin{figure}
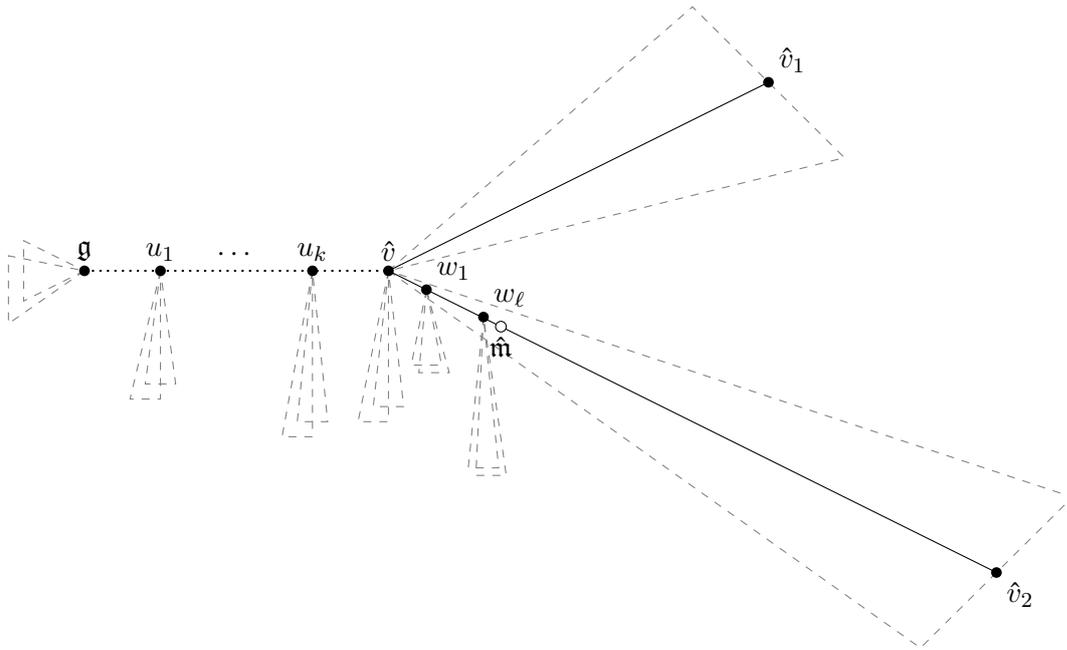

      \begin{wide}
        \figureRemarkMidpointsOfMaximumWeightPathsAreRecognised
        \caption{The path $\hat{p}$ is drawn solid, the direction-preserving path from $\mathfrak{g}$ to $\hat{v}$ is drawn dotted, the two subtrees in $\subtreesOf(\hat{v})$ that contain $\hat{v}_1$ and $\hat{v}_2$ are represented by dashed triangles, and, for each depicted vertex $v$, the possible existence of subtrees in $\subtreesOf(v)$ that correspond to non-depicted edges that are incident to $v$ is hinted at by dashed triangles.}
        \label{figure:remark:midpoints-of-maximum-weight-paths-are-recognised}
      \end{wide}
    \end{figure}

    When $s$ and $s'$ reach the vertex they originated at at the same time, then this vertex is found to be the midpoint of the path $\hat{p}$ and, because $s$ and $s'$ are marked, it is recognised as the midpoint of a maximum-weight path. When one of the signals, say $s$, reaches the vertex it originated at, namely $\hat{v}$, first, then it is slowed down and spreads throughout the graph away from the edge it comes from, in particular, one slowed-down find-midpoint signal, let us denote it by $s''$, travels towards $s'$. When $s''$ reaches a vertex on its way towards $s'$ it stays marked, because from all directions excluding from the direction it is headed have the \emph{slowest} reflected find-midpoint signals that originated at the vertex just or already returned, more precisely, from all directions excluding from the direction it is \emph{coming from and the one it is headed} have the marked reflected find-midpoint signals that originated at the vertex just or already returned and from the direction it is coming from has the slowest but unmarked reflected find-midpoint signal that originated at the vertex and the first marked slowed-down find-midpoint signal just returned. When $s''$ and $s'$ collide, the midpoint of $\hat{p}$ is found and, because $s''$ and $s'$ are marked, it is recognised as the midpoint of a maximum-weight path. 

    A detailed proof is given in the remainder of the present subsection.
  \end{proof-sketch}

  Symbolic notations for the predicates \emph{is a vertex of} and \emph{is an edge of} are introduced in

  \begin{definition}
    Let $\Tree = \ntuple{W, F}$ be a subtree of $\Graph$, let $v$ be a vertex of $\Graph$, and let $e$ be an edge of $\Graph$. We write $v \in \Tree$\graffito{$v \in \Tree$}\index[symbols]{vinTcalligraphic@$v \in \Tree$} instead of $v \in W$ and $e \in \Tree$\graffito{$e \in \Tree$}\index[symbols]{einTcalligraphic@$e \in \Tree$} instead of $e \in F$.
  \end{definition}

  The set of greatest subtrees of a vertex $v$ that correspond to its incident edges is named in

  \begin{definition}
    Let $v$ be a vertex of $\Graph$.
    \begin{aenumerate}
      \item Let $\Edges_v$ be the set of edges that are incident to $v$, and, for each edge $e \in \Edges_v$, let $\mathcal{T}_{v, e}$ be the greatest subtree of $\Graph$ that is rooted at $v$, contains the edge $e$, and does not contain any other edge of $\Edges_v$ (note that by \emph{greatest} we mean greatest with respect to the number of vertices). The set $\setOf{\mathcal{T}_{v, e} \suchThat e \in \Edges_v}$ is denoted by $\subtreesOf(v)$\graffito{set $\subtreesOf(v)$}\index[symbols]{sbtrsv@$\subtreesOf(v)$}.
      \item The set
            \begin{equation*}
              \maxSubtreesOf(v) = \argmax_{\Tree \in \subtreesOf(v)}\radiusOf{\Tree}_v
              \mathnote{set $\maxSubtreesOf(v)$ of trees}
              \index[symbols]{mxsbtrsv@$\maxSubtreesOf(v)$}
            \end{equation*}
            is the set of maximum-radius trees of $\subtreesOf(v)$. In the case that it is a singleton set, we denote its one and only element by $\maxTree_v$\graffito{tree $\maxTree_v$}\index[symbols]{Tvcalligraphichat@$\maxTree_v$}. 
      \item The set
            \begin{equation*}
              \secondMaxSubtreesOf(v) = \argmax_{\Tree \in \subtreesOf(v) \smallsetminus \maxSubtreesOf(v)}\radiusOf{\Tree}_v
              \mathnote{set $\secondMaxSubtreesOf(v)$ of trees}
              \index[symbols]{sbtrsvclosure@$\secondMaxSubtreesOf(v)$}
            \end{equation*}
            is the set of second-maximum-radius trees of $\subtreesOf(v)$. \qedhere
    \end{aenumerate}
  \end{definition}

  \begin{remark}
    For each leaf $v$ of $\Graph$, the set $\maxSubtreesOf(v)$ is a singleton and it is equal to $\subtreesOf(v)$, and the set $\secondMaxSubtreesOf(v)$ is empty.
  \end{remark}

  Things associated with a greatest subtree of a vertex are introduced in

  \begin{definition}
    Let $v$ be a vertex of $\Graph$ and let $\Tree$ be a tree of $\subtreesOf(v)$.
    \begin{aenumerate}
      \item Let $e$ be the edge of $\Tree$ that is incident to $v$. The direction that leads from $v$ onto $e$ is uniquely determined by $\Tree$ and is denoted by $\directionOf_v(\Tree)$\graffito{direction $\directionOf_v(\Tree)$}\index[symbols]{dirvTcalligraphic@$\directionOf_v(\Tree)$}. 
      \item The non-negative integer $\radiusOf{\Tree}_v = \max_{v' \in \Tree} \distanceOf(v, v')$ is called \graffito{radius $\radiusOf{\Tree}_v$ of $\Tree$ with respect to $v$}\define{radius of $\Tree$ with respect to $v$}\index[symbols]{normvsubscript@$\radiusOf{\blank}_v$}.
      \item Let $\Paths_{\Tree}$ be the set of direction-preserving paths in $\Tree$. The set
            \begin{equation*}
              \maxPaths_v(\Tree) = \setOf{p \in \Paths_{\Tree} \suchThat \sourceOf(p) = v \text{ and } \weightOf(p) = \radiusOf{\Tree}_v}
              \mathnote{set $\maxPaths_v(\Tree)$ of paths}
              \index[symbols]{mxpthsvTcalligraphic@$\maxPaths_v(\Tree)$}
            \end{equation*}
            is the set of maximum-weight direction-preserving paths from $v$ in $\Tree$. 
      \item The set
            \begin{equation*}
              \maxVertices_v(\Tree) = \setOf{\targetOf(p) \suchThat p \in \maxPaths_v(\Tree)}
              \mathnote{set $\maxVertices_v(\Tree)$ of vertices}
              \index[symbols]{mxvrtcsvTcalligraphic@$\maxVertices_v(\Tree)$}
            \end{equation*}
            is the set of vertices in $\Tree$ that are furthest away from $v$. \qedhere 
    \end{aenumerate}
  \end{definition}

  \begin{remark}
    We have $\cardinalityOf{\subtreesOf(v)} = \degreeOf(v)$ and $\directionOf_v(\subtreesOf(v)) = \directionOf(v)$.
  \end{remark}

  \begin{remark}
    Each vertex of $\maxVertices_v(\Tree)$ is a leaf.
  \end{remark}

  The unique direction-preserving path from one vertex to another is named in

  \begin{definition}
    Let $v$ and $v'$ be two vertices of $\Graph$. The direction-preserving path in $\Graph$ from $v$ to $v'$ is denoted by $p_{v, v'}$\graffito{path $p_{v, v'}$}\index[symbols]{pvvprime@$p_{v, v'}$} and the vertices on this path are denoted by $\Vertices_{v, v'}$\graffito{set $\Vertices_{v, v'}$ of vertices}\index[symbols]{Vvvprime@$\Vertices_{v, v'}$}.
  \end{definition}

  When and why count signals, initiate signals, and (maybe-marked slowed-down/reflected) find-midpoint signals are created and how they spread throughout the tree is said in

  \begin{remark} 
  \label{remark:how-initiate-and-find-midpoint-signals-spread}
    Let the signal machine $\mathcal{S}$ be in the configuration $c_{\initiateKind}$ at time $0$.
    \begin{aenumerate}
      \item At time $0$, a count signal with empty memory is created at $\general$, and initiate signals, find-midpoint signals with origin $\general$, and maybe-marked slowed-down find-midpoint signals with origin $\general$ and reflection vertex $\general$ are sent from $\general$ in all directions, where the slowed-down signal is marked if and only if $\general$ is a leaf. For each vertex $v \in \Vertices \smallsetminus \setOf{\general}$, at time $\distanceOf(\general, v)$, an initiate signal reaches $v$ from the direction towards $\general$, whereupon
            \begin{aenumerate}
              \item a count signal with empty memory is created at $v$,
              \item initiate signals are sent from $v$ in all directions away from $\general$, and
              \item find-midpoint signals with origin $v$, and maybe-marked slowed-down find-midpoint signals with origin $v$ and reflection vertex $v$ are sent from $v$ in all directions,
            \end{aenumerate}
            where the slowed-down signal is marked if and only if $v$ is a leaf.

            In short, initiate signals spread from $\general$ to all leaves, where they vanish, and they initiate the search for midpoints at all vertices. Note that to simplify the exposition of the forthcoming proofs, we also create count signals at leaves. 
      \item Let $v$ be a vertex of $\Graph$. As said above, at time $\distanceOf(\general, v)$, find-midpoint signals with origin $v$, and maybe-marked slowed-down find-midpoint signals with origin $v$ and reflection vertex $v$ are sent from $v$ in all directions, where the slowed-down signal is marked if and only if $v$ is a leaf. For each vertex $v' \in \Vertices \smallsetminus \setOf{v}$,
            \begin{aenumerate}
              \item at time $\distanceOf(\general, v) + \distanceOf(v, v')$, a find-midpoint signal with origin $v$ reaches $v'$ from the direction towards $v$, whereupon a maybe-marked reflected find-midpoint signal with origin $v$ and reflection vertex $v'$ is sent from $v'$ towards $v$ and find-midpoint signals with origin $v$ are sent from $v'$ in all directions away from $v$, where the reflected signal is marked if and only if $v$ is a leaf and an initiate signal just reached $v'$ (the latter is the case if and only if $\distanceOf(\general, v) + \distanceOf(v, v') = \distanceOf(\general, v')$), and,
              \item at time $\distanceOf(\general, v) + 3 \cdot \distanceOf(v, v')$, a maybe-marked slowed-down find-midpoint signal with origin $v$ and reflection vertex $v$ reaches $v'$ from the direction towards $v$, whereupon maybe-marked slowed-down find-midpoint signals with origin $v$ and reflection vertex $v$ are sent from $v'$ in all directions away from $v$. Note that even if the slowed-down signal that is sent from $v$ at time $\distanceOf(\general, v)$ is marked, the slowed-down signal that reaches $v'$ may be unmarked, and even if the latter signal is marked, the slowed-down signals that are sent from $v'$ may be unmarked.
            \end{aenumerate}
            In short, find-midpoint signals with origin $v$ and maybe-marked slowed-down find-midpoint signals with origin $v$ and reflection vertex $v$ spread from $v$ to all leaves, and whenever one of the former signals reaches a vertex, it is \emph{also} reflected. Note that to simplify the exposition of the forthcoming proofs, we talk as if maybe-marked slowed-down find-midpoint signals only vanish at leaves, although when a marked slowed-down find-midpoint signal collides with a marked reflected find-midpoint signal with the same origin, both signals vanish.
      \item Let $v$ and $v'$ be two vertices of $\Graph$ such that $v \neq v'$. As said above, at time $\distanceOf(\general, v) + \distanceOf(v, v')$ a maybe-marked reflected find-midpoint signal with origin $v$ and reflection vertex $v'$ is sent from $v'$ towards $v$. For each vertex $w$ on the direction-preserving path from $v'$ to $v$ except for $v'$, at time $\distanceOf(\general, v) + \distanceOf(v, v') + \distanceOf(v', w)$, the signal reaches $w$ from the direction towards $v'$, whereupon, if $w \neq v$, the signal is sent from $w$ towards $v$, and otherwise, maybe-marked slowed-down find-midpoint signals with origin $v$ and reflection vertex $v'$ are sent from $v$ in all directions away from $v'$. 

            And, for each vertex $v''$ of $\Graph$ that from the viewpoint of $v$ lies in a direction away from $v'$, which means that $v'' \neq v$ and there is a tree $\Tree \in \subtreesOf(v)$ such that $v' \notin \Tree$ and $v'' \in \Tree$, at time $\distanceOf(\general, v) + \distanceOf(v, v') + \distanceOf(v', v) + 3 \cdot \distanceOf(v, v'')$, a maybe-marked slowed-down find-midpoint signal with origin $v$ and reflection vertex $v'$ reaches $v''$, whereupon maybe-marked slowed-down find-midpoint signals with origin $v$ and reflection vertex $v'$ are sent from $v''$ in all directions away from $v$ (or, equivalently, away from $v'$). 

            In short, each maybe-marked reflected find-midpoint signal travels back to its origin and when it reaches its origin, it is slowed-down and spreads to all leaves away from its reflection vertex. Note that unmarked signals never become marked, but marked signals may become unmarked; precisely when the latter does or does not happen is answered in the present subsection.
    \end{aenumerate}
    What is said above is evident from the definition of the signal machine $\mathcal{S}$ (if it is carefully studied).
  \end{remark}

  When and why leaves sent \emph{marked} reflected/slowed-down find-midpoint signals is said in

  \begin{lemma}
  \label{lemma:when-do-leaves-sent-marked-signals}
    Let $v$ be a leaf of $\Graph$. At time $\distanceOf(\general, v)$, an initiate signal reaches $v$, for each vertex $w \in \Vertices_{\general, v} \smallsetminus \setOf{v}$, a find-midpoint signal with origin $w$ reaches $v$, and no other find-midpoint signal reaches $v$, whereupon
    \begin{aenumerate}
      \item for each vertex $w \in \Vertices_{\general, v} \smallsetminus \setOf{v}$, a marked reflected find-midpoint signal with origin $w$ and reflection vertex $v$ is sent from $v$ towards $w$, which is the only possible direction,
      \item a marked slowed-down find-midpoint signal with origin $v$ and reflection vertex $v$ is sent from $v$ towards $w$, and
      \item no other marked signal is sent from $v$.
    \end{aenumerate} 
    And before time $\distanceOf(\general, v)$ no signals reach and are sent from $v$, and after time $\distanceOf(\general, v)$ no marked signals are sent from $v$ (because after this time no initiate signal reaches $v$).
  \end{lemma}

  \begin{proof}
    This is a direct consequence of \cref{remark:how-initiate-and-find-midpoint-signals-spread} and the definition of $\mathcal{S}$.
  \end{proof}

  When does a signal that is sent from a vertex towards the general at a special time reach the next vertex is answered in

  \begin{lemma}
  \label{lemma:when-does-a-signal-that-is-sent-from-a-vertex-towards-the-general-reach-the-next-vertex}
    Let $v$ be a vertex of $\Graph$, let $\Tree$ be a tree of $\subtreesOf(v) \smallsetminus \maxSubtreesOf(v)$ such that either $v = \general$ or $\general \notin \Tree$, and let $v'$ be the one and only neighbour of $v$ in $\Tree$. Then,
    \begin{aenumerate}
      \item $\maxSubtreesOf(v')$ is a singleton set and its only element $\maxTree_{v'}$ contains $v$,
      \item $\displaystyle
              \maxVertices_v(\Tree) = \begin{dcases*}
                                        \setOf{v'}, &if $v'$ is a leaf,\\
                                        \bigDisjointUnionOf_{\Tree' \in \secondMaxSubtreesOf(v')} \maxVertices_{v'}(\Tree'), &otherwise,
                                      \end{dcases*}
            $
      \item $\radiusOf{\Tree}_v = \distanceOf(v, v') + \max_{\Tree' \in \secondMaxSubtreesOf(v')}\radiusOf{\Tree'}_{v'}$, and
      \item when a signal of speed $1$ is sent from $v'$ towards $v$ at time $\distanceOf(\general, v') + 2 \cdot \max_{\Tree' \in \secondMaxSubtreesOf(v')}\radiusOf{\Tree'}_{v'}$, it reaches $v$ at time $\distanceOf(\general, v) + 2 \cdot \radiusOf{\Tree}_v$, where in the case that $\secondMaxSubtreesOf(v')$ is empty, we define $\max_{\Tree' \in \secondMaxSubtreesOf(v')}\radiusOf{\Tree'}_{v'}$ as $0$. \qedhere
    \end{aenumerate}
  \end{lemma}

  \begin{proof}
    The first item is evident, the second and third follow from it, and the fourth follows from the third with $\distanceOf(\general, v') = \distanceOf(\general, v) + \distanceOf(v, v')$ and the fact that the signal needs the time span $\distanceOf(v, v')$ to traverse the edge from $v'$ to $v$.
  \end{proof}

  The set of all non-leaf vertices whose unique maximum-radius tree contains the general is named in

  \begin{definition} 
    The set of all non-leaf vertices $v$ of $\Graph$ such that $\maxSubtreesOf(v)$ is a singleton set and its only element $\maxTree_v$ contains the vertex $\general$, is denoted by $V_{\general}$\graffito{set $V_{\general}$ of vertices}\index[symbols]{Vgsubscript@$V_{\general}$}.
  \end{definition}

  \begin{remark}
    For each vertex $v \in V_{\general}$, each tree $\Tree \in \subtreesOf(v) \smallsetminus \setOf{\maxTree_v}$, and each vertex $v'$ of $\Tree$, we have $v' \in V_{\general}$. And, for each vertex $v \in V_{\general}$, the set $\secondMaxSubtreesOf(v)$ is non-empty and, for each tree $\Tree \in \secondMaxSubtreesOf(v)$, we have $\radiusOf{\Tree}_v = \max_{\Tree \in \secondMaxSubtreesOf(v)}\radiusOf{\Tree}_v$.
  \end{remark}

  When and why non-leaf vertices whose unique maximum-radius tree contains the general sent \emph{marked} reflected/slowed-down find-midpoint signals is shown in

  \begin{lemma}
  \label{lemma:when-do-non-leaf-vertices-whose-unique-maximum-radius-tree-contains-the-general-sent-marked-signals}
    For each vertex $v \in V_{\general}$, at time $\distanceOf(\general, v) + 2 \cdot \max_{\Tree \in \secondMaxSubtreesOf(v)}\radiusOf{\Tree}_v$, 
    \begin{aenumerate}
      \item the count signal at $v$, before it is updated, has the memory $\directionOf(v) \smallsetminus \directionOf_v(\secondMaxSubtreesOf(v) \cup \setOf{\maxTree_v})$,
      \item for each tree $\Tree \in \secondMaxSubtreesOf(v)$, each leaf $\cev{v} \in \maxVertices_v(\Tree)$, and each vertex $w \in \Vertices_{\general, v}$, a marked reflected find-midpoint signal with origin $w$ and reflection vertex $\cev{v}$ reaches $v$ from direction $\directionOf_v(\Tree)$, and
      \item no other marked reflected find-midpoint signal reaches $v$, 
    \end{aenumerate}
    whereupon
    \begin{aenumerate}
      \item the count signal at $v$, after it is updated, has the memory $\directionOf(v) \smallsetminus \setOf{\directionOf_v(\maxTree_v)}$,
      \item for each tree $\Tree \in \secondMaxSubtreesOf(v)$, each leaf $\cev{v} \in \maxVertices_v(\Tree)$, and each vertex $w \in \Vertices_{\general, v} \smallsetminus \setOf{v}$, a marked reflected find-midpoint signal with origin $w$ and reflection vertex $\cev{v}$ is sent from $v$ in the direction $\directionOf_v(\maxTree_v)$ towards $w$, and no other marked reflected find-midpoint signal with origin $w$ and reflection vertex $\cev{v}$ is sent from $v$, and,
      \item for each tree $\Tree \in \secondMaxSubtreesOf(v)$, each leaf $\cev{v} \in \maxVertices_v(\Tree)$, and each direction $d \in \directionOf(v) \smallsetminus \setOf{\directionOf_v(\Tree)}$, a marked slowed-down find-midpoint signal with origin $v$ and reflection vertex $\cev{v}$ is sent from $v$ in direction $d$ (note that $\directionOf_v(\maxTree_v) \in \directionOf(v) \smallsetminus \setOf{\directionOf_v(\Tree)}$), and, no other marked slowed-down find-midpoint signal with origin $v$ and reflection vertex $\cev{v}$ is sent from $v$. 
    \end{aenumerate}
    And before time $\distanceOf(\general, v) + 2 \cdot \max_{\Tree \in \secondMaxSubtreesOf(v)}\radiusOf{\Tree}_v$, marked signals may reach $v$ but no marked signal is sent from $v$, and after that time, no marked signals as above are sent from $v$. 
  \end{lemma}

  \begin{proof}
    We prove this by induction on $n_v = \max_{(W, F) \in \subtreesOf(v) \smallsetminus \setOf{\maxTree_v}}\cardinalityOf{F}$, for $v \in V_{\general}$. For brevity though, we only treat the existence of signals and not their absence. 

    \proofPart{Base Case (see \cref{figure:when-does-count-signal-memorise-the-penultimate-direction:base-case})}
      Let $v \in V_{\general}$ such that $n_v = 1$. And, let $\Tree \in \subtreesOf(v) \smallsetminus \setOf{\maxTree_v}$. Then, $\Tree$ consists of one edge $e$ whose one end is $v$, whose other end is a leaf $v'$, and whose weight is $\radiusOf{\Tree}_v$. According to \cref{lemma:when-do-leaves-sent-marked-signals}, at time $\distanceOf(\general, v') = \distanceOf(\general, v) + \radiusOf{\Tree}_v$, for each vertex $w \in \Vertices_{\general, v'} \smallsetminus \setOf{v'} = \Vertices_{\general, v}$, a marked reflected find-midpoint signal with origin $w$ and reflection vertex $v'$ is sent from $v'$ towards $v$, in particular, one with origin $v$. These signals reach $v$ at time $\distanceOf(\general, v) + 2 \cdot \radiusOf{\Tree}_v$, whereupon
      \begin{aenumerate}
        \item the count signal at $v$ memorises the direction $\directionOf_v(\Tree)$ (because a marked reflected find-midpoint signal with origin $v$ reached $v$),
        \item for each vertex $w \in \Vertices_{\general, v} \smallsetminus \setOf{v}$, a maybe-marked reflected find-midpoint signal with origin $w$ and reflection vertex $v'$ is sent from $v$ in the direction $\directionOf_v(\maxTree_v)$ towards $w$, and
        \item for each direction $d \in \directionOf(v) \smallsetminus \setOf{\directionOf_v(\Tree)}$, a maybe-marked slowed-down find-midpoint signal with origin $v$ and reflection vertex $v'$ is sent from $v$ in direction $d$.
      \end{aenumerate}
      On the timeline, for the trees of $\subtreesOf(v) \smallsetminus \setOf{\maxTree_v}$ in non-decreasing order with respect to the radius and at the same time for trees with the same radius, the signals reach $v$ and are sent from $v$. For those trees whose radius is less than the second greatest radius among the trees of $\subtreesOf(v)$, which is $\max_{\Tree \in \secondMaxSubtreesOf(v)}\radiusOf{\Tree}_v$, the aforementioned maybe-marked signals that are sent from $v$ are unmarked (because the memory of the count signal, after it is updated, does neither contain the directions of $\directionOf_v(\secondMaxSubtreesOf(v))$ nor the direction $\directionOf_v(\maxTree_v)$). And, for the trees of $\secondMaxSubtreesOf(v)$, the aforementioned maybe-marked signals that are sent from $v$ are marked (because the count signal, after it is updated, has the memory $\directionOf(v) \smallsetminus \setOf{\directionOf_v(\maxTree_v)}$). In conclusion, at time $\distanceOf(\general, v) + 2 \cdot \max_{\Tree \in \secondMaxSubtreesOf(v)}\radiusOf{\Tree}_v$, what is to be proven holds.
      \begin{figure}
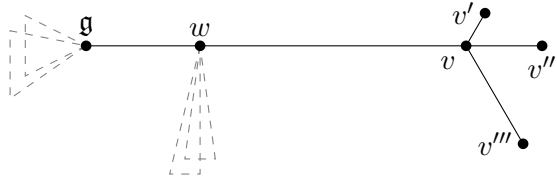

        \myfloatalign
        \figureWhenDoesCountSignalMemoriseThePenultimateDirectionBaseCase
        \caption{Schematic representation of the set-up of the base case of the proof of \cref{lemma:when-do-non-leaf-vertices-whose-unique-maximum-radius-tree-contains-the-general-sent-marked-signals}. The possible existence of subtrees of vertices is hinted at by dashed triangles.}
        \label{figure:when-does-count-signal-memorise-the-penultimate-direction:base-case}
      \end{figure}

    \proofPart{Inductive Step (see \cref{figure:when-does-count-signal-memorise-the-penultimate-direction:inductive-step})}
      Let $v \in V_{\general}$ such that $n_v \geq 2$ and such that what is to be proven holds for each vertex $v' \in V_{\general}$ with $n_{v'} < n_v$, which is the \emph{inductive hypothesis}. And, let $\Tree \in \subtreesOf(v) \smallsetminus \setOf{\maxTree_v}$, let $e$ be the edge of $\Tree$ whose one end is $v$, and let $v'$ be the other end of $e$. The vertex $v'$ is either a leaf or an element of $V_{\general}$ with $n_{v'} < n_v$.
      \begin{figure}
        \myfloatalign
        \figureWhenDoesCountSignalMemoriseThePenultimateDirectionInductiveStep
        \caption{Schematic representation of the set-up of the inductive step of the proof of \cref{lemma:when-do-non-leaf-vertices-whose-unique-maximum-radius-tree-contains-the-general-sent-marked-signals}. The possible existence of subtrees of vertices is hinted at by dashed triangles.}
        \label{figure:when-does-count-signal-memorise-the-penultimate-direction:inductive-step}
      \end{figure}

      In the first case, according to \cref{lemma:when-do-leaves-sent-marked-signals}, at time $\distanceOf(\general, v') = \distanceOf(\general, v) + \distanceOf(v, v')$, for each vertex $w \in \Vertices_{\general, v'} \smallsetminus \setOf{v'} = \Vertices_{\general, v}$, a marked reflected find-midpoint signal with origin $w$ and reflection vertex $v'$ is sent from $v'$ towards $v$; note that $v'$ is the one and only element of $\maxVertices_v(\Tree)$, the set $\secondMaxSubtreesOf(v')$ is empty, and we define $\max_{\Tree' \in \secondMaxSubtreesOf(v')}\radiusOf{\Tree'}_{v'}$ as $0$ (see \cref{lemma:when-does-a-signal-that-is-sent-from-a-vertex-towards-the-general-reach-the-next-vertex}). In the second case, according to the inductive hypothesis, at time $\distanceOf(\general, v') + 2 \cdot \max_{\Tree' \in \secondMaxSubtreesOf(v')}\radiusOf{\Tree'}_{v'}$, for each tree $\Tree' \in \secondMaxSubtreesOf(v')$, each leaf $\cev{v}' \in \maxVertices_{v'}(\Tree')$, and each vertex $w \in \Vertices_{\general, v'} \smallsetminus \setOf{v'} = \Vertices_{\general, v}$, a marked reflected find-midpoint signal with origin $w$ and reflection vertex $\cev{v}'$ is sent from $v'$ towards $v$; note that the set of all vertices $\cev{v}'$ is equal to $\maxVertices_v(\Tree)$ (see \cref{lemma:when-does-a-signal-that-is-sent-from-a-vertex-towards-the-general-reach-the-next-vertex}). 

      In both cases, the marked signals reach $v$ at time $\distanceOf(\general, v) + \distanceOf(v, v') + 2 \cdot \max_{\Tree' \in \secondMaxSubtreesOf(v')}\radiusOf{\Tree'}_{v'} + \distanceOf(v', v)$, which is equal to $\distanceOf(\general, v) + 2 \cdot \radiusOf{\Tree}_v$, whereupon 
      \begin{aenumerate}
        \item the count signal at $v$ memorises the direction $\directionOf_v(\Tree)$ (because at least one marked reflected find-midpoint signal with origin $v$ reached $v$),
        \item for each leaf $\cev{v} \in \maxVertices_v(\Tree)$ and each vertex $w \in \Vertices_{\general, v} \smallsetminus \setOf{v}$, a maybe-marked reflected find-midpoint signal with origin $w$ and reflection vertex $\cev{v}$ is sent from $v$ in the direction $\directionOf_v(\maxTree_v)$ towards $w$, and,
      \item for each leaf $\cev{v} \in \maxVertices_v(\Tree)$ and each direction $d \in \directionOf(v) \smallsetminus \setOf{\directionOf_v(\Tree)}$, a maybe-marked slowed-down find-midpoint signal with origin $v$ and reflection vertex $\cev{v}$ is sent from $v$ in direction $d$.
      \end{aenumerate}
      It follows verbatim as in the base case, that what is to be proven holds.
  \end{proof}

  %

  The set of all vertices whose unique maximum-radius tree does not contain the general is named in

  \begin{definition}
    The set of all vertices $v$ of $\Graph$ such that $\maxSubtreesOf(v)$ is a singleton set and its only element $\maxTree_v$ does \emph{not} contain the vertex $\general$, is denoted by $U_{\general}$\graffito{set $U_{\general}$ of vertices}\index[symbols]{Ugsubscript@$U_{\general}$}.
  \end{definition}

  \begin{remark}
    We have $\general \notin U_{\general}$. And, for each vertex $v \in U_{\general}$, we have $\Vertices_{\general, v} \smallsetminus \setOf{\general} \in U_{\general}$, the vertex $v$ is not a leaf and the set $\secondMaxSubtreesOf(v)$ is non-empty. 
  \end{remark}

  The maximal subtree of a non-general vertex that contains the general is named in

  \begin{definition}
    Let $v$ be a vertex of $\Graph$ such that $v \neq \general$. The tree of $\subtreesOf(v)$ that does contain the vertex $\general$ is denoted by $\Tree_v^{\general}$\graffito{tree $\Tree_v^{\general}$}\index[symbols]{Tvgcalligraphic@$\Tree_v^{\general}$}. And to avoid case differentiations, we define $\Tree_{\general}^{\general}$\graffito{$\Tree_{\general}^{\general} = 0$}\index[symbols]{Tggcalligraphic@$\Tree_{\general}^{\general}$} as the number $0$.
  \end{definition}

  The vertices of the maximal path from a vertex to the general whose vertices excluding its source have the subtree that contains the general as second-maximum-radius subtree is named in

  \begin{definition}
    Let $v$ be a vertex of $U_{\general}$, and let $p$ be the maximum-weight subpath of $p_{\general, v}$ such that $\targetOf(p) = v$ and, for each vertex $w$ on $p$ with $w \neq \sourceOf(p)$, we have $\Tree_w^{\general} \in \secondMaxSubtreesOf(w)$. The set of the vertices on $p$ is denoted by $W_{\general, v}$\graffito{set $W_{\general, v}$ of vertices}\index[symbols]{Wgvsubscript@$W_{\general, v}$}. See \cref{figure:W-general-v}.
    \begin{figure}
      \myfloatalign
      \figureWGeneralV
      \caption{Schematic representation of $\Graph$, where $v \in U_{\general}$, $W_{\general, v} = \setOf{w'', w', w, v}$, $W_{\general, w} = \setOf{w'', w', w}$, $W_{\general, w'} = \setOf{w'', w'}$, $W_{\general, w''} = \setOf{w''}$, and subtrees of vertices and their radii are depicted as dashed triangles of various sizes.}
      \label{figure:W-general-v}
    \end{figure}
  \end{definition}

  \begin{remark}
    We have $v \in W_{\general, v} \smallsetminus \setOf{\general} \subseteq U_{\general}$. And, for each $w \in W_{\general, v}$, we have $W_{\general, w} \subseteq W_{\general, v}$. And, if $\Tree_v^{\general} \in \secondMaxSubtreesOf(v)$, then $W_{\general, v} = W_{\general, v'} \cup \setOf{v}$, where $v'$ is the neighbour of $v$ on $p_{\general, v}$.
  \end{remark}

  When and why vertices whose unique maximum-radius tree does not contain the general sent \emph{marked} reflected/slowed-down find-midpoint signals is shown in

  \begin{lemma}
  \label{lemma:when-do-vertices-whose-unique-maximum-radius-tree-does-not-contain-the-general-sent-marked-signals}
    For each vertex $v \in U_{\general}$, at time $\distanceOf(\general, v) + 2 \cdot \max_{\Tree \in \secondMaxSubtreesOf(v)}\radiusOf{\Tree}_v$,
    \begin{aenumerate}
      \item the count signal at $v$, before it is updated, has the memory $\directionOf(v) \smallsetminus \directionOf_v(\secondMaxSubtreesOf(v) \cup \setOf{\maxTree_v})$,
      \item for each tree $\Tree \in \secondMaxSubtreesOf(v)$ and each leaf $\cev{v} \in \maxVertices_v(\Tree)$, a maybe-marked reflected find-midpoint signal with origin $v$ and reflection vertex $\cev{v}$ reaches $v$ from direction $\directionOf_v(\Tree)$, where the reflected signal is marked if and only if $\Tree \neq \Tree_v^{\general}$, and,
      \item for each vertex $w \in W_{\general, v} \smallsetminus \setOf{v}$, each tree $\Tree \in \secondMaxSubtreesOf(w) \smallsetminus \setOf{\Tree_w^{\general}}$, and each leaf $\cev{w} \in \maxVertices_w(\Tree)$, a marked slowed-down find-midpoint signal with origin $w$ and reflection vertex $\cev{w}$ reaches $v$ from the direction towards $w$, which is the direction $\directionOf_v(\Tree_v^{\general})$,
    \end{aenumerate}
    whereupon
    \begin{aenumerate}
      \item the count signal at $v$, after it is updated, has the memory $\directionOf(v) \smallsetminus \setOf{\directionOf_v(\maxTree_v)}$,
      \item for each tree $\Tree \in \secondMaxSubtreesOf(v) \smallsetminus \setOf{\Tree_v^{\general}}$, each leaf $\cev{v} \in \maxVertices_v(\Tree)$, and each direction $d \in \directionOf(v) \smallsetminus \setOf{\directionOf_v(\Tree)}$, a marked slowed-down find-midpoint signal with origin $v$ and reflection vertex $\cev{v}$ is sent from $v$ in direction $d$ (note that $\setOf{\directionOf_v(\maxTree_v), \directionOf_v(\Tree_v^{\general})} \subseteq \directionOf(v) \smallsetminus \setOf{\directionOf_v(\Tree)}$), and,
      \item for each vertex $w \in W_{\general, v} \smallsetminus \setOf{v}$, each tree $\Tree \in \secondMaxSubtreesOf(w) \smallsetminus \setOf{\Tree_w^{\general}}$, each leaf $\cev{w} \in \maxVertices_w(\Tree)$, and each direction $d \in \directionOf(v) \smallsetminus \setOf{\directionOf_v(\Tree_v^{\general})}$, a marked slowed-down find-midpoint signal with origin $w$ and reflection vertex $\cev{w}$ is sent from $v$ in direction $d$ (note that $\directionOf_v(\maxTree_v) \in \directionOf(v) \smallsetminus \setOf{\directionOf_v(\Tree_v^{\general})}$ and that, if $w = \general$, then $\secondMaxSubtreesOf(w) \smallsetminus \setOf{\Tree_w^{\general}} = \secondMaxSubtreesOf(w)$). \qedhere
    \end{aenumerate} 
  \end{lemma}

  \begin{proof}
    We prove this by induction on $n_v = \cardinalityOf{\Vertices_{\general, v}}$, for $v \in U_{\general}$.

    \proofPart{Base Case (compare \cref{figure:marked-signals-inner-vertices:inductive-step})} 
      Let $v \in U_{\general}$ such that $n_v = 1$.

      First, for each tree $\Tree \in \subtreesOf(v) \smallsetminus \setOf{\Tree_v^{\general}, \maxTree_v}$, according to \cref{lemma:when-do-leaves-sent-marked-signals}, if the neighbour of $v$ in $\Tree$ is a leaf, or \cref{lemma:when-do-non-leaf-vertices-whose-unique-maximum-radius-tree-contains-the-general-sent-marked-signals}, otherwise, and \cref{lemma:when-does-a-signal-that-is-sent-from-a-vertex-towards-the-general-reach-the-next-vertex}, at time $\distanceOf(\general, v) + 2 \cdot \radiusOf{\Tree}_v$, for each leaf $\cev{v} \in \maxVertices_v(\Tree)$, a marked reflected find-midpoint signal with origin $v$ and reflection vertex $\cev{v}$ reaches $v$ from direction $\directionOf_v(\Tree)$, whereupon
      \begin{aenumerate}
        \item the count signal at $v$ memorises the direction $\directionOf_v(\Tree)$,
        \item for each leaf $\cev{v} \in \maxVertices_v(\Tree)$ and each direction $d \in \directionOf(v) \smallsetminus \setOf{\directionOf_v(\Tree)}$, a maybe-marked slowed-down find-midpoint signal with origin $v$ and reflection vertex $\cev{v}$ is sent from $v$ in direction $d$.
      \end{aenumerate}

      Secondly, according to \cref{remark:how-initiate-and-find-midpoint-signals-spread}, at time $\distanceOf(\general, v) + 2 \cdot \radiusOf{\Tree_v^{\general}}_v$, for each leaf $\cev{v} \in \maxVertices_v(\Tree_v^{\general})$, a maybe-marked reflected find-midpoint signal with origin $v$ and reflection vertex $\cev{v}$ reaches $v$ from direction $\directionOf_v(\Tree_v^{\general})$. This signal is actually unmarked, because it was reflected at $\cev{v}$ at time $\distanceOf(\general, v) + \radiusOf{\Tree_v^{\general}}_v = \distanceOf(\general, v) + \distanceOf(v, \cev{v})$, which, because $v \neq \general$ and $\cev{v} \in \Tree_v^{\general}$, is greater than the only time, namely $\distanceOf(\general, \cev{v})$, at which an initiate signal reaches $\cev{v}$.

      Thirdly, because $U_{\general}$ is non-empty, we have $\general \in V_{\general}$ and $v \in \maxTree_{\general}$. According to \cref{lemma:when-do-non-leaf-vertices-whose-unique-maximum-radius-tree-contains-the-general-sent-marked-signals}, at time $2 \cdot \max_{\Tree \in \secondMaxSubtreesOf(\general)}\radiusOf{\Tree}_{\general}$, for each tree $\Tree \in \secondMaxSubtreesOf(\general)$, and each leaf $\cev{\general} \in \maxVertices_{\general}(\Tree)$, a marked slowed-down find-midpoint signal with origin $\general$ and reflection vertex $\cev{\general}$ is sent from $\general$ towards $v$; note that the set of all vertices $\cev{\general}$ is equal to $\maxVertices_v(\Tree_v^{\general})$ (compare \cref{lemma:when-does-a-signal-that-is-sent-from-a-vertex-towards-the-general-reach-the-next-vertex}). The marked signals reach $v$ from the direction towards $\general$, which is the direction $\directionOf_v(\Tree_v^{\general})$, at time $2 \cdot \max_{\Tree \in \secondMaxSubtreesOf(\general)}\radiusOf{\Tree}_{\general} + 3 \cdot \distanceOf(\general, v) = \distanceOf(\general, v) + 2 \cdot \radiusOf{\Tree_v^{\general}}_v$, whereupon
      \begin{aenumerate}
        \item the count signal at $v$ memorises the direction $\directionOf_v(\Tree_v^{\general})$ and,
        \item for each leaf $\cev{\general} \in \maxVertices_v(\Tree_v^{\general})$ and each direction $d \in \directionOf(v) \smallsetminus \setOf{\directionOf_v(\Tree_v^{\general})}$, a maybe-marked slowed-down find-midpoint signal with origin $\general$ and reflection vertex $\cev{\general}$ is sent from $v$ in direction $d$.
      \end{aenumerate} 

      Altogether, on the timeline, for the trees of $\subtreesOf(v) \smallsetminus \setOf{\maxTree_v}$ in non-decreasing order with respect to the radius and at the same time for trees with the same radius, the signals reach $v$ and are sent from $v$. For those trees whose radius is less than the second greatest radius among the trees of $\subtreesOf(v)$, which is $\max_{\Tree \in \secondMaxSubtreesOf(v)}\radiusOf{\Tree}_v$, the aforementioned maybe-marked signals that are sent from $v$ are unmarked (because the memory of the count signal, after it is updated, does neither contain the directions of $\directionOf_v(\secondMaxSubtreesOf(v))$ nor the direction $\directionOf_v(\maxTree_v)$). And, for the trees of $\secondMaxSubtreesOf(v)$, the aforementioned maybe-marked signals that are sent from $v$ are marked (because the count signal, after it is updated, has the memory $\directionOf(v) \smallsetminus \setOf{\directionOf_v(\maxTree_v)}$). Note that, if $\Tree_v^{\general} \in \secondMaxSubtreesOf(v)$, then $W_{\general, v} \smallsetminus \setOf{v} = \setOf{\general}$, and otherwise, $W_{\general, v} \smallsetminus \setOf{v} = \emptyset$. In conclusion, at time $\distanceOf(\general, v) + 2 \cdot \max_{\Tree \in \secondMaxSubtreesOf(v)}\radiusOf{\Tree}_v$, what is to be proven holds. 

    \proofPart{Inductive Step (see \cref{figure:marked-signals-inner-vertices:inductive-step})}
      Let $v \in U_{\general}$ such that $n_v \geq 2$ and such that what is to be proven holds for each vertex $v' \in U_{\general}$ with $n_{v'} < n_v$, which is the \emph{inductive hypothesis}.
      \begin{figure}
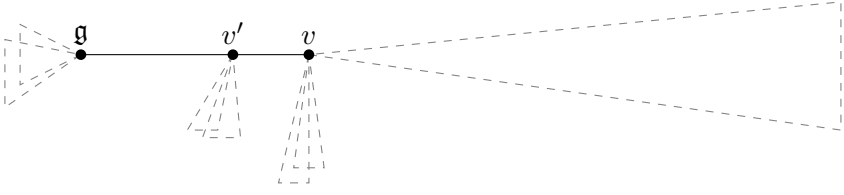

        \myfloatalign
        \figureMarkedSignalsInnerVerticesInducitveStep
        \caption{Schematic representation of the set-up of the inductive step of the proof of \cref{lemma:when-do-vertices-whose-unique-maximum-radius-tree-does-not-contain-the-general-sent-marked-signals}. The possible existence of subtrees of vertices and their radii is hinted at by dashed triangles of various sizes.}
        \label{figure:marked-signals-inner-vertices:inductive-step}
      \end{figure}

      First, for each tree $\Tree \in \subtreesOf(v) \smallsetminus \setOf{\Tree_v^{\general}, \maxTree_v}$, the same as in the base case happens.

      Secondly, as in the base case, at time $\distanceOf(\general, v) + 2 \cdot \radiusOf{\Tree_v^{\general}}_v$, for each leaf $\cev{v} \in \maxVertices_v(\Tree_v^{\general})$, an unmarked reflected find-midpoint signal with origin $v$ and reflection vertex $\cev{v}$ reaches $v$ from direction $\directionOf_v(\Tree_v^{\general})$.

      Thirdly, let $v'$ be the neighbour of $v$ in $\Tree_v^{\general}$. Then, $v' \in \Vertices_{\general, v} \smallsetminus \setOf{\general} \subseteq U_{\general}$ and $n_{v'} < n_v$. According to the inductive hypothesis, at time $\distanceOf(\general, v') + 2 \cdot \max_{\Tree' \in \secondMaxSubtreesOf(v')}\radiusOf{\Tree'}_{v'}$,
      \begin{aenumerate}
        \item for each tree $\Tree' \in \secondMaxSubtreesOf(v') \smallsetminus \setOf{\Tree_{v'}^{\general}}$ and each leaf $\cev{v}' \in \maxVertices_{v'}(\Tree')$, a marked slowed-down find-midpoint signal with origin $v'$ and reflection vertex $\cev{v}'$ is sent from $v'$ towards $v$, and,
        \item for each vertex $w' \in W_{\general, v'} \smallsetminus \setOf{v'}$, each tree $\Tree' \in \secondMaxSubtreesOf(w') \smallsetminus \setOf{\Tree_{w'}^{\general}}$, each leaf $\cev{w}' \in \maxVertices_{w'}(\Tree')$, a marked slowed-down find-midpoint signal with origin $w'$ and reflection vertex $\cev{w}'$ is sent from $v'$ towards $v$.
      \end{aenumerate}
      The marked signals reach $v$ from from direction $\directionOf_v(\Tree_v^{\general})$ at time $\distanceOf(\general, v') + 2 \cdot \max_{\Tree' \in \secondMaxSubtreesOf(v')}\radiusOf{\Tree'}_{v'} + 3 \cdot \distanceOf(v', v) = \distanceOf(\general, v) + 2 \cdot \radiusOf{\Tree_v^{\general}}_v$, whereupon
      \begin{aenumerate}
        \item the count signal at $v$ memorises the direction $\directionOf_v(\Tree_v^{\general})$ and,
        \item for each vertex $w \in W_{\general, v'}$, each tree $\Tree \in \secondMaxSubtreesOf(w) \smallsetminus \setOf{\Tree_w^{\general}}$, each leaf $\cev{w} \in \maxVertices_w(\Tree)$, and each direction $d \in \directionOf(v) \smallsetminus \setOf{\directionOf_v(\Tree_v^{\general})}$, a maybe-marked slowed-down find-midpoint signal with origin $w$ and reflection vertex $\cev{w}$ is sent from $v$ in direction $d$.
      \end{aenumerate}
      Note that, if $\Tree_v^{\general} \in \secondMaxSubtreesOf(v)$, then $W_{\general, v} \smallsetminus \setOf{v} = W_{\general, v'}$, and otherwise, $W_{\general, v} \smallsetminus \setOf{v} = \emptyset$. With that it follows verbatim as in the base case, that what is to be proven holds. 
  \end{proof}

  That midpoints of maximum-weight direction-preserving paths are recognised as such is shown in

  \begin{theorem} 
  \label{theorem:midpoints-of-maximum-weight-paths-are-recognised}
    Let $\hat{p}$ be a maximum-weight direction-preserving path in $\Graph$, let $\hat{v}$ be the vertex on $\hat{p}$ that is nearest to $\general$, let $\hat{v}_1$ and $\hat{v}_2$ be the two ends of $\hat{p}$ such that $\distanceOf(\hat{v}, \hat{v}_1) \leq \distanceOf(\hat{v}, \hat{v}_2)$, and let $\hat{\midpoint}$ be the midpoint of $\hat{p}$. At time $r + d/2$, at the midpoint $\hat{\midpoint}$,
    \begin{aenumerate}
      \item if $\distanceOf(\hat{v}, \hat{v}_1) = \distanceOf(\hat{v}, \hat{v}_2)$, two marked reflected find-midpoint signals with origin $\hat{v}$ and reflection vertices $\hat{v}_1$ and $\hat{v}_2$ collide, and
      \item otherwise, a marked slowed-down find-midpoint signal with origin $\hat{v}$ and reflection vertex $\hat{v}_1$ collides with a marked reflected find-midpoint signal with origin $\hat{v}$ and reflection vertex $\hat{v}_2$. \qedhere
    \end{aenumerate}
  \end{theorem}

  \begin{proof-sketch}
    From a broad perspective and ignoring boundary cases the following happens in the given order (see \cref{figure:midpoints-of-maximum-weight-paths-are-recognised:midpoint-is-not-nearest-to-general}). At time $0$, an initiate signal is sent from $\general$ towards $\hat{v}$. At time $\distanceOf(\general, \hat{v})$, this signal reaches $\hat{v}$, whereupon find-midpoint signals with origin $\hat{v}$ are sent from $\hat{v}$ in all directions. At time $\distanceOf(\general, \hat{v}) + \radiusOf{\Tree_{\hat{v}}^{\general}}_{\hat{v}}$, the slowest but unmarked reflected find-midpoint signal with origin $\hat{v}$ returns to $\hat{v}$ from the direction towards $\general$. At time $\distanceOf(\general, \hat{v}) + 2 \cdot \distanceOf(\hat{v}, \hat{v}_1)$, the slowest and marked reflected find-midpoint signal with origin $\hat{v}$ returns to $\hat{v}$ from the direction towards $\hat{v}_1$, whereupon a marked slowed-down find-midpoint signal with origin $\hat{v}$ is sent towards $\hat{v}_2$. At time $\distanceOf(\general, \hat{v}) + 2 \cdot \distanceOf(\hat{v}, \hat{v}_1) + 3 \cdot \distanceOf(\hat{v}, \hat{\midpoint}) = \distanceOf(\general, \hat{\midpoint}) + 2 \cdot \distanceOf(\hat{v}_1, \hat{\midpoint}) = r + d/2$, this signal reaches $\hat{\midpoint}$. And, at the same time, which is equal to $\distanceOf(\general, \hat{v}) + \distanceOf(\hat{v}, \hat{v}_2) + \distanceOf(\hat{v}_2, \hat{\midpoint})$, the slowest and marked reflected find-midpoint signal with origin $\hat{v}$ that is on its way to return to $\hat{v}$ from the direction towards $\hat{v}_2$ reaches $\hat{\midpoint}$. The two marked signals collide at $\hat{\midpoint}$ recognising it as the midpoint of a maximum-weight path.

    The slowest reflected find-midpoint signal with origin $\hat{v}$ that returns to $\hat{v}$ from the direction towards $\general$ is unmarked, because it reaches the leaf it is reflected at later than the initiate signal and is never marked in the first place. The slowest reflected find-midpoint signal with origin $\hat{v}$ that returns to $\hat{v}$ from the direction towards $\hat{v}_1$ is marked, because it reaches the leaf $\hat{v}_1$ alongside initiate signals and at each vertex on its way back, it is the penultimate marked signal to return (this is essentially \cref{lemma:when-do-non-leaf-vertices-whose-unique-maximum-radius-tree-contains-the-general-sent-marked-signals}). For the same reason, the signal from the direction towards $\hat{v}_2$ is marked when it reaches $\hat{\midpoint}$. And, the slowed-down find-midpoint signal with origin $\hat{v}$ that reaches $\hat{\midpoint}$ from the direction towards $\hat{v}$ is marked, because at each vertex it reaches on its way, it is the penultimate marked signal to do so (this is essentially \cref{lemma:when-do-vertices-whose-unique-maximum-radius-tree-does-not-contain-the-general-sent-marked-signals}). Note that, for each vertex $u$ on the path from $\general$ to $\hat{\midpoint}$ except for $\general$, the direction towards $\general$ is added to the memory of the count signal at $u$, because a marked slowed-down find-midpoint signal whose origin is \emph{not} $u$ reaches $u$ from that direction, and the other directions are added, because marked reflected find-midpoint signals with origin $u$ reach it from those directions; for all other vertices, the latter is the reason the directions are added.
  \end{proof-sketch}

  \begin{proof}
    Let $v_1$ and $v_2$ be the neighbours of $\hat{\midpoint}$ on $\hat{p}$ such that $\distanceOf(v_1, \hat{v}_1) < \distanceOf(v_2, \hat{v}_1)$ (or, equivalently, $\distanceOf(v_1, \hat{v}_2) > \distanceOf(v_2, \hat{v}_2)$).

    First, let $\distanceOf(\hat{v}, \hat{v}_1) = \distanceOf(\hat{v}, \hat{v}_2)$ (see \cref{figure:midpoints-of-maximum-weight-paths-are-recognised:general-lies-on-path-and-midpoint-is-nearest-to-general,figure:midpoints-of-maximum-weight-paths-are-recognised:midpoint-is-nearest-to-general}). Then, $\hat{v} = \hat{\midpoint}$, and, for each index $i \in \setOf{1, 2}$, we have $v_i \neq \general$ (because, if $\general$ lies on $\hat{p}$, then $\general = \hat{v} = \hat{\midpoint} \neq v_i$, and otherwise, $\general$ does not lie on $\hat{p}$ but $v_i$ does), and, $v_i$ is either a leaf or an element of $\Vertices_{\mathfrak{g}}$ (because $\maxSubtreesOf(v_i)$ consists of the tree of $\subtreesOf(v_i)$ that contains $\hat{v}_j$, where $j \in \setOf{1, 2} \smallsetminus \setOf{i}$, and this tree, namely $\hat{\Tree}_{v_i}$, contains $\general$). Hence, according to \cref{lemma:when-do-leaves-sent-marked-signals} or \cref{lemma:when-do-non-leaf-vertices-whose-unique-maximum-radius-tree-contains-the-general-sent-marked-signals}, for each index $i \in \setOf{1, 2}$, at time $\distanceOf(\general, v_i) + 2 \cdot \distanceOf(v_i, \hat{v}_i)$, a marked reflected find-midpoint signal with origin $\hat{\midpoint}$ and reflection vertex $\hat{v}_i$ is sent from $v_i$ towards $\hat{\midpoint}$, and at time $\distanceOf(\general, \hat{\midpoint}) + 2 \cdot \distanceOf(\hat{\midpoint}, \hat{v}_i) = r + d/2$, it reaches $\hat{\midpoint}$, where it collides with the other signal.
    \begin{figure}
      \myfloatalign
      \begin{wide}
        \figureMidpointsOfMaximumWeightPathsAreRecognised
        \caption{Schematic representation of the set-up of the proof of \cref{theorem:midpoints-of-maximum-weight-paths-are-recognised}. Vertices are depicted as dots, points that may or may not be vertices are depicted as circles, paths that may consist of more than one edge are drawn dotted, and paths that consist of precisely one edge are drawn solid.}
        \label{figure:midpoints-of-maximum-weight-paths-are-recognised}
      \end{wide}
    \end{figure}

    Secondly, let $\distanceOf(\hat{v}, \hat{v}_1) \neq \distanceOf(\hat{v}, \hat{v}_2)$ (see \cref{figure:midpoints-of-maximum-weight-paths-are-recognised:general-lies-on-path-and-midpoint-is-not-nearest-to-general,figure:midpoints-of-maximum-weight-paths-are-recognised:midpoint-is-not-nearest-to-general}). Furthermore, let $X_{\general}$ be the union of the set of leaves and the set $V_{\general}$. Then, $\hat{v} \neq \hat{\midpoint}$, and either $v_1 = \general = \hat{v} \in X_{\general}$ or $v_1 \in U_{\general}$, and $v_2 \in X_{\general}$ (because otherwise, there would be a direction-preserving path with more weight than $\hat{p}$, which contradicts that $\hat{p}$ has maximum-weight).

    According to \cref{lemma:when-do-leaves-sent-marked-signals} or \cref{lemma:when-do-non-leaf-vertices-whose-unique-maximum-radius-tree-contains-the-general-sent-marked-signals}, at time $\distanceOf(\general, v_2) + 2 \cdot \distanceOf(v_2, \hat{v}_2)$, a marked reflected find-midpoint signal with origin $\hat{v}$ and reflection vertex $\hat{v}_2$ is sent from $v_2$ towards $\hat{\midpoint}$, and at time $\distanceOf(\general, v_2) + 2 \cdot \distanceOf(v_2, \hat{v}_2) + \distanceOf(v_2, \hat{\midpoint}) = r + d/2$ it reaches $\hat{\midpoint}$.

    If $v_1 = \general = \hat{v} \in X_{\general}$, then, according to \cref{lemma:when-do-leaves-sent-marked-signals} or \cref{lemma:when-do-non-leaf-vertices-whose-unique-maximum-radius-tree-contains-the-general-sent-marked-signals}, at time $2 \cdot \distanceOf(\hat{v}, \hat{v}_1)$ (which is $0$ in the case that $v_1$ is a leaf, because in this case $v_1 = \hat{v} = \hat{v}_1$), a marked slowed-down find-midpoint signal with origin $\hat{v}$ and reflection vertex $\hat{v}_1$ is sent from $\hat{v}$ towards $\hat{\midpoint}$, and at time $2 \cdot \distanceOf(\hat{v}, \hat{v}_1) + \distanceOf(\hat{v}, \hat{\midpoint}) = r + d/2$, it reaches $\hat{\midpoint}$.

    Otherwise, we have $\hat{v} \in W_{\general, v_1} \smallsetminus \setOf{v_1}$ (because, for each vertex $v \in \Vertices_{\hat{v}, v_1} \smallsetminus \setOf{\hat{v}}$, the set $\maxSubtreesOf(v)$ consists of the tree of $\subtreesOf(v)$ that contains $\hat{v}_2$, the tree $\Tree_v^{\general}$ contains $\hat{v}_1$, and due to the maximality of $\hat{p}$, except for $\hat{\Tree}_v$, there can be no tree of $\subtreesOf(v)$ with a greater radius than $\Tree_v^{\general}$), and hence, according to \cref{lemma:when-do-vertices-whose-unique-maximum-radius-tree-does-not-contain-the-general-sent-marked-signals}, at time $\distanceOf(\general, v_1) + 2 \cdot \distanceOf(v_1, \hat{v}_1)$, a marked slowed-down find-midpoint signal with origin $\hat{v}$ and reflection vertex $\hat{v}_1$ is sent from $v_1$ towards $\hat{\midpoint}$, and at time $\distanceOf(\general, v_1) + 2 \cdot \distanceOf(v_1, \hat{v}_1) + \distanceOf(v_1, \hat{\midpoint}) = r + d/2$, it reaches $\hat{\midpoint}$.

    In either case, at time $r + d/2$, at the midpoint $\hat{\midpoint}$, a marked slowed-down find-midpoint signal with origin $\hat{v}$ and reflection vertex $\hat{v}_1$ collides with a marked reflected find-midpoint signal with origin $\hat{v}$ and reflection vertex $\hat{v}_2$.
  \end{proof}

  \begin{remark} 
    On the other hand, for each direction-preserving path that does \emph{not} have maximum-weight, whenever its midpoint is found, one of the signals is unmarked and hence the midpoint is not falsely thought to be one of a maximum-weight path. The interested reader may prove that for herself.
  \end{remark}

  \subsection{Thaw Signals Traverse Midpoints and Thaw Synchronisation of Edges Just in Time}
  \label{subsection:thaw-signals-traverse-midpoints-and-thaw-synchronisation-of-edges-just-in-time}

  The inverse of a path is its traversal from target to source as given in

  \begin{definition}
    Let $p = (v_0, e_1, v_1, \dotsc, e_n, v_n)$ be a path in $\Graph$. The path $\invert(p) = (v_n, e_n, \dotsc, v_1, e_1, v_0)$ is called \define{inverse of $p$}\graffito{inverse $\invert(p)$ of $p$}\index[symbols]{invp@$\invert(p)$}. 
  \end{definition}

  \begin{remark}
    The inverse of an empty path is the empty path itself.
  \end{remark}

  \begin{remark}
    The weight and the midpoint of the inverse of a path is the same as the weight and the midpoint of the path itself.
  \end{remark}

  An undirected path does not know which of its ends is its source and which is its target and it can be represented as in

  \begin{definition}
    Let $\leftrightarrow$\graffito{equivalence relation $\leftrightarrow$ on $\Paths$}\index[symbols]{arrowleftright@$\leftrightarrow$} be the equivalence relation on $\Paths$ given by $p \leftrightarrow \invert(p)$. Each equivalence class $\equivalenceClassOf{p}_\leftrightarrow \in \Paths \modulo {\leftrightarrow}$ is called \graffito{undirected path $\equivalenceClassOf{p}_\leftrightarrow$}\define{undirected path}\index{path!undirected} and the non-negative real number $\weightOf(\equivalenceClassOf{p}_\leftrightarrow) = \weightOf(p)$ is called \define{weight of $\equivalenceClassOf{p}_\leftrightarrow$}\graffito{weight $\weightOf(\equivalenceClassOf{p}_\leftrightarrow)$ of $\equivalenceClassOf{p}_\leftrightarrow$}\index[symbols]{omegapequivalenceclass@$\weightOf(\equivalenceClassOf{p}_\leftrightarrow)$}.
  \end{definition}

  \begin{remark}
    The equivalence class of an empty path is the singleton set that consists of the empty path.
  \end{remark}

  For each path $p$, the set(s) of paths with the same source (or target), less weight, and midpoints on the continuum representation of $p$ are named in

  \begin{definition}
    For each direction-preserving path $p \in \Paths_{\directionPreserving}$, let
    \begin{equation*}
      \Sigma_p = \setOf{p' \in \Paths_{\directionPreserving} \suchThat
                        \sourceOf(p') = \sourceOf(p) \text{, }
                        \weightOf(p') < \weightOf(p) \text{, and }
                        \midpoint_{p'} \in \imageOf\continuumRepresentationOf{p}} \mathnote{set $\Sigma_p$}\index[symbols]{Sigmap@$\Sigma_p$}
    \end{equation*}
    and let
    \begin{equation*}
      T_p = \setOf{p' \in \Paths_{\directionPreserving} \suchThat
                   \targetOf(p') = \targetOf(p) \text{, }
                   \weightOf(p') < \weightOf(p) \text{, and }
                   \midpoint_{p'} \in \imageOf\continuumRepresentationOf{p}}. \mathnote{set $T_p$}\index[symbols]{Tp@$T_p$} \qedhere
    \end{equation*}
  \end{definition}

  \begin{remark}
    Note that $T_p = \invert(\Sigma_{\invert(p)})$.
  \end{remark}

  \begin{remark}
    Let $p$ be a direction-preserving path of $\Graph$. For each path $p' \in \Sigma_p \cup T_p$, we have $\distanceOf(\midpoint_p, \midpoint_{p'}) = \weightOf(p)/2 - \weightOf(p')/2$.
  \end{remark}

  The set of undirected direction-preserving paths can be turned into a graph by having an edge from an undirected direction-preserving path to each such path that has one end in common with the path, has less weight than the path, and has the greatest weight among the paths with the former two properties. The edges can be weighted by the distance of the midpoints of the undirected paths. The graph and its edge-weighting are introduced in

  \begin{definition}
    Let $\Vertices_{\thawKind} = \Paths_{\directionPreserving} \modulo {\leftrightarrow}$, let
    \begin{equation*}
      \Edges_{\thawKind} = \setOf{(\equivalenceClassOf{p}_\leftrightarrow, \equivalenceClassOf{p'}_\leftrightarrow) \in \Vertices_{\thawKind} \times \Vertices_{\thawKind} \suchThat p' \in \argmax_{p'' \in \Sigma_p \cup T_p} \weightOf(p'')},
    \end{equation*} 
    and let
    \begin{align*}
      \weightOf_{\thawKind} \from \Edges_{\thawKind} &\to \R_{\geq 0},\\
      (\equivalenceClassOf{p}_\leftrightarrow, \equivalenceClassOf{p'}_\leftrightarrow) &\mapsto \distanceOf(\midpoint_p, \midpoint_{p'}).
    \end{align*}
    The triple $\Graph_{\thawKind} = \ntuple{\Vertices_{\thawKind}, \Edges_{\thawKind}, \weightOf_{\thawKind}}$\graffito{edge-weighted directed acyclic graph $\Graph_{\thawKind}$}\index[symbols]{GTcalligraphictypewriter@$\Graph_{\thawKind}$} is an edge-weighted directed acyclic graph.
  \end{definition}

  \begin{remark}
    For each maximum-weight direction-preserving path $\hat{p}$ in $\Graph$, the in-degree of $\equivalenceClassOf{\hat{p}}_\leftrightarrow$ in $\Graph_{\thawKind}$ is $0$, because there are only edges to equivalence classes of less-weight paths.
  \end{remark}

  \begin{remark}
  \label{remark:weight-of-edge-is-half-weight-of-longer-path-minus-half-weight-of-shorter-path}
    For each edge $(\equivalenceClassOf{p}_\leftrightarrow, \equivalenceClassOf{p'}_\leftrightarrow)$ of $\Graph_{\thawKind}$, we have $\weightOf_{\thawKind}(\equivalenceClassOf{p}_\leftrightarrow,\allowbreak \equivalenceClassOf{p'}_\leftrightarrow) = \weightOf(p)/2 - \weightOf(p')/2$.
  \end{remark}

  \begin{remark}
    There is a bijection between the vertices of $\Vertices_{\thawKind}$ and the set of all midpoint signals that are created by the signal machine $\mathcal{S}$. The reason is that each midpoint signal memorises two words of directions in a set, one that leads from its position to one end of the path it designates the midpoint of and the other to the other end; because these words are stored in a set, the midpoint signal does not differentiate between source and target of its path.

    Each maximum-weight vertex of $\Graph_{\thawKind}$ is a maximum-weight undirected direction-preserving path in $\Graph$ and is, under suitable identifications and definitions, a longest undirected direction-preserving path in the continuum representation $M$ of $\Graph$. And, each minimum-weight vertex of $\Graph_{\thawKind}$ is one of weight $0$, is an undirected empty path in $\Graph$, and is, under suitable identifications, an empty path in $\Graph$ and in $M$, and a vertex in $\Graph$ and in $M$.

    For each path $p_{\thawKind}$ in $\Graph_{\thawKind}$ that starts in a maximum-weight vertex and ends in a minimum-weight vertex, in the time evolution of the signal machine $\mathcal{S}$, there is a thaw signal that traverses the midpoint signals represented by the vertices on the path in the order they occur on the path such that the time the thaw signal takes to get from the vertex $v_{\thawKind}$ on the path to the next vertex $v_{\thawKind}'$ on the path is precisely the weight of the edge $(v_{\thawKind}, v_{\thawKind}')$, in particular, the time the thaw signal takes to reach the end of its path is the weight of $p_{\thawKind}$.

    To show that the synchronisation of all edges is thawed and finishes at the same time, we have to show that, for each edge, there is a thaw signal that collides with the midpoint signal of the edge and reaches the end of its path at one end of the edge, and that all thaw signals reach the ends of their paths at the same time. The former is equivalent to showing that, for each edge $e \in \Edges$ with the two ends $v_0$ and $v_1$, there is a maximum-weight direction-preserving path $p$ in $\Graph$ such that there is a path in $\Graph_{\thawKind}$ from $\equivalenceClassOf{p}_{\leftrightarrow}$ to $\equivalenceClassOf{(v_0, e, v_1)}_{\leftrightarrow}$. And the latter is equivalent to showing that the weights of all paths from maximum-weight vertices to minimum-weight vertices in $\Graph_{\thawKind}$ are the same. See \cref{corollary:the-weight-of-paths-to-edges-and-vertices-in-the-thaw-graph}.

    For each non-empty direction-preserving path $p$ in $\Graph$, the midpoint of $p$ is found at time $\max\setOf{\distanceOf(\general, \sourceOf(p)),\allowbreak \distanceOf(\general, \targetOf(p))} + \weightOf(p)/2$ (see \cref{lemma:the-time-at-which-the-midpoint-of-a-path-is-found}) and, for each maximum-weight direction-preserving path $\hat{p}$ in $\Graph$, the midpoint of $\hat{p}$ is found at time $r + d/2$ (see \cref{lemma:the-time-at-which-the-midpoints-of-longest-paths-are-found}). If the thaw signals that emanate from the midpoints of maximum-weight direction-preserving paths at the time $r + d/2$ reach the ends of their paths after the time span $d/2$, then synchronisation finishes at the optimal time $r + d$. The condition is equivalent to showing that the weights of all paths from maximum-weight vertices to minimum-weight vertices in $\Graph_{\thawKind}$ are equal to $d/2$. See \cref{corollary:the-weight-of-paths-to-edges-and-vertices-in-the-thaw-graph}.
  \end{remark}

  The thaw signals that spread from the midpoint signal of a path $p$ eventually collide with the midpoint signals of all paths that have the same source or target as $p$, less weight, and whose midpoints lie on $p$. This is what is essentially shown in

  \begin{lemma}
  \label{lemma:existence-of-path-in-midpoint-signal-graph}
    Let $p$ be a path in $\Graph$ and let $p' \in \Sigma_p \cup T_p$. There is a path from $\equivalenceClassOf{p}_\leftrightarrow$ to $\equivalenceClassOf{p'}_\leftrightarrow$ in $\Graph_{\thawKind}$.
  \end{lemma}

  \begin{proof}
    \begin{description}
      \item[Case 1:] $p' \in \Sigma_p$. If $p' \in \argmax_{p'' \in \Sigma_p} \weightOf(p'')$, then $(\equivalenceClassOf{p}_\leftrightarrow,\allowbreak \equivalenceClassOf{p'}_\leftrightarrow) \in \Edges_{\thawKind}$ and the path consisting of this edge is one from $\equivalenceClassOf{p}_\leftrightarrow$ to $\equivalenceClassOf{p'}_\leftrightarrow$. Otherwise, there is a $q \in \argmax_{p'' \in \Sigma_p} \weightOf(p'')$ such that $\weightOf(p') < \weightOf(q)$. Then, $(\equivalenceClassOf{p}_\leftrightarrow, \equivalenceClassOf{q}_\leftrightarrow) \in \Edges_{\thawKind}$. And, because $\setOf{p', q} \subseteq \Sigma_p$, we have $\sourceOf(p') = \sourceOf(p) = \sourceOf(q)$. Thus, because $\continuumRepresentationOf{q}(\weightOf(q) / 2) = \midpoint_q \in \imageOf \continuumRepresentationOf{p}$, the paths $q$ and $p$ are direction-preserving, and the multigraph $\Graph$ is a tree, we have $\continuumRepresentationOf{q}\restrictedTo_{\closedInterval{0, \weightOf(q) / 2}} = \continuumRepresentationOf{p}\restrictedTo_{\closedInterval{0, \weightOf(q) / 2}}$ and, analogously, we have $\continuumRepresentationOf{p'}\restrictedTo_{\closedInterval{0, \weightOf(p') / 2}} = \continuumRepresentationOf{p}\restrictedTo_{\closedInterval{0, \weightOf(p') / 2}}$. Hence, because $\weightOf(p') < \weightOf(q)$,
            \begin{equation*}
              \midpoint_{p'}
              = \continuumRepresentationOf{p'}(\weightOf(p') / 2)
              = \continuumRepresentationOf{p}(\weightOf(p') / 2)
              = \continuumRepresentationOf{q}(\weightOf(p') / 2)
              \in \imageOf\continuumRepresentationOf{q}.
            \end{equation*}
            Therefore, $p' \in \Sigma_q$. Now, if $p' \in \argmax_{p'' \in \Sigma_q} \weightOf(p'')$, then $(\equivalenceClassOf{q}_\leftrightarrow,\allowbreak \equivalenceClassOf{p'}_\leftrightarrow) \in \Edges_{\thawKind}$, and the path consisting of the edges $(\equivalenceClassOf{p}_\leftrightarrow, \equivalenceClassOf{q}_\leftrightarrow)$ and $(\equivalenceClassOf{q}_\leftrightarrow, \equivalenceClassOf{p'}_\leftrightarrow)$ is a path from $\equivalenceClassOf{p}_\leftrightarrow$ to $\equivalenceClassOf{p'}_\leftrightarrow$. Otherwise, because $\Vertices_{\thawKind}$ is finite and $\weightOf(q) > \weightOf(p')$, it follows by induction that there is a path from $\equivalenceClassOf{q}_\leftrightarrow$ to $\equivalenceClassOf{p'}_\leftrightarrow$ and therefore one from $\equivalenceClassOf{p}_\leftrightarrow$ to $\equivalenceClassOf{p'}_\leftrightarrow$. 
      \item[Case 2:] $p' \in T_p$. Then, $\invert(p') \in \Sigma_{\invert(p)}$. Hence, according to the first case, there is a path from $\equivalenceClassOf{\invert(p)}_\leftrightarrow$ to $\equivalenceClassOf{\invert(p')}_\leftrightarrow$. Therefore, because $\equivalenceClassOf{p}_\leftrightarrow = \equivalenceClassOf{\invert(p)}_\leftrightarrow$ and $\equivalenceClassOf{p'}_\leftrightarrow = \equivalenceClassOf{\invert(p')}_\leftrightarrow$, there is a path from $\equivalenceClassOf{p}_\leftrightarrow$ to $\equivalenceClassOf{p'}_\leftrightarrow$. \qedhere
    \end{description}
  \end{proof}

  Each midpoint signal of a path eventually collides with a matching thaw signal that originated at the midpoint of a maximum-weight direction-preserving path. This is what is essentially shown in

  \begin{lemma}
  \label{lemma:existence-of-paths-in-the-thaw-graph}
    Let $p$ be a direction-preserving path in $\Graph$. There is a maximum-weight direction-preserving path $\hat{p}$ in $\Graph$ such that there is a path from $\equivalenceClassOf{\hat{p}}_\leftrightarrow$ to $\equivalenceClassOf{p}_\leftrightarrow$ in $\Graph_{\thawKind}$.
  \end{lemma}

  \begin{proof}
    If $p$ is a maximum-weight path, then the path $\hat{p} = p$ in $\Graph$ and the empty path $(\equivalenceClassOf{p}_\leftrightarrow)$ in $\Graph_{\thawKind}$ have the desired properties. From now on, let $p$ not be a maximum-weight path. Furthermore, let $\hat{p}$ be a maximum-weight path in $\Graph$ and let $p_d$ be the minimum-weight path in $\Graph$ such that $\sourceOf(p_d)$ lies on $p$ and $\targetOf(p_d)$ lies on $\hat{p}$. Then, there are paths $p_1$ and $p_2$ in $\Graph$ such that $p_1 \concat p_2 = p$ and $\targetOf(p_1) = \sourceOf(p_d)$ as well as $\targetOf(\invert(p_2)) = \sourceOf(p_d)$. And, there are paths $\hat{p}_1$ and $\hat{p}_2$ in $\Graph$ such that $\hat{p}_1 \concat \hat{p}_2 = \hat{p}$ and $\targetOf(p_d) = \sourceOf(\hat{p}_2)$ as well as $\targetOf(p_d) = \sourceOf(\invert(\hat{p}_1))$. Let
    \begin{equation*}
      q_1 = \begin{dcases*} 
              p_1, &if $\weightOf(p_1) \geq \weightOf(p_2)$,\\
              \invert(p_2), &otherwise,
            \end{dcases*}
    \end{equation*}
    let
    \begin{equation*}
      q_2 = \begin{dcases*} 
              p_2, &if $\weightOf(p_1) \geq \weightOf(p_2)$,\\
              \invert(p_1), &otherwise,
            \end{dcases*}
    \end{equation*}
    and let $q = q_1 \concat q_2$. Furthermore, let
    \begin{equation*}
      \hat{q}_1 = \begin{dcases*}
                    \hat{p}_1, &if $\weightOf(\hat{p}_2) \geq \weightOf(\hat{p}_1)$,\\
                    \invert(\hat{p}_2), &otherwise,
                  \end{dcases*}
    \end{equation*}
    let
    \begin{equation*}
      \hat{q}_2 = \begin{dcases*}
                    \hat{p}_2, &if $\weightOf(\hat{p}_2) \geq \weightOf(\hat{p}_1)$,\\
                    \invert(\hat{p}_1), &otherwise,
                  \end{dcases*}
    \end{equation*}
    and let $\hat{q} = \hat{q}_1 \concat \hat{q}_2$. Moreover, let $p' = q_1 \concat p_d \concat \hat{q}_2$. Then, $q \in \setOf{p, \invert(p)}$ as well as $\hat{q} \in \setOf{\hat{p}, \invert(\hat{p})}$. And, $\weightOf(q_1) \geq \weightOf(q_2)$ as well as $\weightOf(\hat{q}_1) \leq \weightOf(\hat{q}_2)$. And, $\targetOf(p') = \targetOf(\hat{q})$ as well as $\sourceOf(q) = \sourceOf(p')$. And, because $\weightOf(q_1) \geq \weightOf(q_2)$, we have $\midpoint_q \in \imageOf\continuumRepresentationOf{q_1} \subseteq \imageOf\continuumRepresentationOf{p'}$.
    \begin{description}
      \item[Case 1:] $p'$ is a maximum-weight path. Then, because $p$ is not a maximum-weight path, we have $\weightOf(q) = \weightOf(p) < \weightOf(p')$. Hence, because $\sourceOf(q) = \sourceOf(p')$ and $\midpoint_q \in \imageOf\continuumRepresentationOf{p'}$, we have $q \in \Sigma_{p'}$. In conclusion, according to \cref{lemma:existence-of-path-in-midpoint-signal-graph}, there is a path from $\equivalenceClassOf{p'}_\leftrightarrow$ to $\equivalenceClassOf{p}_\leftrightarrow = \equivalenceClassOf{q}_\leftrightarrow$.
      \item[Case 2:] $p'$ is not a maximum-weight path. Then, $\weightOf(q_1 \concat p_d \concat \hat{q}_2) = \weightOf(p') < \weightOf(\hat{q}) = \weightOf(\hat{q}_1 \concat \hat{q}_2)$ and thus $\weightOf(q_1 \concat p_d) < \weightOf(\hat{q}_1)$, in particular, $\weightOf(q_1) < \weightOf(\hat{q}_1)$. Hence, because $\weightOf(q_2) \leq \weightOf(q_1) < \weightOf(\hat{q}_1) \leq \weightOf(\hat{q}_2)$, we have $\weightOf(q) = \weightOf(q_1) + \weightOf(q_2) < \weightOf(q_1) + \weightOf(\hat{q}_2) \leq \weightOf(p')$. And, because $\weightOf(q_1 \concat q_d) < \weightOf(\hat{q}_1) \leq \weightOf(\hat{q}_2)$, we have $\midpoint_{p'} \in \imageOf\continuumRepresentationOf{\hat{q}_2} \subseteq \imageOf\continuumRepresentationOf{\hat{q}}$. And, recall that $\midpoint_{q} \in \imageOf\continuumRepresentationOf{p'}$. Altogether, because $\targetOf(p') = \targetOf(\hat{q})$ and $\sourceOf(q) = \sourceOf(p')$, we have $p' \in T_{\hat{q}}$ and $q \in \Sigma_{p'}$. Therefore, according to \cref{lemma:existence-of-path-in-midpoint-signal-graph}, there is a path from $\equivalenceClassOf{\hat{p}}_\leftrightarrow = \equivalenceClassOf{\hat{q}}_\leftrightarrow$ to $\equivalenceClassOf{p'}_\leftrightarrow$ and there is a path from $\equivalenceClassOf{p'}_\leftrightarrow$ to $\equivalenceClassOf{q}_\leftrightarrow = \equivalenceClassOf{p}_\leftrightarrow$. In conclusion, there is a path from $\equivalenceClassOf{\hat{p}}_\leftrightarrow$ to $\equivalenceClassOf{p}_\leftrightarrow$. \qedhere
    \end{description}
  \end{proof}

  The time it takes a thaw signal from the midpoint of a path to collide with the midpoint signal of a matching path is given by half the former path's length minus half the latter path's length. This is what is essentially shown in

  \begin{lemma}
  \label{lemma:the-weight-of-paths-in-the-thaw-graph}
    For each path $p_{\thawKind}$ in $\Graph_{\thawKind}$, the weight of $p_{\thawKind}$ is equal to $\weightOf(\sourceOf(p_{\thawKind}))/2 - \weightOf(\targetOf(p_{\thawKind}))/2$. 
  \end{lemma}

  \begin{proof}
    We prove the statement by induction. 

    \proofPart{Base Case}
      Each empty path $p_{\thawKind}$ in $\Graph_{\thawKind}$ has weight $0$, has the same source and target vertices, and $\weightOf(\sourceOf(p_{\thawKind}))/2 - \weightOf(\targetOf(p_{\thawKind}))/2$ is equal to $0$ as needed.

    \proofPart{Inductive Step}
      Let $p_{\thawKind} = (\equivalenceClassOf{p_0}_\leftrightarrow, \equivalenceClassOf{p_1}_\leftrightarrow, \dotsc, \equivalenceClassOf{p_n}_\leftrightarrow)$ be a non-empty path in $\Graph_{\thawKind}$ such that the weight of the subpath $(\equivalenceClassOf{p_1}_\leftrightarrow, \dotsc, \equivalenceClassOf{p_n}_\leftrightarrow)$ is equal to $\weightOf(\equivalenceClassOf{p_1}_\leftrightarrow)/2 - \weightOf(\equivalenceClassOf{p_n}_\leftrightarrow)/2$. Then, according to \cref{remark:weight-of-edge-is-half-weight-of-longer-path-minus-half-weight-of-shorter-path}, we have $\weightOf_{\thawKind}(\equivalenceClassOf{p_0}_\leftrightarrow, \equivalenceClassOf{p_1}_\leftrightarrow) = \distanceOf(\midpoint_{p_0}, \midpoint_{p_1}) = \weightOf(p_0)/2 - \weightOf(p_1)/2 = \weightOf(\equivalenceClassOf{p_0}_\leftrightarrow)/2 - \weightOf(\equivalenceClassOf{p_1}_\leftrightarrow)/2$. Hence, the weight of the path $p_{\thawKind}$ is equal to $\weightOf(\equivalenceClassOf{p_0}_\leftrightarrow)/2 - \weightOf(\equivalenceClassOf{p_1}_\leftrightarrow)/2 + \weightOf(\equivalenceClassOf{p_1}_\leftrightarrow)/2 - \weightOf(\equivalenceClassOf{p_n}_\leftrightarrow)/2 = \weightOf(\equivalenceClassOf{p_0}_\leftrightarrow)/2 - \weightOf(\equivalenceClassOf{p_n}_\leftrightarrow)/2$.
  \end{proof}

  The time it takes a thaw signal from a maximum-weight direction-preserving path to collide with the midpoint signal of a matching path is essentially given in

  \begin{corollary}
  \label{corollary:the-weight-of-paths-in-the-thaw-graph}
    For each maximum-weight direction-preserving path $\hat{p}$ in $\Graph$ such that there is a path from $\equivalenceClassOf{\hat{p}}_\leftrightarrow$ to $\equivalenceClassOf{p}_\leftrightarrow$ in $\Graph_{\thawKind}$, the weight of this path is $d/2 - \weightOf(p)/2$.
  \end{corollary}

  \begin{proof}
    This is a direct consequence of \cref{lemma:the-weight-of-paths-in-the-thaw-graph} with the fact that $\weightOf(\hat{p}) = d$.
  \end{proof}

  The time it takes a thaw signal from a maximum-weight direction-preserving path to collide with the midpoint signal of an edge it thaws and to reach the ends of the edge is essentially given in

  \begin{corollary}
  \label{corollary:the-weight-of-paths-to-edges-and-vertices-in-the-thaw-graph}
    Let $e$ be an edge of $\Graph$, and let $v_0$ and $v_1$ be the two ends of $e$. There is a maximum-weight direction-preserving path $\hat{p}$ such that there is a path from $\equivalenceClassOf{\hat{p}}_\leftrightarrow$ to $\equivalenceClassOf{(v_0, e, v_1)}_\leftrightarrow$ and all such paths have weight $d / 2 - \weightOf(e) / 2$, and there are also paths from $\equivalenceClassOf{\hat{p}}_\leftrightarrow$ to $\equivalenceClassOf{(v_0)}_\leftrightarrow$ and to $\equivalenceClassOf{(v_1)}_\leftrightarrow$ and all such paths have weight $d/2$.
  \end{corollary}

  \begin{proof}
    This is a direct consequence of \cref{lemma:existence-of-paths-in-the-thaw-graph} and corollary \ref{corollary:the-weight-of-paths-in-the-thaw-graph}. 
  \end{proof}

  \addtocontents{toc}{\protect\vspace{\beforebibskip}}
  \addcontentsline{toc}{section}{\refname}
  \printbibliography
\end{document}